\providecommand{\review}{0}
\providecommand{\blind}{0}
\providecommand{\material}{1}

\ifnum\review=0
	\documentclass[a4paper, twoside, reqno]{amsart}
\else
	\documentclass[a4paper, 12pt, twoside, reqno]{amsart}
\fi

\usepackage{amsmath}
\usepackage{amssymb}
\usepackage[foot]{amsaddr}
\usepackage{upgreek}
\usepackage{accents}
\usepackage{mathtools}
\usepackage{multicol}
\usepackage[
	backend=biber,
	citestyle=numeric-comp,
	bibstyle=authoryear,
	dashed=false,
	firstinits=true,
	maxcitenames=3,
	maxbibnames=6,
	doi=false,
	isbn=false,
	date=year,
	defernumbers=true]{biblatex}
\usepackage[
	hidelinks,
	pdfproducer={Latex with hyperref},
	pdfcreator={pdflatex}]{hyperref}
\usepackage[acronym]{glossaries}
\usepackage{glossaries-prefix}
\usepackage{stmaryrd}
\usepackage{enumitem}
\usepackage{pifont}
\usepackage{marvosym}
\usepackage{graphicx}
\usepackage[font=footnotesize]{caption}
\usepackage{subcaption}
\usepackage{authoraftertitle}
\usepackage{textcase}

\input{01-2-definitions.tex}

\renewcommand{\figurename}{Fig.}

\makeatletter
\def\cleardoublepage{\clearpage%
	\if@twoside
		\ifodd\c@page\else
			\vspace*{\fill}
			\hfill
			\begin{center}
				This page intentionally left blank.
			\end{center}
			\vspace{\fill}
			\thispagestyle{empty}
			\newpage
			\if@twocolumn\hbox{}\newpage\fi
		\fi
	\fi
}
\makeatother

\ifnum\review=1
	\usepackage{setspace}
	\calclayout
	\doublespace
\fi

\makeglossaries

\newacronym{sm}{SM}{supplementary material}

\newacronym{as}{a.s.}{almost surely}
\newacronym{iid}{i.i.d.}{independent and identically distributed}
\newacronym{iff}{iff}{if and only if}

\newacronym[
    plural=pdfs,
    firstplural=probability density functions]{pdf}{pdf}
{probability density function}

\newacronym[
    plural=cdfs,
    firstplural=cumulative distribution functions]{cdf}{cdf}
{cumulative distribution function}

\newacronym[
    prefixfirst={a\ },
    prefix={an\ }]{rv}{rv}
{random variable}

\newacronym{kde}{KDE}{\textit{kernel density estimator}}
\newacronym{kdde}{KDDE}{\textit{kernel density derivative estimation}}

\newacronym{eda}{EDA}{\textit{exploratory data analysis}}

\newacronym{bh}{BH}{\textit{bump hunting}}

\newacronym[
    plural=HDRs,
    firstplural=highest density regions]{hdr}{HDR}
{\textit{highest density region}}

\newacronym{sizer}{SiZer}{\textit{SIgnificant ZERo crossings of derivatives}}

\newacronym{vc}{VC}{Vapnik–Chervonenkis}
\newacronym{pm}{PM}{pointwise measurable}

\newacronym[
    plural=GPs,
    firstplural=Gaussian processes]{gp}{GP}
{\textit{Gaussian process}}

\newacronym{aui}{a.u.i.}{\textit{asymptotically uniformly integrable}}

\newacronym{tvar}{TVaR}{\textit{Tail Value at Risk}}

\newacronym{nba}{NBA}{\textit{National Basketball Association}}
\newacronym{nfl}{NFL}{\textit{National Football League}}
\newacronym{mlb}{MLB}{\textit{Major League Baseball}}

\newacronym{3pl}{3PL}{\textit{three-point line}}
\newacronym{mph}{mph}{\textit{miles per hour}}

\newcommand{\orcid}[1]{%
	\href{https://orcid.org/#1}{%
		\includegraphics[height=2ex]{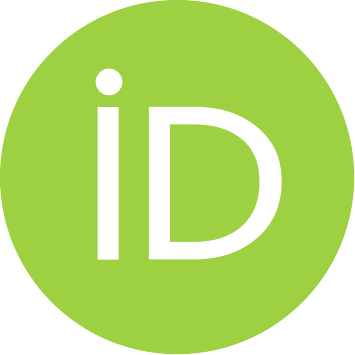}%
		\url{https://orcid.org/#1}
	}%
}

\newcommand{\supplement}{%
	\gls*{sm}%
	\ifnum\blind=0%
		~\cite{Supplement}%
	\fi%
}

\newcommand{\sourcecode}{%
	\ifnum\blind=0%
		~\cite{SourceCode}%
	\fi%
}

\newcommand{\subfigurewidthtwoandone}{0.4\textwidth}
\newcommand{\graphicswidthtwoandone}{\textwidth}

\newcommand{\hyndalpha}{\alpha}
\newcommand{\hyndfalpha}{\pdf_{\hyndalpha}}
\newcommand{\hyndhdr}{\{\x: \pdf(\x) \geq \hyndfalpha\}}
\newcommand{\hyndhdralpha}{R(\hyndfalpha)}

\newcommand{\ubar}[1]{\underaccent{\bar}{#1}}

\newcommand{\approxobject}[1]{\tilde{#1}}
\newcommand{\normalizedobject}[1]{\bar{#1}}

\newcommand{\eqdef}{\coloneqq}

\newcommand{\genindexi}{i}
\newcommand{\genindexj}{j}
\newcommand{\genindexk}{k}
\newcommand{\genindexn}{n}

\newcommand{\sumitogen}[1]{\sum_{\genindexi = 1}^{#1}}
\newcommand{\sumitosamplesize}{\sumitogen{\samplesize}}
\newcommand{\sumitod}{\sumitogen{\dimension}}

\newcommand{\sumijtogen}[1]{%
	\sumitogen{#1}%
	\sum_{\genindexj = 1}^{#1}%
}

\newcommand{\sumijtod}{\sumijtogen{\dimension}}
\newcommand{\sumijtotwo}{\sumijtogen{2}}

\newcommand{\constant}{C}
\newcommand{\dimension}{d}
\newcommand{\samplesize}{n}
\newcommand{\samplesizestart}{\samplesize_0}
\newcommand{\rootsamplesize}{\sqrt{\samplesize}}
\newcommand{\oneoversamplesize}{\frac{1}{\samplesize}}
\newcommand{\numderivatives}{m}
\newcommand{\extrasmoothing}{r}
\newcommand{\difforder}{\ell}
\newcommand{\maxdifforder}{\difforder}
\newcommand{\derivativeorder}{k}
\newcommand{\bandwidthtargetorder}{r}

\newcommand{\x}{\mathrm{x}}
\newcommand{\y}{\mathrm{y}}
\newcommand{\z}{\mathrm{z}}
\newcommand{\vart}{t}
\newcommand{\p}{p}
\newcommand{\q}{q}

\newcommand{\pdfvarx}{x}
\newcommand{\randomvectorx}{\mathrm{X}}
\newcommand{\randomvarx}{X}
\newcommand{\randomvary}{Y}
\newcommand{\samplevar}{\randomvectorx}
\newcommand{\nsample}{\samplevar_1, \dots, \samplevar_{\samplesize}}
\newcommand{\samplevarrealization}{\mathrm{x}}
\newcommand{\cdfarg}{t}
\newcommand{\derivativevar}{a}

\newcommand{\mean}{\boldsymbol{\mu}}
\newcommand{\covariancematrix}{\Sigma}

\newcommand{\pdf}{f}
\newcommand{\pdfsecondderivative}{\pdf''}
\newcommand{\pdfx}{\pdf(\x)}
\newcommand{\functionmap}[2]{#1 \rightarrow #2}
\newcommand{\lambdafunction}[2]{#1 \mapsto #2}
\newcommand{\functiondef}[3]{#1: \functionmap{#2}{#3}}
\newcommand{\funcarg}{\cdot}
\newcommand{\comp}{\circ}

\newcommand{\curvaturefunctional}{\phi}
\newcommand{\curvaturefunctionalof}[1]{\curvaturefunctional[#1]}
\newcommand{\curvaturesignselector}{s}
\newcommand{\curvaturesign}{(-1)^{\curvaturesignselector}}

\newcommand{\testfunction}{\varphi}
\newcommand{\testpdf}{p}

\newcommand{\positivepart}[1]{(#1)_{+}}

\newcommand{\goesto}{\rightarrow}
\newcommand{\goestoinfty}{\goesto \infty}
\newcommand{\ngoestoinfty}{\genindexn\goestoinfty}
\newcommand{\limninf}{\lim_{\ngoestoinfty}}

\newcommand{\norm}[1]{\lVert #1 \rVert}
\newcommand{\absoluteval}[1]{\left\vert #1 \right\vert}
\newcommand{\smallabsval}[1]{\vert #1 \vert}
\newcommand{\innerprod}[2]{\langle #1, #2 \rangle}

\newcommand{\supxinrd}{\sup_{\x \in \rd}}
\newcommand{\supxinbumpdomain}{\sup_{\x \in \bumpdomain}}
\newcommand{\smallsupnorm}[1]{\supxinrd \smallabsval{#1}}
\newcommand{\smallsupnormbumpdomain}[1]{\supxinbumpdomain \smallabsval{#1}}

\newcommand{\differential}{\mathrm{D}}
\newcommand{\diffoperator}{\partial}
\newcommand{\divergence}{\text{div}}
\newcommand{\rddivergence}[1]{\divergence_{\rd}\left( #1 \right)}
\newcommand{\divergenceof}[1]{\divergence(#1)}
\newcommand{\grad}[1]{\nabla#1}
\newcommand{\gradf}{\grad{\pdf}}
\newcommand{\normalizedgrad}{\normalizedobject{\nabla}}
\newcommand{\normalizedgradof}[1]{\normalizedgrad#1}
\newcommand{\normalizedgraddefinition}{\normalizedgradof{\pdf} = \grad{\pdf} / \sqrt{1 + \norm{\grad{\pdf}}^2}}
\newcommand{\hessian}[1]{\differential^{2}#1}
\newcommand{\laplacian}{\Delta}
\newcommand{\laplacianof}[1]{\laplacian#1}
\newcommand{\secondorderdifffunctional}[2]{\differential_{#1#2}}

\newcommand{\trace}[1]{\mathrm{tr}(#1)}
\newcommand{\determinant}[1]{\det({#1})}
\newcommand{\transpose}[1]{#1^\top}
\newcommand{\vectorize}{\mathrm{vec}}
\newcommand{\zerovector}{\boldsymbol{0}}

\newcommand{\graphf}{\mathcal{G}}
\newcommand{\graphfpair}{(\x, \pdfx)}
\newcommand{\setgraphf}{\{\graphfpair \in \rdpone: \x \in \rd\}}
\newcommand{\chart}{\mathcal{X}}
\newcommand{\embedding}{\mathcal{F}}
\newcommand{\graphfpoint}{p}
\newcommand{\tangent}{T}
\newcommand{\tangentspace}{\tangent_{\graphfpoint}\graphf}
\newcommand{\tangentbundle}{\tangent\graphf}
\newcommand{\basisvectorfield}[1]{\partial_{#1}}
\newcommand{\tangentvector}[2]{\basisvectorfield{#1}\vert_{#2}}
\newcommand{\ithtangentvectorp}{\tangentvector{\genindexi}{\graphfpoint}}
\newcommand{\ithdifferential}[1]{\differential_{#1}}
\newcommand{\vectorfieldx}{X}
\newcommand{\vectorfieldy}{Y}
\newcommand{\vectorfieldv}{\vec{V}}
\newcommand{\immersion}{\iota}
\newcommand{\tangentdifferential}[1]{d{#1}}
\newcommand{\immersedembedding}{\embedding^{\immersion}}
\newcommand{\dimmersionithbasis}{\tangentdifferential{\immersion}(\basisvectorfield{\genindexi})}
\newcommand{\embeddingfieldindex}[1]{\ithdifferential{#1}\immersedembedding}
\newcommand{\embeddingfield}{\embeddingfieldindex{\genindexi}}
\newcommand{\euclideanbasis}[1]{\mathbf{e}_{#1}}
\newcommand{\metric}{g}
\newcommand{\kronecker}{\delta}
\newcommand{\kroneckerdd}[2]{\kronecker_{#1#2}}
\newcommand{\kroneckerup}[2]{\kronecker^{#1}_{#2}}
\newcommand{\tangenthyperplane}[1]{\pi_{#1}}
\newcommand{\unitvec}[1]{\hat{#1}}
\newcommand{\normalvecfield}{N}
\newcommand{\downwardnormalvecfield}{\check{\normalvecfield}}
\newcommand{\christoffel}{\Gamma}
\newcommand{\christoffelwithindices}[3]{\christoffel_{#1#2}^{#3}}
\newcommand{\secondfundform}{A}
\newcommand{\shapeop}{S}
\newcommand{\chartcomp}[1]{#1 \comp \chart}
\newcommand{\ithdiffnormal}{\differential_{\genindexi}\downwardnormalvecfield}
\newcommand{\extrinsicdivergence}[1]{\divergence_{\graphf}^{\lightning}(#1)}
\newcommand{\intrinsicdivergence}[1]{\divergence_{\graphf}(#1)}
\newcommand{\normalcurvature}{\kappa}
\newcommand{\principalcurvature}{\normalcurvature}
\newcommand{\meancurvature}{\mathfrak{D}}
\newcommand{\gaussiancurvature}{\mathfrak{B}}
\newcommand{\connection}{\nabla^{\graphf}}
\newcommand{\metricmatrix}{G}
\newcommand{\inversemetricmatrix}{\metricmatrix^{-1}}
\newcommand{\shapeopmatrix}{S}
\newcommand{\secondfundformmatrix}{A}
\newcommand{\hessianmatrix}{\hessian{\pdf}}

\newcommand{\boundary}{\partial}

\newcommand{\naturals}{\mathbb{N}}
\newcommand{\integers}{\mathbb{Z}}
\newcommand{\positiveintegers}{\integers_{+}}
\newcommand{\extendedpositiveintegers}{\positiveintegers \cup \{\infty\}}

\newcommand{\closedzeroinfinterval}{[0, \infty)}
\newcommand{\openzeroinfinterval}{(0, \infty)}

\newcommand{\reals}{\mathbb{R}}
\newcommand{\rtwo}{\reals^2}
\newcommand{\rd}{\reals^{\dimension}}
\newcommand{\rdpone}{\reals^{\dimension + 1}}

\newcommand{\zerolevelset}[2]{%
	\{ \x \in #1 : #2(\x) = 0 \}%
}

\newcommand{\zeroupperlevelset}[2]{%
	\{ \x \in #1 : #2(\x) \geq 0 \}%
}

\newcommand{\zerolowerlevelset}[2]{%
	\{ \x \in #1 : #2(\x) \leq 0 \}%
}

\newcommand{\rdzerolevelset}[1]{\zerolevelset{\rd}{#1}}
\newcommand{\rdzeroupperlevelset}[1]{\zeroupperlevelset{\rd}{#1}}
\newcommand{\rdzerolowerlevelset}[1]{\zerolowerlevelset{\rd}{#1}}

\newcommand{\seta}{A}
\newcommand{\setb}{B}
\newcommand{\openset}{\mathcal{U}}
\newcommand{\vectorspace}{\mathcal{V}}

\newcommand{\bandwidth}{h}
\newcommand{\smoothedobject}[1]{#1_{\bandwidth}}
\newcommand{\smoothedpdf}{\smoothedobject{\pdf}}
\newcommand{\kernel}{K}
\newcommand{\kernelx}{\kernel(\x)}
\newcommand{\kernelh}{\smoothedobject{\kernel}}
\newcommand{\kdef}[2]{\hat{\pdf}_{#1, #2}}
\newcommand{\kdefnh}{\kdef{\samplesize}{\bandwidth}}
\newcommand{\kdefnhx}{\kdefnh(\x)}

\newcommand{\boundaryof}[1]{\boundary#1}
\newcommand{\bump}{\mathcal{B}}
\newcommand{\bumpof}[1]{\bump^{#1}}
\newcommand{\genericbumpboundary}{\boundaryof{\genericbump}}
\newcommand{\approxbump}{\approxobject{\bump}_{\samplesize, \bandwidth}^{\curvaturefunctional}}
\newcommand{\approxbumpboundary}{\boundaryof{\approxbump}}
\newcommand{\genericbump}{\bumpof{\curvaturefunctional}}
\newcommand{\concavebump}{\bumpof{\eigenvalue_1}}
\newcommand{\convexbump}{\bumpof{\eigenvalue_{\dimension}}}
\newcommand{\meancurvaturebump}{\bumpof{\normalizedgrad}}
\newcommand{\laplacianbump}{\bumpof{\laplacian}}
\newcommand{\hessiandeterminantbump}{\bumpof{\det}}
\newcommand{\bumpdomain}{\Theta}

\newcommand{\classcontinous}[2]{\mathcal{C}^{#1}(#2)}
\newcommand{\classlp}[1]{\mathcal{L}^{#1}}
\newcommand{\classbounded}[1]{\ell^{\infty}(#1)}

\newcommand{\eigenvalue}{\lambda}
\newcommand{\eigenvalueof}[2]{\eigenvalue_{#1}[#2]}
\newcommand{\eigenvalueofpdf}[1]{\eigenvalueof{#1}{\pdf}}
\newcommand{\largestpdfeigenvalue}{\eigenvalueofpdf{1}}
\newcommand{\smallestpdfeigenvalue}{\eigenvalueofpdf{\dimension}}
\newcommand{\eigenvalueofkdefnh}[1]{\eigenvalueof{#1}{\kdefnh}}
\newcommand{\eigenvalueofpdfx}[1]{\eigenvalueofpdf{#1}(\x)}
\newcommand{\diffconsecutiveeig}[2]{\{\eigenvalueof{#1}{#2}(\x) - \eigenvalueof{#1 + 1}{#2}(\x)\}}

\newcommand{\criterionfunction}{\Psi}
\newcommand{\approxcriterionfunction}{\approxobject{\criterionfunction}}
\newcommand{\solutionmanifold}{\mathcal{M}}
\newcommand{\approxsolutionmanifold}{\approxobject{\solutionmanifold}}

\newcommand{\zerovectorlevelset}[1]{\{ \x \in \rd : #1(\x) = 0 \}}

\newcommand{\distancepointset}[2]{d(#1, #2)}
\newcommand{\hausdorffdistance}[2]{\mathrm{Haus}(#1, #2)}
\newcommand{\norminf}[1]{\norm{#1}_{\infty}}
\newcommand{\norminfk}[2]{\norm{#1}_{\infty, #2}}
\newcommand{\directeddistance}[2]{\sup_{\x \in #1} \distancepointset{\x}{#2}}
\newcommand{\fattening}[2]{#1 \oplus #2}

\newcommand{\fatteningexcess}{\varepsilon}
\newcommand{\fatteningconstant}{\delta}

\newcommand{\fatbumpboundary}{\fattening{\genericbumpboundary}{\fatteningconstant}}
\newcommand{\criteriongradientconstant}{\lambda}
\newcommand{\fatmanifold}{\fattening{\solutionmanifold}{\fatteningconstant}}

\newcommand{\criterionhypotheses}{A}
\newcommand{\criterionhypothesesinline}{(\mathrm{\criterionhypotheses})}
\newcommand{\approxcriterionhypotheses}{B}

\newcommand{\underlyingcurvature}{\varphi}
\newcommand{\curvaturerdzerolevelset}[1]{\rdzerolevelset{\curvaturefunctionalof{#1}}}
\newcommand{\minbetweenextrasmoothingandtwo}{\min \{\extrasmoothing, 2\}}

\newcommand{\infdiffconsecutiveeig}[1]{\inf_{\x \in \fatbumpboundary} \diffconsecutiveeig{\genindexi}{#1}}
\newcommand{\erroreigenvalue}[1]{\sup_{\x \in \rd} \smallabsval{\eigenvalueof{#1}{\pdf}(\x) - \eigenvalueof{#1}{\kdefnh}(\x)}}

\newcommand{\norminfcriteriondiff}{\norminf{\approxcriterionfunction - \criterionfunction}}
\newcommand{\norminfcurvaturediff}{\norminfk{\curvaturefunctionalof{\kdefnh} - \curvaturefunctionalof{\pdf}}{\derivativeorder}}
\newcommand{\hausdorffdbetweenboundaries}{\hausdorffdistance{\approxbumpboundary}{\genericbumpboundary}}

\newcommand{\probability}{\mathbb{P}}
\newcommand{\probabilityof}[1]{\probability \left( #1 \right)}
\newcommand{\cdfof}[1]{F_{#1}}
\newcommand{\almostsureconverges}{\xrightarrow[\samplesizegoestoinfty]{\gls*{as}}}
\newcommand{\expectedvalue}{\mathbb{E}}
\newcommand{\expectedvalueof}[1]{\expectedvalue \left[ #1 \right]}
\newcommand{\smallexpectedvalueof}[1]{\expectedvalue [#1]}
\newcommand{\covariance}[2]{\mathrm{Cov}(#1, #2)}
\newcommand{\followdistribution}{\sim}
\newcommand{\equaldistribution}{\stackrel{d}{=}}
\newcommand{\weaklyconverges}[2]{#1 \leadsto #2}
\newcommand{\convergesinprob}[2]{#1 \xrightarrow{\probability} #2}
\newcommand{\kolmogorovdistance}[2]{\uprho_{\mathrm{\acrshort*{cdf}}} \left( #1, #2 \right)}

\newcommand{\bigoh}{O}
\newcommand{\littleoh}{o}
\newcommand{\littleohone}{\littleoh(1)}
\newcommand{\littleohinvsqrtn}{\littleoh(\samplesize^{-1/2})}
\newcommand{\sameorder}{\asymp}
\newcommand{\tendstoassamplesizegoestoinfty}{\xrightarrow[\samplesizegoestoinfty]{}}

\newcommand{\integralonrd}{\int_{\rd}}
\newcommand{\integralonrddx}[1]{\integralonrd #1 \ d{\x}}

\newcommand{\univariatekernel}{\upkappa}
\newcommand{\derivativevectorcomponent}{\beta}
\newcommand{\derivativevector}{\boldsymbol{\derivativevectorcomponent}}
\newcommand{\derivativevectororder}{\absoluteval{\derivativevector}}
\newcommand{\indexeddiffoperator}{\diffoperator^{\derivativevector}}
\newcommand{\partialx}[1]{\partial\pdfvarx_{#1}}
\newcommand{\partialxvectorindex}[1]{\partialx{#1}^{\derivativevectorcomponent_{#1}}}
\newcommand{\partialpdf}{\indexeddiffoperator\pdf}
\newcommand{\partialkdefnh}{\indexeddiffoperator\kdefnh}
\newcommand{\partialpdfx}{\partialpdf(\x)}
\newcommand{\partialkdefnhx}{\partialkdefnh(\x)}
\newcommand{\partialpdfxindex}[1]{\frac{\partial^2 \pdf}{\partialx{#1}^2}}
\newcommand{\samplesizegoestoinfty}{\samplesize\goestoinfty}
\newcommand{\limsamplesizeinf}{\lim_{\samplesizegoestoinfty}}
\newcommand{\positiveintegerspowd}{\positiveintegers^{\dimension}}
\newcommand{\positiveintegerspowdord}[1]{\positiveintegerspowd[#1]}
\newcommand{\supnormderivativekde}{\smallsupnorm{\partialkdefnhx - \smallexpectedvalueof{\partialkdefnhx}}}
\newcommand{\orderfraclog}[1]{\frac{\log \samplesize}{\samplesize \bandwidth^{\dimension + #1}}}
\newcommand{\inlineorderfraclog}[1]{\samplesize^{-1} \bandwidth^{-(\dimension + #1)} \log \samplesize}
\newcommand{\boundsupnormderivativekde}{\orderfraclog{2\absoluteval{\derivativevector}}}
\newcommand{\boundariascastro}{b}
\newcommand{\sqrtboundsupnormderivativekde}{\sqrt{\boundsupnormderivativekde}}
\newcommand{\totalsuperror}{\smallsupnorm{\partialkdefnhx - \partialpdfx}}
\newcommand{\bias}{\smallexpectedvalueof{\partialkdefnhx} - \partialpdfx}
\newcommand{\biassuperror}{\smallsupnorm{\bias}}
\newcommand{\biasintegrand}{\partialpdf(\x - \bandwidth\z) - \partialpdfx}
\newcommand{\indexedderivativevector}[1]{\derivativevector_{#1}}
\newcommand{\multivariatepartialpdf}[1]{\diffoperator^{\indexedderivativevector{1}, \dots, \indexedderivativevector{\numderivatives}}#1}
\newcommand{\indexeddiffoperatorindex}[1]{\diffoperator^{\indexedderivativevector{#1}} \pdf(\x)}

\newcommand{\holderexponent}{\alpha}
\newcommand{\integralradius}{\delta}

\newcommand{\minmaxexponent}{s}

\newcommand{\taylorremainder}[1]{\mathcal{R}_{#1}(\x, \z, \bandwidth)}

\newcommand{\bootstrap}[1]{#1^{*}}
\newcommand{\empiricalbootstrapprob}{\bootstrap{\probability}_{\samplesize}\{\originalsample\}}

\newcommand{\resamplesize}{m}
\newcommand{\kdefmh}{\kdef{\resamplesize}{\bandwidth}}

\newcommand{\originalprocess}{\mathfrak{X}}
\newcommand{\originalsample}{\originalprocess_{\samplesize}}

\newcommand{\bootstrapsupreumumnormerror}[1]{\bootstrap{\mathcal{E}_{\samplesize, \bandwidth}}[#1 | \originalsample]}
\newcommand{\bootstrapsupreumumnormerrorm}[1]{\bootstrap{\mathcal{E}_{\resamplesize, \bandwidth}}[#1 | \originalsample]}
\newcommand{\bootstrapkdefnhof}[1]{\bootstrap{\kdefnh}(#1 | \originalsample)}
\newcommand{\bootstrapkdefmhof}[1]{\bootstrap{\kdefmh}(#1 | \originalsample)}
\newcommand{\bootstrapsamplevar}{\bootstrap{\samplevar}}
\newcommand{\conditionalkdefnhof}[1]{\kdefnh(#1 | \originalsample)}

\newcommand{\oneminusconfidence}{\alpha}
\newcommand{\confidence}{1 - \oneminusconfidence}

\newcommand{\margin}{\zeta_{\samplesize, \bandwidth}^{\oneminusconfidence}}
\newcommand{\marginrv}{Z_{\samplesize, \bandwidth}}
\newcommand{\limitingmarginrv}{\mathcal{Z}}
\newcommand{\approxmargin}{\approxobject{\zeta}_{\samplesize, \bandwidth}^{\oneminusconfidence}}

\newcommand{\superrorcurvature}{\smallsupnormbumpdomain{\curvaturefunctionalof{\kdefnh}(\x) - \curvaturefunctionalof{\smoothedpdf}(\x)}}
\newcommand{\supreumumnormerror}[1]{\mathcal{E}_{\samplesize, \bandwidth}[#1]}
\newcommand{\supreumumnormerrorfunctional}[1]{\mathcal{S}_{\samplesize, \bandwidth}[#1]}

\newcommand{\absvalsecondderivatives}{\smallabsval{\differential_{\genindexi\genindexj} \kdefnh(\x) - \differential_{\genindexi\genindexj} \smoothedpdf(\x)}}
\newcommand{\supnormsecondderivativesij}{\supreumumnormerror{\secondorderdifffunctional{\genindexi}{\genindexj}}}
\newcommand{\bootstrapsupnormderivativesij}{\bootstrapsupreumumnormerror{\secondorderdifffunctional{\genindexi}{\genindexj}}}

\newcommand{\gaussianprocess}{\mathbb{B}}
\newcommand{\gaussianprocessbootstrap}{\gaussianprocess_{\originalsample}}

\newcommand{\empiricalprocess}{\mathbb{G}_{\samplesize}}
\newcommand{\empiricalprocessbootstrap}{\mathbb{G}_{\resamplesize}}

\newcommand{\gaussianprocessrv}{\mathbf{B}}
\newcommand{\gaussianprocessrvh}{\smoothedobject{\gaussianprocessrv}}
\newcommand{\normalizedgaussianprocessrvh}{\smoothedobject{\normalizedobject{\gaussianprocessrv}}}
\newcommand{\gaussianprocessrvhprime}{\gaussianprocessrvh'}
\newcommand{\gaussianprocessrvnh}{\gaussianprocessrv_{\samplesize, \bandwidth}\{\originalsample\}}
\newcommand{\sumsupgaussians}{\normalizedgaussianprocessrvh[\secondorderdifffunctional{\genindexi}{\genindexj}]}

\newcommand{\kernelsvcclass}{\mathcal{K}}
\newcommand{\coverableclass}{\mathcal{F}}
\newcommand{\measurableclass}{\mathcal{F}}
\newcommand{\envelopefunction}{\Psi}
\newcommand{\vcclass}{\mathcal{F}}
\newcommand{\vcclassh}{\vcclass_{\bandwidth}}
\newcommand{\vcclasshprime}{\vcclassh'}
\newcommand{\measurableclassmember}{\varphi}
\newcommand{\vcclassmember}{\varphi}
\newcommand{\vcbound}{b}
\newcommand{\vckernelsbound}{C}
\newcommand{\coveringnumber}{\mathcal{N}}
\newcommand{\coveringepsilon}{\epsilon}
\newcommand{\coveringconstanta}{A}
\newcommand{\coveringconstantnu}{\nu}
\newcommand{\coveringprobmeasure}{\mathbb{Q}}
\newcommand{\supovervcclass}{\sup_{\vcclassmember \in \vcclass}}
\newcommand{\supovervcclassof}[1]{\supovervcclass \vert #1(\vcclassmember) \vert}
\newcommand{\supovervcclassh}{\sup_{\vcclassmember \in \vcclassh}}
\newcommand{\supovervcclasshof}[1]{\supovervcclassh \vert #1(\vcclassmember) \vert}
\newcommand{\supovervcclasshprime}{\sup_{\vcclassmember \in \vcclasshprime}}
\newcommand{\supovervcclasshprimeof}[1]{\supovervcclasshprime \vert #1(\vcclassmember) \vert}

\newcommand{\processstdv}{\sigma}
\newcommand{\anticoncentrationconstant}{\mu}
\newcommand{\anticoncentrationconstantprime}{\anticoncentrationconstant'}
\newcommand{\chernozhukovuniversalconst}{A}
\newcommand{\chernozhukovfirstconst}{\chernozhukovuniversalconst_1}
\newcommand{\modifiedchernofirstconst}{\tilde{\chernozhukovuniversalconst}_1}
\newcommand{\chernozhukovsecondtconst}{\chernozhukovuniversalconst_2}
\newcommand{\chernozhukovfreeconst}{\gamma}

\newcommand{\genericdiffop}{\mathcal{D}}
\newcommand{\scaling}{\sqrt{\samplesize\bandwidth^{\dimension + \difforder}}}

\newcommand{\bootstrapconvergencerategen}[1]{%
	\left[%
		\frac{\log^7 \samplesize}{\samplesize \bandwidth^{\dimension + #1}}%
		\right]^{1/8}}

\newcommand{\bootstrapconvergencerateorder}{\bootstrapconvergencerategen{\difforder}}
\newcommand{\diffcdfgaussiansup}{\Omega_{\samplesize, \bandwidth}(\originalsample)}

\newcommand{\quantile}{\mathcal{Q}}
\newcommand{\quantilefuncofp}[2]{\quantile_{#2}(#1)}
\newcommand{\quantilefuncofrv}[2]{\quantile_{#2} \{ #1 \}}
\newcommand{\tvar}[2]{\mathrm{\acrshort*{tvar}}_{#1} \left\{ #2 \right\}}

\newcommand{\smoothedbump}{\smoothedobject{\bump}}
\newcommand{\smoothedbumpof}[1]{\smoothedbump^{#1}}
\newcommand{\genericsmoothedbump}{\smoothedbumpof{\curvaturefunctional}}

\newcommand{\upperconfidentbump}[1]{\bar{\bump}_{\samplesize, \bandwidth}^{#1}}
\newcommand{\genericupperconfidentbump}{\upperconfidentbump{\curvaturefunctional}(\margin)}
\newcommand{\genericupperconfidentbumpapprox}{\upperconfidentbump{\curvaturefunctional}(\approxmargin)}
\newcommand{\lowerconfidentbump}[1]{\ubar{\bump}_{\samplesize, \bandwidth}^{#1}}
\newcommand{\genericlowerconfidentbump}{\lowerconfidentbump{\curvaturefunctional}(\margin)}
\newcommand{\genericlowerconfidentbumpapprox}{\lowerconfidentbump{\curvaturefunctional}(\approxmargin)}

\newcommand{\confidentbumpset}[1]{%
	\{ \x \in \bumpdomain : \curvaturesign \curvaturefunctionalof{\kdefnh}(\x) \geq #1 \}%
}

\input{01-6-bibliography.tex}

\begin{document}

	\newcommand{\papertitle}{Bump hunting through density curvature features}

\title{\papertitle}
\date{}

\hypersetup{
  pdftitle={\papertitle},
}

	\ifnum\blind=0
		
\newcommand{\joseechacon}{Jos\'e E. Chac\'on}
\newcommand{\addressjoseechaon}{Departamento  de  Matem\'aticas,  Universidad  de  Extremadura, Badajoz, Spain.}
\newcommand{\emailjoseechacon}{jechacon@unex.es}
\newcommand{\orcidjoseechacon}{0000-0002-3675-1960}
\newcommand{\footjoseechacon}{$^{\dagger}$}
\newcommand{\footaddressjoseechaon}{\footjoseechacon\addressjoseechaon}
\newcommand{\footorcidjoseechacon}{\footjoseechacon\orcid{\orcidjoseechacon}}

\newcommand{\javierfdezserrano}{Javier Fern\'andez Serrano}
\newcommand{\addressjavierfdezserrano}{Departamento  de  Matem\'aticas,  Universidad  Aut\'onoma de Madrid, Madrid, Spain.}
\newcommand{\emailjavierfdezserrano}{javier.fernandezs01@estudiante.uam.es}
\newcommand{\orcidjavierfdezserrano}{0000-0001-5270-9941}
\newcommand{\footjavierfdezserrano}{$^{\ddagger}$}
\newcommand{\footaddressjavierfdezserrano}{\footjavierfdezserrano\addressjavierfdezserrano}
\newcommand{\footorcidjavierfdezserrano}{\footjavierfdezserrano\orcid{\orcidjavierfdezserrano}}

\author{\joseechacon\footjoseechacon}
\address{\footaddressjoseechaon}
\email{\footjoseechacon\emailjoseechacon \ \Letter}
\thanks{\footorcidjoseechacon}

\author{\javierfdezserrano\footjavierfdezserrano}
\address{\footaddressjavierfdezserrano}
\email{\footjavierfdezserrano\emailjavierfdezserrano}
\thanks{\footorcidjavierfdezserrano}

\hypersetup{
  pdfauthor={\joseechacon, \javierfdezserrano}
}

\ifnum\material=1
	\makeatletter
	\def\blfootnote{\xdef\@thefnmark{}\@footnotetext}
	\makeatother

	\AtEndDocument{
		\blfootnote{
			\begin{flushleft}
				\hspace{1em}\mbox{\textit{Authors}:\space} \textsc{\authors}.
				\\
				\bigskip
				\hspace{1em}\textsc{\footaddressjoseechaon}
				\\
				\hspace{1em}\textsc{\footaddressjavierfdezserrano}
				\\
				\hspace{1em}\mbox{\textit{E-mail addresses}:\space} \texttt{\emails}.
				\\
				\hspace{1em}\footorcidjoseechacon.
				\\
				\hspace{1em}\footorcidjavierfdezserrano.
			\end{flushleft}
		}
	}
\fi

	\fi

	\newcommand{\abstractmeta}{%
	Bump hunting deals with finding in sample spaces meaningful data subsets known as bumps.
	These have traditionally been conceived as modal or concave regions in the graph of the underlying density function.
	We define an abstract bump construct based on curvature functionals of the probability density.
	Then, we explore several alternative characterizations involving derivatives up to second order.
	In particular, a suitable implementation of Good and Gaskins' original concave bumps is proposed in the multivariate case.
	Moreover, we bring to exploratory data analysis concepts like the mean curvature and the Laplacian that have produced good results in applied domains.
	Our methodology addresses the approximation of the curvature functional with a plug-in kernel density estimator.
	We provide theoretical results that assure the asymptotic consistency of bump boundaries in the Hausdorff distance with affordable convergence rates.
	We also present asymptotically valid and consistent confidence regions bounding curvature bumps.
	The theory is illustrated through several use cases in sports analytics with datasets from the NBA, MLB and NFL.
	We conclude that the different curvature instances effectively combine to generate insightful visualizations.
}

\begin{abstract}
	\abstractmeta
\end{abstract}

\hypersetup {
  pdfsubject={\abstractmeta},
}

\newcommand{\keywordbumphunting}{bump hunting}
\newcommand{\keywordconcavity}{concavity}
\newcommand{\keywordgaussiancurvature}{Gaussian curvature}
\newcommand{\keywordkdederivatives}{kernel density derivative estimation}
\newcommand{\keywordlaplacian}{Laplacian}
\newcommand{\keywordmeancurvature}{mean curvature}

\keywords{%
	\keywordbumphunting,
	\keywordconcavity,
	\keywordgaussiancurvature,
	\keywordkdederivatives,
	\keywordlaplacian,
	\keywordmeancurvature
}

\hypersetup{
	pdfkeywords={%
		\keywordbumphunting,
		\keywordconcavity,
		\keywordgaussiancurvature,
		\keywordkdederivatives,
		\keywordlaplacian,
		\keywordmeancurvature
	},
}

\newcommand{\mscnonparamestimation}{62G05}
\newcommand{\mscnonparamasymptotics}{62G20}
\newcommand{\mscgeometry}{60D05}
\newcommand{\mscstatsdatascience}{62R07}

\subjclass[2020]{
	\mscnonparamestimation \ (Primary),
	\mscnonparamasymptotics,
	\mscgeometry,
	\mscstatsdatascience
}

	\maketitle

	\begin{refsection}

		\section{Introduction}

The subject of \gls*{bh} refers to the set estimation task~\cite{Baillo2000} of discovering meaningful data regions, called \textit{bumps}, in a sample space~\cite{Good1980}.
The most representative example is the modal regions in a \gls*{pdf}, which are literally bumps in its graph.
Even though the concept has a broader scope, \gls*{bh} remains relatively unexplored.

Consider the problem of identifying made shots on a basketball court.
Coaches, scouts and other personnel might be interested in extracting shooting patterns for adopting specific pre-game strategies, assessing talent or working on player development.
\figurename~\ref{fig:nba-shot-bumps} illustrates four different ways of constructing bumps with basketball shot data.
\figurename~\ref{fig:35-hdr-bump} and \figurename~\ref{fig:95-hdr-bump} correspond to~\citeauthor{Hyndman1996}'s classical \gls{hdr} configurations, while \figurename~\ref{fig:concave-bump} and \figurename~\ref{fig:laplacian-bump} follow our novel \textit{curvature}-based characterizations.
Each of them presents a distinctive perspective on the underlying shooting tendencies.
\figurename~\ref{fig:35-hdr-bump} and \figurename~\ref{fig:concave-bump} point at fine-grained locations, whereas \figurename~\ref{fig:95-hdr-bump} and \figurename~\ref{fig:laplacian-bump} cover entire influence areas.
Smaller regions suggest spots to prioritize in an offensive or defensive scheme.
The larger ones \textit{connect the dots}, revealing general trends.
Both views complement each other to offer a complete picture.

\renewcommand{\subfigurewidthtwoandone}{0.4\textwidth}
\renewcommand{\graphicswidthtwoandone}{0.7\textwidth}

\begin{figure}
	\centering
	\begin{subfigure}[b]{\subfigurewidthtwoandone}
		\centering
		\includegraphics[width=\graphicswidthtwoandone]{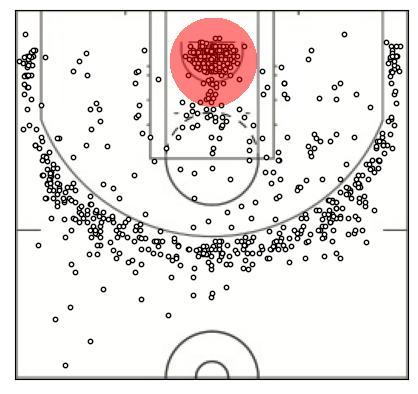}
		\caption{35\%-\gls*{hdr}}
		\label{fig:35-hdr-bump}
	\end{subfigure}
	\begin{subfigure}[b]{\subfigurewidthtwoandone}
		\centering
		\includegraphics[width=\graphicswidthtwoandone]{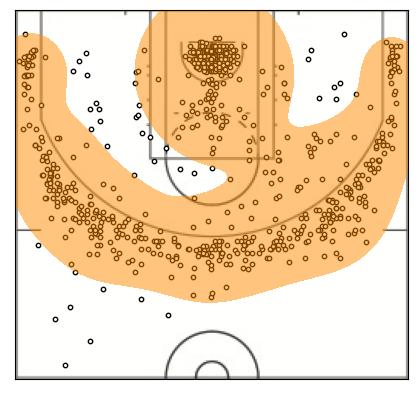}
		\caption{95\%-\gls*{hdr}}
		\label{fig:95-hdr-bump}
	\end{subfigure}
	\begin{subfigure}[b]{\subfigurewidthtwoandone}
		\centering
		\includegraphics[width=\graphicswidthtwoandone]{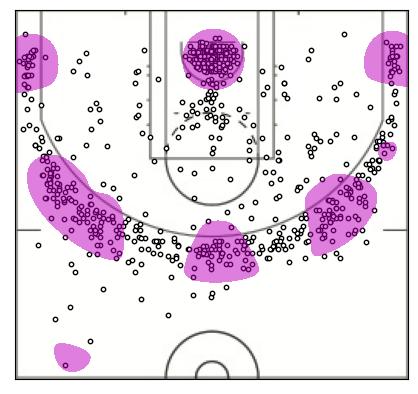}
		\caption{Concave bumps}
		\label{fig:concave-bump}
	\end{subfigure}
	\begin{subfigure}[b]{\subfigurewidthtwoandone}
		\centering
		\includegraphics[width=\graphicswidthtwoandone]{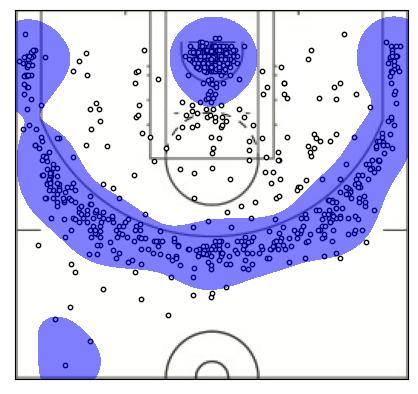}
		\caption{Laplacian bumps}
		\label{fig:laplacian-bump}
	\end{subfigure}
	\caption{
		Four ways of constructing bumps for basketball converted shot data.
		The exact 804 made shot locations are scattered across each sub-figure.
		The top left and right bumps correspond to \glspl*{hdr} comprising 35\% and 95\% of all observations.
		The bottom left bumps highlight regions where the \gls*{pdf} subgraph is locally concave.
		The bottom right bumps comprise points where the Laplacian of the underlying \gls*{pdf} takes negative values.
	}
	\label{fig:nba-shot-bumps}
\end{figure}

\subsection{Goals}

We propose a new \gls*{bh} curvature-based methodology addressing some blind spots of classical methods.
\figurename~\ref{fig:35-hdr-bump} and \figurename~\ref{fig:95-hdr-bump} either miss or mask relevant information.
The finer-grained 35\%-\gls*{hdr} does not include the perimeter concave bumps in \figurename~\ref{fig:concave-bump}.
Meanwhile, the 95\%-\gls*{hdr} fails to keep the short, mid and long ranges well separated, as opposed to the Laplacian bumps in \figurename~\ref{fig:laplacian-bump}.

\vspace{1em}
\paragraph{\textit{Contributions}}

The main contributions of this paper are:
\begin{itemize}[leftmargin=2em]
	\renewcommand{\labelitemi}{\tiny{\ding{228}}}
	\item Presenting a general set estimation framework for curvature-based \gls*{bh}.
	\item Extending concave bumps to the multivariate setting.
	\item Introducing mean curvature and Laplacian bumps.
	\item Deriving consistency convergence rates for curvature bump boundaries.
	\item Building valid and consistent confidence regions for curvature bumps.
	\item Showcasing the numerous applications of curvature-based \gls*{bh}.
\end{itemize}

\subsection{Related work}

One of the first \gls*{bh} references was due to~\citeauthor{Good1980} in~\citeyear{Good1980}~\cite{Good1980}.
They offered a premier definition of a bump as the concave region delimited between two inflection points.
Moreover, they suggested an extension to the multivariate case.
\figurename~\ref{fig:concave-bump} corresponds to our implementation of multivariate concave bumps.

In~\citeyear{Hyndman1996}, \citeauthor{Hyndman1996} introduced the concept of \gls*{hdr}, which he conceives as level sets of the \gls*{pdf} $\pdf$ that enclose a certain probability mass~\cite{Hyndman1996}.
More formally, the $(1 - \hyndalpha)$-level \gls*{hdr} is defined as $\hyndhdralpha = \hyndhdr$, where $\hyndfalpha$ is the largest value such that $\probabilityof{\randomvectorx \in \hyndhdralpha} \geq 1 - \hyndalpha$, and the \gls*{rv} $\randomvectorx$ is such that $\randomvectorx \followdistribution \pdf$.
\glspl*{hdr} satisfy the nice property of being the \textit{smallest} sets with a given probability mass.

In~\citeyear{Chaudhuri1999}, \citeauthor{Chaudhuri1999} presented \gls*{sizer}, envisioning bumps as places where the first derivative becomes zero~\cite{Chaudhuri1999}.
In~\citeyear{Chaudhuri2002},~\citeauthor{Chaudhuri1999} showcased the role of second derivatives in an unpublished manuscript~\cite{Chaudhuri2002}.
Also in~\citeyear{Godtliebsen2002},~\citeauthor{Godtliebsen2002} explored curvature features from a pointwise perspective by assessing Hessian eigenvalue sign combinations in the bivariate case~\cite{Godtliebsen2002}.
A multivariate extension to~\cite{Godtliebsen2002} was formulated in~\citeyear{Duong2008} by~\citeauthor{Duong2008}, targeting the pointwise significance of non-zero Hessian determinants~\cite{Duong2008}.
Lastly, in~\citeyear{Marron2021},~\citeauthor{Marron2021} elaborate on second derivatives in their book~\citetitle{Marron2021}~\cite{Marron2021}.

\subsection{Outline}

The new methodology is presented in Section~\ref{sec:methods}.
The \supplement{} provides the necessary differential geometry foundations.
In turn, Section~\ref{sec:asymptotics} is entirely dedicated to asymptotic consistency and inference results.
A sports analytics application is explored in Section~\ref{sec:application}.
The \supplement{} includes all the proofs and computational details.
We reflect on the proposed methodology in Section~\ref{sec:discussion}.

		\section{Methods}

\label{sec:methods}

Our methodology finds alternative ways of analyzing sample spaces by exploiting \glspl*{pdf}' curvature properties, adhering to~\citeauthor{Chaudhuri1999}'s defence of \gls*{pdf} derivatives.
Considering~\citeauthor{Hyndman1996}'s approach a well-established tool, we believe there are still some blind spots to address with curvature.

\citeauthor{Hyndman1996}'s \glspl*{hdr} have the advantage of always including \textit{global} modes.
However, they may generally miss \textit{local} modes if small enough; lowering the threshold $\hyndalpha$ might not capture them without obfuscating the \gls*{hdr}.
On the other hand, when varying $\hyndalpha$ works, questions remain on the specific value it should take.
Moreover, sometimes it is necessary to explore the whole range of $\hyndalpha \in (0, 1)$ to recover all the relevant \gls*{pdf} features~\cite{Stuetzle2003}.

Consider a $\dimension$-variate \gls*{pdf} $\functiondef{\pdf}{\rd}{\closedzeroinfinterval}$.
We define bumps as subsets of $\rd$ of the form
\begin{equation}
	\label{eq:abstract-bump}
	\genericbump
	=
	\rdzeroupperlevelset{\curvaturesign \curvaturefunctionalof{\pdf}}
	\,,
\end{equation}
for some functional $\curvaturefunctional$ measuring the \textit{curvature} of $\pdf$ at any point, and some sign selector $\curvaturesignselector \in \{0, 1\}$ that will usually be kept implicit.
If the gradient $\grad{\curvaturefunctionalof{\pdf}}$ does not vanish near the zero-level set of $\curvaturefunctionalof{\pdf}$, the bump boundary $\genericbumpboundary$ is retrieved by substituting the inequality with an equality sign in~\eqref{eq:abstract-bump}~\cite[Remark 3.1]{Qiao2020}; see Theorem~\ref{th:stability-theorem} ahead for a formal condition~\cite[Assumption G]{Chen2017}~\cite[]{Chen2022}.
Contrary to \glspl*{hdr}, the idea behind~\eqref{eq:abstract-bump} is that $\curvaturefunctional$ carries an implicit threshold, say zero, to determine if a point belongs to the bump, solving the arbitrariness of the choice of $\hyndalpha$ in \glspl*{hdr}.

Once some curvature functional is chosen, we propose to employ a kernel plug-in estimator of $\genericbump$, replacing $\pdf$ with its \gls*{kde} in~\eqref{eq:abstract-bump}.
Thus, given a sample $\nsample$ of \gls*{iid} random variables with \gls*{pdf} $\pdf$ and a bandwidth $\bandwidth > 0$, we consider the \gls*{kde} of $\pdf$ as
\begin{equation}
	\label{eq:kernel-estimator}
	\kdefnhx
	=
	\oneoversamplesize
	\sumitosamplesize
	\kernelh(\x - \samplevar_{\genindexi})
	=
	\frac{1}{\samplesize \bandwidth^{\dimension}}
	\sumitosamplesize
	\kernel \left( \frac{\x - \samplevar_{\genindexi}}{\bandwidth} \right)
	\,,
\end{equation}
for some kernel function $\kernel$, typically a $\dimension$-variate \gls*{pdf}.
Using~\eqref{eq:kernel-estimator}, we then define the plug-in estimator of~\eqref{eq:abstract-bump} as
\begin{equation}
	\label{eq:approx-bump}
	\approxbump
	=
	\rdzeroupperlevelset{\curvaturesign \curvaturefunctionalof{\kdefnh}}
	\,.
\end{equation}

To a first approximation, a \textit{scalar} bandwidth is chosen for simplicity.
\citeauthor{Chacon2018} demonstrated that, for $\dimension > 1$, unconstrained bandwidth matrices produce significant performance gains, especially in \gls*{kdde}~\cite[Section 5.2]{Chacon2018}.
Preliminary experiments seem to support their recommendation also for curvature-based \gls*{bh}.
Nonetheless, all the theoretical developments and, consequently, all the exhibition figures in this paper obey this simplification.
On the other hand, the kernel $\kernel$ has a lower impact on the results.
Most of the statements in Section~\ref{sec:asymptotics} do not impose a particular choice.
However, all of them are compatible with the Gaussian kernel (see~\cite{AriasCastro2016, Chen2015, Chen2016}), which is almost universally preferred in a multivariate setting~\cite[p. 15]{Chacon2018}.

For $\dimension = 1$,~\citeauthor{Chaudhuri2002} studied the functional $\curvaturefunctionalof{\pdf} = \pdfsecondderivative$, which leads to \textit{concave} bumps, if $\curvaturesignselector = 1$, or \textit{convex} dips, if $\curvaturesignselector = 0$.
Different alternatives arise in the multivariate case.
The geometrical concepts in the \supplement{} lay the grounds for characterizing bumps in alternative ways to \glspl*{hdr}.
Considering \glspl*{pdf} as hypersurfaces, notions like the mean and Gaussian curvatures find new usages in statistics.
\figurename~\ref{fig:mixture-bumps} illustrates the two main kinds of curvature bumps in this paper.
Even though $\curvaturefunctional$ may \textit{a priori} depend on partial derivatives of $\pdf$ of arbitrary order, the theory of hypersurfaces in the \supplement{} suggests that our quest for curvature features is essentially fulfilled with up to second derivatives of the \gls*{pdf} $\pdf$.

Given the connection of curvature with second derivatives, we propose targeting $\bandwidthtargetorder = 2$ in one of the standard bandwidth selectors~\cite{Chacon2013}.
The same heuristic worked well for \gls*{kde}-based applications such as mean shift clustering or feature significance testing~\cite[Chapter 6]{Chacon2018}.

\renewcommand{\subfigurewidthtwoandone}{0.4\textwidth}
\renewcommand{\graphicswidthtwoandone}{\textwidth}

\begin{figure}
	\centering
	\begin{subfigure}[b]{\subfigurewidthtwoandone}
		\centering
		\includegraphics[width=\graphicswidthtwoandone]{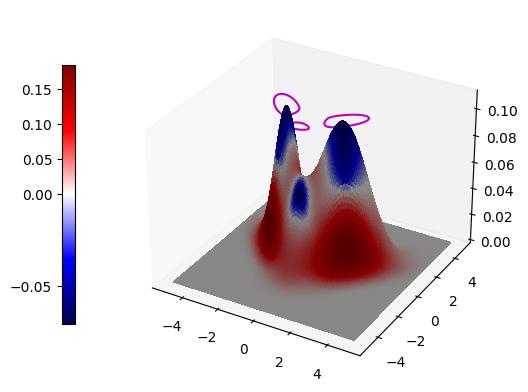}
		\caption{Concave bumps}
		\label{fig:mixture-bump-concave}
	\end{subfigure}
	\begin{subfigure}[b]{\subfigurewidthtwoandone}
		\centering
		\includegraphics[width=\graphicswidthtwoandone]{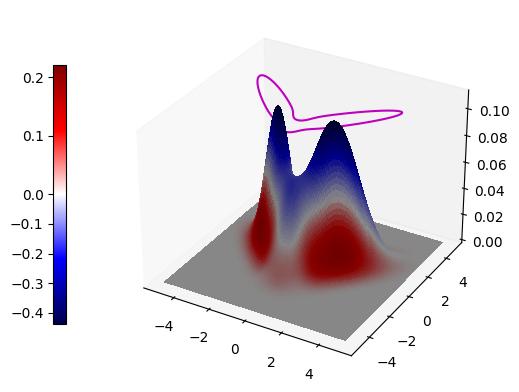}
		\caption{Mean curvature bump}
		\label{fig:mixture-bump-mean-curvature}
	\end{subfigure}
	\begin{subfigure}[b]{\subfigurewidthtwoandone}
		\centering
		\vspace{3 ex}
		\includegraphics[width=\graphicswidthtwoandone]{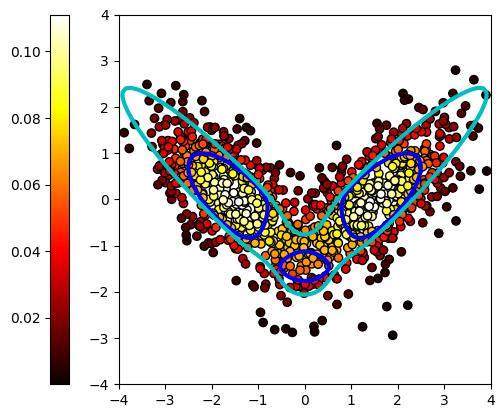}
		\caption{Both kinds of bumps}
		\label{fig:mixture-bump-scatter}
	\end{subfigure}
	\caption{
	Curvature bumps for a bivariate Gaussian mixture encompassing two equally-weighted components with means $\mean_1 = [-3/2, 0]$, $\mean_2 = [3/2, 0]$ and covariance matrices $\covariancematrix_1 = [1, -0.7; -0.7, 1]$, $\covariancematrix_2 = [1, 0.7; 0.7, 1]$.
	The top two sub-figures show the same graph of the \gls*{pdf} $\pdf$.
	The area colours refer to the values taken by a specific curvature functional $\curvaturefunctionalof{\pdf}$ at each point.
	For the left-hand picture, this function is the $\largestpdfeigenvalue$ that defines concave bumps~\eqref{eq:concave-bump}; on the right, it is the mean curvature $\divergenceof{\normalizedgradof{\pdf}}$ in~\eqref{eq:mean-curvature-bump}.
	The magenta \textit{halos} represent the zero level sets of those functionals and, thus, the corresponding bump boundaries.
	Concave and mean curvature bump boundaries show in blue and cyan in the bottom sub-figure, along with a 1,000-observation random sample from the mixture, where each point is coloured according to the value of $\pdf$.
	}
	\label{fig:mixture-bumps}
\end{figure}

\subsection{Concavity and convexity}

Given a sufficiently smooth \gls*{pdf} $\pdf$, let us define $\eigenvalueofpdf{\genindexi}$, for $\genindexi \in \{1, 2, \dots, \dimension\}$, as the function mapping $\x \in \rd$ to the $\genindexi$-th largest possibly repeated eigenvalue of $\hessian{\pdf}(\x)$, the Hessian matrix of $\pdf$ at $\x$, i.e.,
\begin{equation}
	\label{eq:ordered-eigenvalues}
	\eigenvalueofpdfx{1} \geq \eigenvalueofpdfx{2} \geq \dots \geq \eigenvalueofpdfx{\dimension}
	\,,
\end{equation}
for all $\x \in \rd$.
As mentioned in the \supplement{}, the eigenvalues of the Hessian (or the shape operator, equivalently) determine local concavity and convexity.
Let us assume that $\curvaturesign \eigenvalueofpdf{\genindexi} > 0$, for all $\genindexi$ on some subset $\openset \subset \rd$.
If $\curvaturesignselector = 0$, $\pdf$ will be locally convex, whereas if $\curvaturesignselector = 1$, it will be locally concave on $\openset$.
Considering the ordering of functions~\eqref{eq:ordered-eigenvalues}, we can express the former concave and convex bumps in terms of a single functional, aligned with a specific sign $\curvaturesignselector$, as, respectively,
\begin{multicols}{2}
	\noindent
	\begin{equation}
		\label{eq:concave-bump}
		\concavebump
		=
		\rdzerolowerlevelset{\largestpdfeigenvalue}
		\,,
	\end{equation}
	\columnbreak
	\begin{equation}
		\label{eq:convex-bump}
		\convexbump
		=
		\rdzeroupperlevelset{\smallestpdfeigenvalue}
		\,.
	\end{equation}
\end{multicols}
\vspace*{-\multicolsep}

The concave region~\eqref{eq:concave-bump} yields the most recognizable flavour of bumps in the literature, this time in a multivariate setting.
It is the method depicted in \figurename~\ref{fig:concave-bump}.
As for~\eqref{eq:convex-bump}, they are actually not bumps but \textit{dips}.
Assuming non-degenerate Hessians, concave bumps typically delineate areas near local \gls*{pdf} modes, while convex dips do with local minima.
Consequently, the former and the latter are known as \textit{peaks} and \textit{holes}~\cite[Table 1]{Godtliebsen2002}.

When concave bumps contain local modes, they make the most natural definition of a $\dimension$-dimensional neighbourhood.
Although straightforward, considering modal regions as $\fatteningexcess$-fattenings or enlargements (see Section~\ref{sec:notational-preliminaries} below) poses challenges regarding the choice of $\fatteningexcess > 0$, as similarly argued for $\hyndalpha$ in \glspl*{hdr}.
Besides, employing a single radius $\fatteningexcess$ limits the overall expressiveness of the bump.
On the other hand, if we saw modal regions as \textit{basins of attraction} instead~\cite{Chacon2015}, despite $\fatteningexcess$ disappearing and attaining more flexibility, we would not be pursuing a solution to a \gls*{bh} problem any more but a clustering one, giving up on the cohesive sense of bumps.
In this respect, concave bumps provide us with an elegant compromise answer.

Moreover, this modal vicinity notion seamlessly incorporates the missing mode scenario.
Concave bumps point out incipient modal regions as the central \textit{mouth} in \figurename~\ref{fig:mixture-bump-concave} and \figurename~\ref{fig:mixture-bump-scatter}, which does not contain a mode.
Such \textit{weak} modal regions are well-known in the context of univariate mode hunting as \textit{shoulders}, representing complicated cases~\cite{Cheng1999}.
As for \gls*{bh}, $\dimension$-dimensional shoulders deserve attention as evidence of hidden structure.
See the \acrshort*{nfl} application in the \supplement{} for an interpretable dynamic shoulder.
In turn, the mouth in \figurename~\ref{fig:mixture-bumps} is characteristic of mixtures whose components influence each other significantly.
All in all, concave bumps subsume the modal regions, having a slightly broader reach.

\subsection{Gradient divergence}

Concave bumps may be too restrictive in some use cases.
Imagine the \gls*{pdf} graph as a \textit{landscape}, with mountains being local high-density regions.
Concave bumps originate near mountain \textit{peaks}, missing most of the \textit{hillside}.
Mean curvature allows the discovery of entire mountain chains.

The shape operator is a linear map of the tangent space that measures how a manifold bends in different directions (see the \supplement{} for a formal definition).
Let us consider its eigenvalues: the principal curvatures.
Concavity requires all principal curvatures to be negative.
By contrast, the mean curvature adds them all so that only the net sign matters.
Computing curvature in this way fills the gaps between concave peaks in a \textit{long ridge}~\cite[Table 1]{Godtliebsen2002}, as depicted in \figurename~\ref{fig:mixture-bump-mean-curvature} and \figurename~\ref{fig:mixture-bump-scatter} in the form of a \textit{boomerang}.

The \supplement{} shows the connection between the mean curvature and divergence of the normalized version of the gradient $\normalizedgraddefinition$.
The divergence operator takes positive values when the argument field \textit{diverges} from a point, whereas the sign is negative when it \textit{converges}.
Therefore, we define the mean curvature bump as
\begin{equation}
	\label{eq:mean-curvature-bump}
	\meancurvaturebump
	=
	\rdzerolowerlevelset{\divergenceof{\normalizedgradof{\pdf}}}
	\,.
\end{equation}

When the gradient is slight, as is usually the case for \glspl*{pdf} (one can even tweak the scale of the random variables to make $\norm{\gradf}$ small), the Laplacian $\laplacianof{\pdf} = \divergenceof{\grad{\pdf}} = \sumitod \diffoperator^2 \pdf / \partialx{\genindexi}^2$ roughly approximates the mean curvature (see~\cite[Equation 5.28]{Folland2002}).
Hence, we define the Laplacian bump as
\begin{equation}
	\label{eq:laplacian-bump}
	\laplacianbump
	=
	\rdzerolowerlevelset{\laplacianof{\pdf}}
	\,.
\end{equation}
Note that $\concavebump \subset \laplacianbump$.
Even though~\eqref{eq:laplacian-bump} may be less intrinsic than~\eqref{eq:mean-curvature-bump}, it has a more straightforward form, for $\laplacian$ is a second-order linear differential operator on $\pdf$.
A discretized version of the Laplacian operator has been used for contour detection in image processing through the \textit{Laplacian-of-Gaussian} algorithm~\cite{Haralick1992}.
We have already seen an example of a Laplacian bump in \figurename~\ref{fig:laplacian-bump}.
The results would have been almost indistinguishable if the mean curvature had been employed.

The term \textit{ridge} was used above to convey a mountain range covering several peaks following~\cite{Godtliebsen2002}.
Ridges also refer in the statistical literature to a specific definition of higher-dimensional \gls*{pdf} modes~\cite{Chen2015}.
This concept of ridge shares with Laplacian and mean curvature bumps the ability to unveil filament-like structures.
However, ridges are intrinsically one-dimensional in their most typical form.
For them to extend to $\rd$, one would need to take an $\fatteningexcess$-enlargement, introducing some arbitrariness and rigidity with $\fatteningexcess$ that gradient divergence bumps do not have.
In our context, we will stick to the informal meaning of \textit{ridge} in the following sections.

\subsection{Intrinsic curvature}

The Gaussian curvature is an intrinsic measure derived from the shape operator (see the \supplement{} for a precise definition).
This and the Hessian determinant provide alternative ways to detect warps.
The analysis of these two notions is more subtle than in the previous sections: from the definition of Gaussian curvature in the \supplement{}, many sign combinations among the multiplied principal curvatures produce the same net sign.

The Gaussian curvature and the Hessian determinant differ by a positive factor; thus, if we set the bump detection threshold at zero, we can restrict our analysis to the latter.
In the bivariate case, the bump
\begin{equation}
	\label{eq:hessian-determinant-bump}
	\hessiandeterminantbump
	=
	\zeroupperlevelset{\rtwo}{\determinant{\hessian{\pdf}}}
\end{equation}
coincides with the union of~\eqref{eq:concave-bump} and~\eqref{eq:convex-bump}.
Therefore,~\eqref{eq:hessian-determinant-bump} is helpful for detecting both concave bumps and convex dips simultaneously.
We will refer to~\eqref{eq:hessian-determinant-bump} as a \textit{Gaussian bump}.

		\section{Asymptotics}

\label{sec:asymptotics}

This section will demonstrate the soundness of plug-in estimators in the asymptotic regime for curvature bumps.

\subsection{Consistency}

\label{sec:consistency}

We rely on a recent result by~\citeauthor{Chen2022} to prove consistency~\cite{Chen2022}.
Let
\begin{multicols}{2}
	\noindent
	\begin{equation}
		\label{eq:solution-manifold}
		\solutionmanifold
		=
		\zerovectorlevelset{\criterionfunction}
		\,,
	\end{equation}
	\columnbreak
	\begin{equation}
		\label{eq:approx-solution-manifold}
		\approxsolutionmanifold
		=
		\zerovectorlevelset{\approxcriterionfunction}
	\end{equation}
\end{multicols}
\vspace*{-\multicolsep}
be two solution manifolds defined by their criterion functions $\functiondef{\criterionfunction, \approxcriterionfunction}{\rd}{\reals}$, respectively.
\citeauthor{Chen2022}'s stability theorem shows that $\solutionmanifold$ and $\approxsolutionmanifold$ are near whenever the criterion functions and their derivatives are close.
In our context, $\criterionfunction$ will represent a curvature measure and $\approxcriterionfunction$ the corresponding kernel plug-in estimator so that $\solutionmanifold$ and $\approxsolutionmanifold$ are the boundaries of the associated curvature bumps.

\subsubsection{Notational preliminaries}

\label{sec:notational-preliminaries}

The theory of convergence in the uniform norm for \gls*{kdde} allows applying~\citeauthor{Chen2022}'s stability theorem to the curvature \gls*{bh} problem.

Vectors of nonnegative integers $\derivativevector = (\derivativevectorcomponent_1, \dots, \derivativevectorcomponent_{\dimension}) \in \positiveintegerspowd$ shall represent partial derivatives through $\partialpdf = \partial^{\derivativevectororder}\pdf / \partialxvectorindex{1} \cdots \partialxvectorindex{\dimension}$, where $\derivativevectororder = \sumitod \derivativevectorcomponent_{\genindexi}$.
Let us call $\positiveintegerspowdord{\derivativeorder} = \{\derivativevector \in \positiveintegerspowd : \derivativevectororder \leq \derivativeorder\}$.
We also include the case $\derivativevector = \zerovector$, which represents the identity.
Let us also define, for any derivative index vectors $\indexedderivativevector{1}, \dots, \indexedderivativevector{\numderivatives} \in \positiveintegerspowd$, the function $\functiondef{\multivariatepartialpdf{\pdf}}{\rd}{\reals^{\numderivatives}}$ as $\multivariatepartialpdf{\pdf}(\x) = (\indexeddiffoperatorindex{1}, \dots, \indexeddiffoperatorindex{\numderivatives})$.

We will denote $\classcontinous{\difforder}{\seta}$ the class of functions $\functiondef{\testfunction}{\seta \subset \rd}{\reals}$ with continuous partial derivatives up to $\difforder$-th order.
Likewise, we will say that a function $\functiondef{\testfunction}{\rd}{\reals}$ is Hölder continuous with exponent $\holderexponent \in (0, 1]$ if there exists $\constant \in \openzeroinfinterval$ such that $\absoluteval{\testfunction(\x) - \testfunction(\y)} \leq \constant \norm{\x - \y}^{\holderexponent}$, for all $\x, \y \in \rd$~\cite{Jiang2017}.
By convention, we include the case $\holderexponent = 0$ when Hölder continuity does not hold for any positive exponent.

For any $\functiondef{\testfunction}{\rd}{\reals}$ and some $\seta \subset \rd$, we denote $\norminf{\testfunction} = \sup_{\x \in \seta}{\smallabsval{\testfunction(\x)}}$, and we will indicate that the supremum is over $\seta$ by explicitly stating that  $\norminf{\testfunction}$ satisfies some property \textit{on} $\seta$.
Also, write $\norminfk{\testfunction}{\derivativeorder} = \max \left\{ \norminf{\indexeddiffoperator\testfunction} : \derivativevector \in \positiveintegerspowd, \derivativevectororder = \derivativeorder \right\}$.
All these norms will formalize how close the criterion functions and their respective derivatives are.

On the other hand, the stability theorem invokes some other concepts related to sets.
Let us define the distance from a point $x \in \rd$ to some subset $\seta \subset \rd$ as $\distancepointset{\x}{\seta} = \inf_{\y \in \seta} \norm{\x - \y}$, and the $\fatteningexcess$-fattening of a set $\seta \subset \rd$, where $\fatteningexcess > 0$, as $\fattening{\seta}{\fatteningexcess} = \{ \x \in \rd : \distancepointset{\x}{\seta} \leq \fatteningexcess \}$.
Finally, the \textit{Hausdorff distance} between two subsets $\seta, \setb \subset \rd$ is $\hausdorffdistance{\seta}{\setb} = \max \left\{ \directeddistance{\setb}{\seta}, \directeddistance{\seta}{\setb} \right\}$.

The problem of uniformly bounding the \gls*{kdde} error refers to finding an infinitesimal bound for $\totalsuperror$.
Note that the latter is bounded by the \textit{bias} $\biassuperror$ plus the \textit{stochastic error} $\supnormderivativekde$.
We will analyze both terms separately.

\subsubsection{Bias analysis}

Lemma~\ref{lm:bias-order} is an extended version of~\cite[Lemma 2]{AriasCastro2016} with alternative hypotheses to ensure consistency under less stringent differentiability assumptions.
Namely, we resort to Hölder and uniform continuity, following the example of~\cite{Jiang2017} and~\cite[Theorem 1.1, p. 42]{Nadaraya1989}.

\begin{lemma}
	\label{lm:bias-order}
	Let $\derivativevector \in \positiveintegerspowd$ be a partial derivative index vector.
	Let $\pdf$ be a \gls*{pdf} in $\classcontinous{\derivativevectororder + \extrasmoothing}{\rd}$, for some $\extrasmoothing \in \extendedpositiveintegers$, with all partial derivatives bounded up to $(\derivativevectororder + \extrasmoothing)$-th order.
	Assume that $\partialpdf$ is Hölder continuous on $\rd$ with exponent $\holderexponent \in [0, 1]$.
	If the exponent is $\holderexponent = 0$, then ultimately assume that $\partialpdf$ is uniformly continuous.
	Finally, let $\kdefnh$ be the \gls*{kde} of $\pdf$ based on a true \gls*{pdf} kernel $\kernel$ vanishing at infinity and satisfying the moment constraints
	\begin{multicols}{2}
		\noindent
		\begin{equation*}
			\integralonrddx{\x \ \kernelx}
			=
			\zerovector
			\,,
		\end{equation*}
		\columnbreak
		\begin{equation*}
			\integralonrddx{\absoluteval{\x_{\genindexi} \x_{\genindexj}} \ \kernelx}
			<
			\infty
			\,,
		\end{equation*}
	\end{multicols}
	\vspace*{-\multicolsep}
	for all $\genindexi, \genindexj \in \{1, \dots, \dimension\}$.
	Then,
	\begin{equation*}
		\biassuperror
		=
		\begin{cases}
			\bigoh(\bandwidth^{\minmaxexponent}),                  & \mathrm{if} \ \max\{\extrasmoothing, \holderexponent\} > 0 \\
			\littleohone \ \mathrm{as} \ \bandwidth \rightarrow 0, & \mathrm{otherwise}
		\end{cases}
		\,,
	\end{equation*}
	where $\minmaxexponent = \max \{\holderexponent, \minbetweenextrasmoothingandtwo \}$.
\end{lemma}

\subsubsection{Stochastic error analysis}

Lemma~\ref{lm:supnorm-bigoh} below appears as an auxiliary result in~\cite{AriasCastro2016} in the case $\maxdifforder = 3$, but the proof works for an arbitrary $\maxdifforder$.

\begin{lemma}[\citeauthor{AriasCastro2016},~\citeyear{AriasCastro2016}~\cite{AriasCastro2016}]
	\label{lm:supnorm-bigoh}
	Let $\pdf$ be a bounded \gls*{pdf} in $\rd$ and let $\kdefnh$ be the \gls*{kde} of $\pdf$.
	Fix a nonnegative integer $\maxdifforder$ as the maximum partial derivative order.
	Assume that $\kernel$ is a product kernel of the form $\kernel(\pdfvarx_1, \dots, \pdfvarx_{\dimension}) = \prod_{\genindexi = 1}^{\dimension} \univariatekernel_{\genindexi}(\pdfvarx_{\genindexi})$, where each $\univariatekernel_{\genindexi}$ is a univariate \gls*{pdf} of class $\classcontinous{\maxdifforder}{\reals}$.
	Further, assume that all the partial derivatives up to $\maxdifforder$-th order of $\kernel$ are of bounded variation and integrable on $\rd$.
	Then, there exists $\boundariascastro \in (0, 1)$ such that, if $\bandwidth \equiv \bandwidth_{\samplesize}$ is a sequence satisfying $\log \samplesize \leq \samplesize\bandwidth^{\dimension} \leq \boundariascastro\samplesize$, then
	\begin{equation*}
		\supnormderivativekde
		=
		\bigoh\left(\sqrtboundsupnormderivativekde\right)
		\,,
	\end{equation*}
	\gls*{as} for all $\derivativevector \in \positiveintegerspowdord{\maxdifforder}$.
\end{lemma}

Finally, note that Lemma~\ref{lm:supnorm-bigoh} also holds for a sufficiently small but constant $\bandwidth$.

\subsubsection{Total error analysis}

Combining Lemma~\ref{lm:bias-order} and Lemma~\ref{lm:supnorm-bigoh}, we obtain a general consistency result in the supremum norm for \gls*{kdde}.
We will focus on the Gaussian kernel for simplicity, but any other satisfying the conditions in both Lemma~\ref{lm:bias-order} and Lemma~\ref{lm:supnorm-bigoh} would do.

\begin{theorem}
	\label{th:derivative-convergence}
	Let $\derivativevector \in \positiveintegerspowd$ be a partial derivative index vector.
	Let $\pdf$ be a \gls*{pdf} in $\classcontinous{\derivativevectororder + \extrasmoothing}{\rd}$, for some $\extrasmoothing \in \extendedpositiveintegers$, with all partial derivatives bounded up to $(\derivativevectororder + \extrasmoothing)$-th order.
	Assume that $\partialpdf$ is Hölder continuous on $\rd$ with exponent $\holderexponent \in [0, 1]$.
	If the exponent is $\holderexponent = 0$, then ultimately assume that $\partialpdf$ is uniformly continuous.
	Let $\kdefnh$ be the \gls*{kde} of $\pdf$ based on the Gaussian kernel.
	Finally, let $\bandwidth \equiv \bandwidth_{\samplesize}$ be a sequence converging to zero as $\samplesizegoestoinfty$ and satisfying $\samplesize\bandwidth^{\dimension} \geq \log \samplesize$.
	Then,
	\begin{equation*}
		\totalsuperror
		=
		\begin{cases}
			\bigoh \left( \bandwidth^{\minmaxexponent} + \sqrtboundsupnormderivativekde \right),
			\
			 & \mathrm{if} \ \max\{\extrasmoothing, \holderexponent\} > 0
			\\
			\littleohone + \bigoh \left( \sqrtboundsupnormderivativekde \right),
			\
			 & \mathrm{otherwise}
		\end{cases}
		\,,
	\end{equation*}
	\gls*{as} as $\samplesizegoestoinfty$, where $\minmaxexponent = \max \{\holderexponent, \minbetweenextrasmoothingandtwo \}$.
	In particular,
	\begin{equation*}
		\totalsuperror \almostsureconverges 0 \ ,
		\ \mathrm{if} \
		\boundsupnormderivativekde \tendstoassamplesizegoestoinfty 0
		\,.
	\end{equation*}
\end{theorem}

\subsubsection{Manifold stability}

Theorem~\ref{th:stability-theorem} gathers the essential elements of~\citeauthor{Chen2022}'s stability theorem needed in our context.

\begin{theorem}[\citeauthor{Chen2022},~\citeyear{Chen2022}~\cite{Chen2022}]
	\label{th:stability-theorem}
	Let $\functiondef{\criterionfunction, \approxcriterionfunction}{\rd}{\reals}$ and let $\solutionmanifold$ and $\approxsolutionmanifold$ be as defined in~\eqref{eq:solution-manifold} and~\eqref{eq:approx-solution-manifold}, respectively.
	Assume that:
	\begin{enumerate}[label=$\criterionhypotheses$\arabic*.]
		\item There exists $\fatteningconstant > 0$ such that $\criterionfunction$ has bounded first-order derivatives on $\fatmanifold$.
		\item There exists $\criteriongradientconstant > 0$ such that $\norm{\grad{\criterionfunction}(\x)} > \criteriongradientconstant$, for all $\x \in \fatmanifold$.
		\item $\norminfcriteriondiff$ is sufficiently small on $\rd$.
	\end{enumerate}
	Moreover, suppose that:
	\begin{enumerate}[label=$\approxcriterionhypotheses$\arabic*.]
		\item $\approxcriterionfunction$ has bounded first-order derivatives on $\fatmanifold$.
		\item $\norminfk{\approxcriterionfunction - \criterionfunction}{1}$ is sufficiently small on $\fatmanifold$.
	\end{enumerate}
	Then, $\hausdorffdistance{\approxsolutionmanifold}{\solutionmanifold} = \bigoh(\norminfcriteriondiff)$.
\end{theorem}

We have introduced in Theorem~\ref{th:stability-theorem} a slight relaxation on the differentiability constraint for $\approxcriterionfunction$.
\citeauthor{Chen2022} supposes differentiability and bounds on $\rd$, whereas we allow for a narrower domain $\fatmanifold$.
This deviation is justified since hypotheses $\criterionhypothesesinline$ imply $\approxsolutionmanifold \subset \fattening{\solutionmanifold}{\fatteningexcess} \subset \fatmanifold$, where $\fatteningexcess < \fatteningconstant$.
Since \glspl*{pdf} typically vanish at infinity, it might be unfeasible to ask $\approxcriterionfunction = \curvaturefunctionalof{\kdefnh}$ to be differentiable everywhere.
This is the case for the eigenvalues~\eqref{eq:ordered-eigenvalues} in Proposition~\ref{prop:distinct-eigenvalues} below, where condition~\eqref{eq:eigenvalue-separability} would not hold if the infimum were taken over $\rd$.

Finally, putting all the pieces together, we get the following main result.

\begin{theorem}
	\label{th:bump-boundary-convergence}
	Assume the following:
	\begin{itemize}[leftmargin=2em]
		\renewcommand{\labelitemi}{\tiny$\blacklozenge$}
		\item Let $\curvaturefunctional$ be a curvature functional defined over $\dimension$-variate \glspl*{pdf} depending on their partial derivatives up to $\maxdifforder$-th order.
		      More formally, given a \gls*{pdf} $\testpdf$, we have $\curvaturefunctionalof{\testpdf} = \underlyingcurvature \comp \multivariatepartialpdf{\testpdf}$, for some $\functiondef{\underlyingcurvature}{\reals^{\numderivatives}}{\reals}$ and derivative index vectors $\indexedderivativevector{1}, \dots, \indexedderivativevector{\numderivatives} \in \positiveintegerspowdord{\maxdifforder}$.
		\item Let $\pdf$ be a \gls*{pdf} in $\classcontinous{\maxdifforder + \extrasmoothing}{\rd}$, for some $\extrasmoothing \in \{1, 2, \dots, \infty \}$, with all partial derivatives bounded up to $(\maxdifforder + \extrasmoothing)$-th order.
		      If $\extrasmoothing = 1$, further assume that the $(\maxdifforder + 1)$-th partial derivatives of $\pdf$ are either Hölder continuous with exponent $\holderexponent \in (0, 1]$ or uniformly continuous.
		\item Let $\kdefnh$ be the \gls*{kde} of $\pdf$ based on the Gaussian kernel.
		\item Let $\bandwidth \equiv \bandwidth_{\samplesize}$ converge to zero and satisfy $\limsamplesizeinf \inlineorderfraclog{2\maxdifforder + 2} = 0$.
	\end{itemize}
	Let the curvature bump boundary and its plug-in estimator, respectively, be
	\begin{multicols}{2}
		\noindent
		\begin{equation*}
			\genericbumpboundary
			=
			\curvaturerdzerolevelset{\pdf}
			\,,
		\end{equation*}
		\columnbreak
		\begin{equation*}
			\approxbumpboundary
			=
			\curvaturerdzerolevelset{\kdefnh}
			\,.
		\end{equation*}
	\end{multicols}
	\vspace*{-\multicolsep}
	Further, suppose that:
	\begin{itemize}[leftmargin=2em]
		\renewcommand{\labelitemi}{\tiny$\blacklozenge$}
		\item There exists $\fatteningconstant > 0$ such that $\underlyingcurvature \in \classcontinous{1}{\openset}$, for some open set $\openset \subset \reals^{\numderivatives}$ containing the images of $\fatbumpboundary$ under both $\multivariatepartialpdf{\pdf}$ and $\multivariatepartialpdf{\kdefnh}$ \gls*{as}
		\item There exists $\criteriongradientconstant > 0$ such that $\norm{\grad{\curvaturefunctionalof{\pdf}}(\x)} > \criteriongradientconstant$, for all $\x \in \fatbumpboundary$.
	\end{itemize}
	Then,
	\begin{equation*}
		\hausdorffdbetweenboundaries
		=
		\bigoh \left(
		\bandwidth^{\minbetweenextrasmoothingandtwo} + \sqrt{\orderfraclog{2\maxdifforder}}
		\right)
		\almostsureconverges
		0
		\,.
	\end{equation*}
\end{theorem}

The optimal bound is $\hausdorffdbetweenboundaries = \bigoh ([ \samplesize^{-1} \log \samplesize ]^{2 / (\dimension + 2\maxdifforder + 4)})$, achieved with $\bandwidth \sameorder [ \samplesize^{-1} \log \samplesize ]^{1 / (\dimension + 2\maxdifforder + 4)}$ ($\extrasmoothing \geq 2$).
The former coincides up to a logarithmic term with the optimum in \gls*{kdde} for $\maxdifforder$-th order partial derivatives according to the \textit{root mean integrated square error} criterion, which is $\bigoh(\samplesize^{-2 / (\dimension + 2\maxdifforder + 4)})$~\cite{Chacon2011}.

Theorem~\ref{th:bump-boundary-convergence} straightforwardly leads to bump boundary convergence results for the determinants and traces of the shape operator and the Hessian matrix.

\begin{example}
	Consider the Laplacian and Gaussian bumps~\eqref{eq:laplacian-bump} and~\eqref{eq:hessian-determinant-bump} for a bivariate \gls{pdf} $\functiondef{\pdf}{\reals^2}{\closedzeroinfinterval}$, with $\curvaturefunctionalof{\pdf}$ equal to, respectively,
	\begin{multicols}{2}
		\noindent
		\begin{equation*}
			\trace{\hessian{\pdf}}
			\equiv
			\laplacianof{\pdf}
			=
			\partialpdfxindex{1} + \partialpdfxindex{2}
			\,,
		\end{equation*}
		\columnbreak
		\begin{equation*}
			\determinant{\hessian{\pdf}}
			=
			\partialpdfxindex{1} \partialpdfxindex{2} - \left( \frac{\partial^2 \pdf}{\partialx{1}\partialx{2}} \right)^2
			\,.
		\end{equation*}
	\end{multicols}
	\vspace*{-\multicolsep}
	For the trace, the underlying derivative functional is $\underlyingcurvature(\derivativevar_1, \derivativevar_2) = \derivativevar_1 + \derivativevar_2$,
	considering $\indexedderivativevector{1} = (2, 0)$ and $\indexedderivativevector{2} = (0, 2)$.
	In turn, the functional is $\underlyingcurvature(\derivativevar_1, \derivativevar_2, \derivativevar_3) = \derivativevar_1 \derivativevar_2 - \derivativevar_3^2$ for the determinant, taking $\indexedderivativevector{1}$ and $\indexedderivativevector{2}$ as before plus $\indexedderivativevector{3} = (1, 1)$.
	In both cases, $\underlyingcurvature$ is an infinitely smooth function over $\openset = \reals^{\numderivatives}$, making every $\fatteningconstant > 0$ satisfy the requirement in Theorem~\ref{th:bump-boundary-convergence} without imposing additional hypotheses on the original \gls*{pdf} and its \gls*{kde}.
\end{example}

The case for the Hessian eigenvalues is more involved.
The functions $\eigenvalueofpdf{\genindexi}$ in~\eqref{eq:ordered-eigenvalues} are not generally $\rd$-differentiable.
To solve this differentiability issue, we will follow the standard assumption in~\citeauthor{Kato1995}'s book that, for every $\x \in \rd$, all the eigenvalues of $\hessian{\pdf}(\x)$ have multiplicity one~\cite[Theorem 5.16, p. 119]{Kato1995}.
We will ask for an even stronger hypothesis to ensure that all plug-in estimators $\eigenvalueofkdefnh{\genindexi}$ are eventually distinct everywhere for large $\samplesize$ \gls*{as}

\begin{proposition}
	\label{prop:distinct-eigenvalues}
	Let $\pdf$ be a \gls*{pdf} and let $\kdefnh$ be its \gls*{kde}.
	Let us assume that $\pdf$ and $\kdefnh$ satisfy all the conditions in Theorem~\ref{th:derivative-convergence} so that the second-order partial derivatives of $\pdf$ are consistently approximated with plug-in estimators.
	Let us call $\genericbumpboundary$ the bump boundary for the criterion function $\curvaturefunctional \equiv \eigenvalueofpdf{\genindexj}$, for some $\genindexj \in \{1, \dots, \dimension\}$.
	If there exists $\fatteningconstant > 0$ such that
	\begin{equation}
		\label{eq:eigenvalue-separability}
		\infdiffconsecutiveeig{\pdf} > 0
		\,,
	\end{equation}
	for all $\genindexi \in \{1, \dots, \dimension - 1\}$, then~\eqref{eq:eigenvalue-separability} also holds \gls*{as} for $\samplesize$ sufficiently large if we replace $\pdf$ by $\kdefnh$.
	In particular, both $\eigenvalueofpdf{\genindexj}$ and $\eigenvalueofkdefnh{\genindexj}$ are infinitely differentiable functions of the second-order partial derivatives of $\pdf$ and $\kdefnh$, respectively, on some neighbourhood $\fatbumpboundary$ \gls*{as} for $\samplesize$ sufficiently large.
\end{proposition}

\subsection{Inference}

In this section, we derive bootstrap inference for curvature bumps, following similar steps as in the scheme developed by~\citeauthor{Chen2017} for \gls*{pdf} level sets~\cite{Chen2017}.
To accommodate the required techniques, we will exclusively focus on curvature functionals $\curvaturefunctional$ deriving from the \gls*{pdf} Hessian $\hessianmatrix$.

\subsubsection{Inference scheme}

We will simplify the inference problem by targeting $\functiondef{\smoothedpdf}{\rd}{\closedzeroinfinterval}$, given by $\smoothedpdf(\x) = \smallexpectedvalueof{\kdefnhx}$, instead of $\pdf$, considering the bias negligible for a small $\bandwidth$.
There are compelling arguments favouring $\smoothedpdf$ against $\pdf$ for inference purposes (see~\cite[Section 2.2]{Chen2017} for a thorough discussion).

Let us call $\genericsmoothedbump$ the smoothed version of~\eqref{eq:abstract-bump} derived by replacing $\pdf$ with $\smoothedpdf$.
We will assume that $\genericsmoothedbump \subset \bumpdomain$, for some $\bumpdomain \subset \rd$, or at least that the inferential procedure focuses on $\genericsmoothedbump \cap \bumpdomain$.
Ideally, $\bumpdomain$ should be as \textit{small} as possible (hopefully $\bumpdomain \neq \rd$) so that the resulting confidence regions are \textit{efficient}.

Given $\oneminusconfidence \in (0, 1)$, a path for narrowing down a ($\confidence$)-level confidence region for $\genericsmoothedbump$ is constructing two sets
\begin{equation}
	\label{eq:bump-confidence-bounds}
	\begin{aligned}
		\genericupperconfidentbump
		 & =
		\confidentbumpset{-\margin}
		\\
		\genericlowerconfidentbump
		 & =
		\confidentbumpset{\margin}
	\end{aligned}
	\,,
\end{equation}
for some margin $\margin \in \closedzeroinfinterval$.
Note that $\genericlowerconfidentbump \subset \approxbump \subset \genericupperconfidentbump$, thus~\eqref{eq:bump-confidence-bounds} are set bounds for the $\approxbump$ in~\eqref{eq:approx-bump} approximating $\genericsmoothedbump$.
This \textit{vertical} scheme is similar to \citeauthor{Chen2017}'s second method for \gls*{pdf} level set inference~\cite{Chen2017} and a particular case of~\citeauthor{Mammen2013}'s universal approach~\cite{Mammen2013}.

Our inference results will establish conditions to ensure the previous set inequality eventually holds too with probability $\confidence$ when replacing $\approxbump$ with $\genericsmoothedbump$ while the set bounds~\eqref{eq:bump-confidence-bounds} draw nearer $\genericsmoothedbump$, namely
\begin{equation}
	\label{eq:asymptotic-validity-inference}
	\left\{
	\begin{aligned}
		\probabilityof{
			\genericlowerconfidentbumpapprox \subset \genericsmoothedbump \subset \genericupperconfidentbumpapprox
		}
		 & \geq
		\confidence + \littleohone
		\\
		\approxmargin
		 & =
		\littleohone
	\end{aligned}
	\right.
	\,,
\end{equation}
as $\samplesizegoestoinfty$, for some sequence $\{ \approxmargin \}_{\samplesize = 1}^{\infty}$.
The inference scheme~\eqref{eq:asymptotic-validity-inference} can be proven for all curvature bumps using Theorem~\ref{th:abstract-inference-indirect}.
From Section~\ref{sec:consistency}, it is an exercise to realize that, under the conditions in which~\eqref{eq:asymptotic-validity-inference} will hold, and with a few mild additional assumptions, the boundaries of the set bounds~\eqref{eq:bump-confidence-bounds} converge in the Hausdorff distance to $\boundaryof{\genericsmoothedbump}$.

In what follows, we will equivalently denote $\quantilefuncofrv{\randomvarx}{\p} \equiv \quantilefuncofp{\p}{\randomvarx}$ the $\p$-th quantile, $\p \in (0, 1)$, of the \gls*{rv} $\randomvarx$~\cite[p. 304]{vanderVaart1998}.

\begin{theorem}
	\label{th:abstract-inference-indirect}
	In the context described above, assume the following:
	\begin{enumerate}[label=\Roman*]
		\item \label{itm:bounding-rv} There exists a sequence of random variables $\{\marginrv\}_{\samplesize = 1}^{\infty}$ such that, for sufficiently large $\samplesize \in \naturals$, $\supreumumnormerrorfunctional{\curvaturefunctional} \equiv \superrorcurvature \leq \marginrv$ \gls*{as}
		      Let us further assume that $\rootsamplesize \marginrv$ converges weakly~\cite{vanderVaart1996} to some \gls*{rv} $\limitingmarginrv$ as $\samplesizegoestoinfty$, denoted by $\weaklyconverges{\rootsamplesize \marginrv}{\limitingmarginrv}$.
		      Suppose that $\limitingmarginrv$ has a continuous and strictly increasing \gls*{cdf}.
		\item \label{itm:margin-bounding-quantile} For each $\oneminusconfidence \in (0, 1)$, there is $\{ \margin \}_{\samplesize = 1}^{\infty}$ satisfying $\margin \geq \quantilefuncofrv{\marginrv}{\confidence}$, for all $\samplesize \in \naturals$, and $\limsamplesizeinf \margin = 0$.
		\item \label{itm:convergent-approx-margin} For each $\oneminusconfidence \in (0, 1)$, there is $\{ \approxmargin \}_{\samplesize = 1}^{\infty}$ satisfying $\smallabsval{\approxmargin - \margin} = \littleohinvsqrtn$ as $\samplesizegoestoinfty$.
	\end{enumerate}
	Then, for all $\oneminusconfidence \in (0, 1)$, the asymptotic validity of the inference scheme~\eqref{eq:asymptotic-validity-inference} holds.
\end{theorem}

The following sections will introduce theoretical results leading to bootstrap estimates $\approxmargin$ that can be feasibly computed in practice.

\citeauthor{Mammen2013}'s approach~\cite{Mammen2013} achieves a sharp asymptotic coverage probability $\confidence$ in~\eqref{eq:asymptotic-validity-inference}.
A key difference separating their proposal from~\citeauthor{Chen2017}'s and ours is that they manage to bootstrap from \pgls*{rv} that is a supremum over a neighbourhood of the level set, unlike $\supreumumnormerrorfunctional{\curvaturefunctional}$ in Theorem~\ref{th:abstract-inference-indirect}, which considers the whole $\bumpdomain$.
See~\cite{Qiao2019} for an overview of similar local strategies for level sets.
Based on that, \citeauthor{Mammen2013}'s method will generally be less conservative.

\subsubsection{Bootstrap outline}

The main point to fill the Theorem~\ref{th:abstract-inference-indirect} template is approximating the stochastic errors for second-order linear differential operators $\genericdiffop$
\begin{equation}
	\label{eq:inference-uniform-error}
	\supreumumnormerror{\genericdiffop}
	=
	\smallsupnormbumpdomain{\genericdiffop\kdefnh(\x) - \genericdiffop\smoothedpdf(\x)}
	\,,
\end{equation}
using bootstrap estimates
\begin{equation}
	\label{eq:inference-uniform-error-bootstrap}
	\bootstrapsupreumumnormerror{\genericdiffop}
	=
	\smallsupnormbumpdomain{\genericdiffop\bootstrapkdefnhof{\x} - \genericdiffop\conditionalkdefnhof{\x}}
	\,,
\end{equation}
where $\bootstrapkdefnhof{\funcarg}$ denotes the \gls*{kde} based on $\samplesize$ \gls*{iid} random variables $\bootstrapsamplevar_1, \dots, \bootstrapsamplevar_{\samplesize} \followdistribution \empiricalbootstrapprob$ of the empirical bootstrap probability measure $\empiricalbootstrapprob$ assigning equal masses $1/\samplesize$ to each component $\samplevarrealization_{\genindexi} \in \rd$ of a particular $\samplesize$-size \gls*{iid} realization $\originalsample = \{\samplevarrealization_1, \dots, \samplevarrealization_{\samplesize}\}$ from $\pdf$, and $\conditionalkdefnhof{\funcarg}$ is the realization of the \gls*{kde} based on $\originalsample$, i.e.,
\begin{multicols}{2}
	\noindent
	\begin{equation*}
		\bootstrapkdefnhof{\x}
		=
		\oneoversamplesize
		\sumitosamplesize
		\kernelh(\x - \bootstrapsamplevar_{\genindexi})
		\,,
	\end{equation*}
	\columnbreak
	\begin{equation*}
		\conditionalkdefnhof{\x}
		=
		\oneoversamplesize
		\sumitosamplesize
		\kernelh(\x - \samplevarrealization_{\genindexi})
		\,.
	\end{equation*}
\end{multicols}
\vspace*{-\multicolsep}
Assume that both~\eqref{eq:inference-uniform-error} and~\eqref{eq:inference-uniform-error-bootstrap} use the same kernel $\kernel$ everywhere.
Estimating confidence regions for curvature bumps will go through, directly or indirectly, approximating the \gls*{cdf} of~\eqref{eq:inference-uniform-error} with that of~\eqref{eq:inference-uniform-error-bootstrap}.

\subsubsection{Gaussian process approximation}

Lemma~\ref{lm:chernozhukov} below allows a \gls*{gp} approximation between the suprema~\eqref{eq:inference-uniform-error} and~\eqref{eq:inference-uniform-error-bootstrap}.
See~\cite{vanderVaart1996} for further knowledge about \glspl*{gp}.
The \textit{empirical process}~\cite{vanderVaart1996} on a sample $\nsample$ of \gls*{iid} $\dimension$-dimensional random variables indexed by a class $\measurableclass$ of measurable functions $\functiondef{\measurableclassmember}{\rd}{\reals}$ is defined as the functional $\empiricalprocess$ mapping a function $\measurableclassmember \in \measurableclass$ to the \gls*{rv}
\begin{equation*}
	\empiricalprocess(\measurableclassmember)
	=
	\frac{1}{\rootsamplesize}
	\sumitosamplesize
	\left(
	\measurableclassmember(\samplevar_{\genindexi})
	- \expectedvalueof{\measurableclassmember(\samplevar_{\genindexi})}
	\right)
	\,.
\end{equation*}
Lemma~\ref{lm:chernozhukov} invokes the \gls*{pm} and \gls*{vc}-type classes of functions.
We refer the reader to~\cite{vanderVaart1996} for the former and briefly define the latter, including the auxiliary Definition~\ref{def:covering-number}.

\begin{definition}[Covering number~\cite{vanderVaart1996}]
	\label{def:covering-number}
	Let $(\vectorspace, \norm{\funcarg})$ be a vector space with a seminorm and let $\coverableclass \subset \vectorspace$.
	We define the $\coveringepsilon$-\textit{covering number} of $\coverableclass$, denoted by $\coveringnumber(\coverableclass, \vectorspace, \coveringepsilon)$, as the minimum number of $\coveringepsilon$-balls of the form $\{ \x \in \vectorspace : \norm{\x - \y} < \coveringepsilon \}$, where $\y \in \vectorspace$, needed to cover $\coverableclass$.
\end{definition}

\begin{definition}[\acrshort{vc}-type class of functions~\cite{Chernozhukov2014}]
	Let $\measurableclass$ be a class of measurable functions $\functiondef{\measurableclassmember}{\rd}{\reals}$.
	Let $\envelopefunction$ be an \textit{envelope} function for $\measurableclass$, i.e., $\functiondef{\envelopefunction}{\rd}{\reals}$ measurable such that $\sup_{\measurableclassmember \in \measurableclass} \absoluteval{\measurableclassmember(\x)} \leq \envelopefunction(\x)$ for all $\x \in \rd$.
	An $\measurableclass$ class equipped with an envelope $\envelopefunction$ is called a \gls{vc}-type class if there exist $\coveringconstanta, \coveringconstantnu \in \openzeroinfinterval$ such that, for all $\coveringepsilon \in (0, 1)$,
	\begin{equation*}
		\sup_{\coveringprobmeasure} \coveringnumber(\measurableclass, \classlp{2}(\rd; \coveringprobmeasure), \coveringepsilon \norm{\envelopefunction}_{2, \coveringprobmeasure})
		\leq
		\left( \frac{\coveringconstanta}{\coveringepsilon} \right)^{\coveringconstantnu}
		\,,
	\end{equation*}
	where the supremum is taken over all finitely discrete probability measures $\coveringprobmeasure$ defined on $\rd$ and $\norm{\envelopefunction}_{2, \coveringprobmeasure} = ( \integralonrd \absoluteval{\envelopefunction}^2 d\coveringprobmeasure )^{1/2}$ is the seminorm of $\classlp{2}(\rd; \coveringprobmeasure)$.
\end{definition}

We will denote the Kolmogorov distance as $\kolmogorovdistance{\randomvarx}{\randomvary} = \sup_{\cdfarg \in \reals} \absoluteval{\cdfof{\randomvarx}(\cdfarg) - \cdfof{\randomvary}(\cdfarg)}$, where $\cdfof{\randomvarx}$ is the \gls*{cdf} of the \gls*{rv} $\randomvarx$.
Likewise, $\randomvarx \equaldistribution \randomvary$ will denote equality in distribution between the random variables.

\begin{lemma}[\citeauthor{Chernozhukov2014},~\citeyear{Chernozhukov2014}~\cite{Chernozhukov2014, Chen2015, Chen2016}]
	\label{lm:chernozhukov}
	Consider a sample $\nsample$ of \gls*{iid} random variables.
	Let $\vcclass$ be a \gls*{pm} and \gls*{vc}-type class of functions with constant envelope $\vcbound \in \openzeroinfinterval$.
	Let $\processstdv \in \openzeroinfinterval$ be such that $\supovervcclass{\expectedvalueof{\vcclassmember(\samplevar_1)^2}} \leq \processstdv^2 \leq \vcbound^2$.
	Let $\gaussianprocess$ be a centred tight \gls*{gp} with sample paths on the space of bounded functions $\classbounded{\vcclass}$, and with covariance function
	\begin{equation}
		\label{eq:covariance-gaussian-process}
		\covariance{\gaussianprocess(\vcclassmember_1)}{\gaussianprocess(\vcclassmember_2)}
		=
		\expectedvalueof{\vcclassmember_1(\samplevar_1) \vcclassmember_2(\samplevar_1)}
		-
		\expectedvalueof{\vcclassmember_1(\samplevar_1)}
		\expectedvalueof{\vcclassmember_2(\samplevar_1)}
		\,,
	\end{equation}
	for $\vcclassmember_1, \vcclassmember_2 \in \vcclass$.
	Then, there exists \pgls*{rv} $\gaussianprocessrv \equaldistribution \supovervcclassof{\gaussianprocess}$ such that, for all $\chernozhukovfreeconst \in (0, 1)$ and $\samplesize$ sufficiently large,
	\begin{equation*}
		\probabilityof{
			\absoluteval{\supovervcclassof{\empiricalprocess} - \gaussianprocessrv}
			>
			\chernozhukovfirstconst
			\frac{\vcbound^{1/3} \processstdv^{2/3} \log^{2/3} \samplesize}
			{\chernozhukovfreeconst^{1/3} \samplesize^{1/6}}
		}
		\leq
		\chernozhukovsecondtconst \chernozhukovfreeconst
		\,,
	\end{equation*}
	where $\empiricalprocess$ is based on $\nsample$, and $\chernozhukovfirstconst, \chernozhukovsecondtconst$ are universal constants.
\end{lemma}

If we apply Lemma~\ref{lm:chernozhukov} to~\eqref{eq:inference-uniform-error}, we get the following result.

\begin{theorem}
	\label{th:gaussian-approximation-supnorm}
	Let $\genericdiffop$ denote any linear $\difforder$-th order differential operator.
	Let $\kernel \in \classcontinous{\difforder}{\rd}$ be a kernel with bounded $\difforder$-th derivatives.
	Further, suppose that the class
	\begin{equation}
		\label{eq:vc-class-kernels}
		\kernelsvcclass
		=
		\left \{
		\lambdafunction{\y \in \rd}{%
			\indexeddiffoperator\kernel \left( \frac{\x - \y}{\bandwidth} \right)
		}
		: \x \in \bumpdomain, \bandwidth > 0, \derivativevector \in \positiveintegerspowd, \derivativevectororder = \difforder
		\right \}
	\end{equation}
	is \gls*{vc}-type.
	Let $\bandwidth \equiv \bandwidth_{\samplesize}$ be a sequence with $\bandwidth \in (0, 1)$ and $\bandwidth^{-(\dimension + \difforder)} = \bigoh(\log \samplesize)$.
	Moreover, let $\gaussianprocess$ be a \gls*{gp} with the same properties as in Lemma~\ref{lm:chernozhukov} and indexed by
	\begin{equation}
		\label{eq:vc-class-scaled-kernels}
		\vcclassh = \left \{
		\lambdafunction{\y \in \rd}{%
			\frac{1}{\sqrt{\bandwidth^{\dimension + \difforder}}}
			\genericdiffop\kernel \left( \frac{\x - \y}{\bandwidth} \right)
		}
		: \x \in \bumpdomain
		\right \}
		\,.
	\end{equation}
	Then, there exists $\gaussianprocessrvh \equaldistribution \supovervcclasshof{\gaussianprocess}$ such that, for $\samplesize$ sufficiently large,
	\begin{equation*}
		\kolmogorovdistance{\scaling \supreumumnormerror{\genericdiffop}}{\gaussianprocessrvh}
		=
		\bigoh \left( \bootstrapconvergencerateorder \right)
		\tendstoassamplesizegoestoinfty
		0
		\,.
	\end{equation*}
	Moreover, if we fix $\bandwidth \in (0, 1)$ and define $\normalizedgaussianprocessrvh = \gaussianprocessrvh / \sqrt{\bandwidth^{\dimension + \difforder}}$, then $\rootsamplesize \ \supreumumnormerror{\genericdiffop}$ converges in probability to $\normalizedgaussianprocessrvh$, denoted $\convergesinprob{\rootsamplesize \ \supreumumnormerror{\genericdiffop}}{\normalizedgaussianprocessrvh}$, as $\samplesizegoestoinfty$.
\end{theorem}

A similar result establishes the asymptotic distribution for~\eqref{eq:inference-uniform-error-bootstrap}.

\begin{theorem}
	\label{th:gaussian-approximation-supnorm-bootstrap}
	Let $\genericdiffop$ denote any linear $\difforder$-th order differential operator.
	Let $\kernel \in \classcontinous{\difforder}{\rd}$ be a kernel with bounded $\difforder$-th derivatives.
	Further, suppose that the class $\kernelsvcclass$ in~\eqref{eq:vc-class-kernels} is \gls*{vc}-type.
	Moreover, let $\gaussianprocessbootstrap$ be a \gls*{gp} with the same properties as in Lemma~\ref{lm:chernozhukov}, indexed by $\vcclassh$ as in~\eqref{eq:vc-class-scaled-kernels}, and with covariance
	\begin{equation*}
		\covariance{\gaussianprocessbootstrap(\vcclassmember_1)}{\gaussianprocessbootstrap(\vcclassmember_2)}
		=
		\oneoversamplesize
		\sumitosamplesize
		\vcclassmember_1(\samplevarrealization_{\genindexi})
		\vcclassmember_2(\samplevarrealization_{\genindexi})
		-
		\frac{1}{\samplesize^2}
		\prod_{\genindexj = 1}^{2}
		\left(
		\sumitosamplesize
		\vcclassmember_{\genindexj}(\samplevarrealization_{\genindexi})
		\right)
		\,,
	\end{equation*}
	where $\samplevarrealization_{\genindexi}$ is the $\genindexi$-th observation in $\originalsample$.
	If $\bandwidth \equiv \bandwidth_{\samplesize}$ is a sequence with $\bandwidth \in (0, 1)$ and $\bandwidth^{-(\dimension + \difforder)} = \bigoh(\log \samplesize)$, then there exists $\gaussianprocessrvnh \equaldistribution \supovervcclasshof{\gaussianprocessbootstrap}$ such that, for $\samplesize$ sufficiently large,
	\begin{equation*}
		\kolmogorovdistance{\scaling \bootstrapsupreumumnormerror{\genericdiffop}}{\gaussianprocessrvnh}
		=
		\bigoh \left( \bootstrapconvergencerateorder \right)
		\tendstoassamplesizegoestoinfty
		0
		\,.
	\end{equation*}
\end{theorem}

Theorem~\ref{th:gaussian-approximation-supnorm-bootstrap} holds for \textit{any} observations $\originalsample$.
The applicability of this theorem relies on the assumption that $\weaklyconverges{\gaussianprocessrvnh}{\gaussianprocessrvh}$ \gls*{as}
This connection crystallises in the following result, which can be straightly derived from Theorem~\ref{th:gaussian-approximation-supnorm} and Theorem~\ref{th:gaussian-approximation-supnorm-bootstrap}.

\begin{theorem}
	\label{th:bootstrap-approximation}
	Let $\genericdiffop$ denote any linear $\difforder$-th order differential operator.
	Let $\kernel \in \classcontinous{\difforder}{\rd}$ be a kernel with bounded $\difforder$-th derivatives.
	Further, suppose that the class $\kernelsvcclass$ in~\eqref{eq:vc-class-kernels} is \gls*{vc}-type.
	Let $\bandwidth \equiv \bandwidth_{\samplesize}$ be a sequence with $\bandwidth \in (0, 1)$ and $\bandwidth^{-(\dimension + \difforder)} = \bigoh(\log \samplesize)$.
	Moreover, let us write $\diffcdfgaussiansup \equiv \kolmogorovdistance{\gaussianprocessrvnh}{\gaussianprocessrvh}$, where $\gaussianprocessrvh$ and $\gaussianprocessrvnh$ are as in Theorem~\ref{th:gaussian-approximation-supnorm} and Theorem~\ref{th:gaussian-approximation-supnorm-bootstrap}, respectively.
	Let us allow $\originalsample$ to vary as a random sample from the \gls*{pdf} $\pdf$ underlying the covariance structure~\eqref{eq:covariance-gaussian-process} of $\gaussianprocessrvh$.
	Further, suppose that $\diffcdfgaussiansup = \littleohone$ \gls*{as} under the previous hypotheses on $\bandwidth$.
	Then, for $\samplesize$ sufficiently large,
	\begin{equation*}
		\kolmogorovdistance{\scaling \bootstrapsupreumumnormerror{\genericdiffop}}{\gaussianprocessrvh}
		=
		\bigoh \left( \diffcdfgaussiansup + \bootstrapconvergencerateorder \right)
		\almostsureconverges
		0
		\,.
	\end{equation*}
\end{theorem}

We can state sufficient conditions under which $\diffcdfgaussiansup$ would converge to zero \gls*{as}
Corollary~\ref{cor:bootstrap-approximation} gathers all the previous findings in an easy, ready-to-use form.

\begin{corollary}
	\label{cor:bootstrap-approximation}
	In the hypotheses of Theorem~\ref{th:bootstrap-approximation}, if we further take a constant $\bandwidth$ and define $\normalizedgaussianprocessrvh = \gaussianprocessrvh / \sqrt{\bandwidth^{\dimension + \difforder}}$, then
	\begin{equation*}
		\begin{aligned}
			\kolmogorovdistance{\rootsamplesize \ \supreumumnormerror{\genericdiffop}}{\normalizedgaussianprocessrvh}
			 & \tendstoassamplesizegoestoinfty
			0
			\\
			\kolmogorovdistance{\rootsamplesize \ \bootstrapsupreumumnormerror{\genericdiffop}}{\normalizedgaussianprocessrvh}
			 & \almostsureconverges
			0
		\end{aligned}
		\,.
	\end{equation*}
	In particular, $\weaklyconverges{\rootsamplesize \ \bootstrapsupreumumnormerror{\genericdiffop}}{\normalizedgaussianprocessrvh}$ \gls*{as}
	Moreover, $\convergesinprob{\rootsamplesize \ \supreumumnormerror{\genericdiffop}}{\normalizedgaussianprocessrvh}$.
	Finally, $\normalizedgaussianprocessrvh$ has a continuous and strictly increasing \gls*{cdf}.
\end{corollary}

\subsubsection{Inference for curvature bumps}

The results from the previous section hold the key to ensuring~\eqref{eq:asymptotic-validity-inference} for curvature bumps.

\vspace{1em}
\paragraph{\textit{Laplacian bumps}}

Theorem~\ref{th:inference-laplacian-bumps} straightly follows from Corollary~\ref{cor:bootstrap-approximation} and Theorem~\ref{th:abstract-inference-indirect}.

\begin{theorem}
	\label{th:inference-laplacian-bumps}
	Let us fix $\bandwidth \in (0, 1)$.
	Let $\bootstrapsupreumumnormerror{\funcarg}$ be as defined in~\eqref{eq:inference-uniform-error-bootstrap} with \gls*{kde} based on a kernel $\kernel \in \classcontinous{2}{\rd}$ with bounded second derivatives.
	Taking $\difforder = 2$, suppose that the class $\kernelsvcclass$ in~\eqref{eq:vc-class-kernels} is \gls*{vc}-type.
	For any $\oneminusconfidence \in (0, 1)$, define the margin $\approxmargin = \quantilefuncofrv{\bootstrapsupreumumnormerror{\laplacian}}{\confidence}$.
	Then, for all $\oneminusconfidence \in (0, 1)$, the asymptotic validity of the inference scheme~\eqref{eq:asymptotic-validity-inference} holds \gls*{as} for the smoothed version of the Laplacian bump~\eqref{eq:laplacian-bump}.
\end{theorem}

\vspace{1em}
\paragraph{\textit{Concave bumps \& convex dips}}

Concave bumps and convex dips are more involved.
To obtain a parallel result to Theorem~\ref{th:inference-laplacian-bumps}, we will borrow the \gls*{tvar} concept from financial risk management~\cite{Dhaene2006}.
The \gls*{tvar} at level $\p \in (0, 1)$ of \pgls*{rv} $\randomvarx$ is defined as
\begin{equation*}
	\tvar{\p}{\randomvarx}
	=
	\frac{1}{1 - \p}
	\int_{\p}^1
	\quantilefuncofp{\q}{\randomvarx}
	\ d\q
	\,.
\end{equation*}
The \gls*{tvar} is utilised to \textit{aggregate} risks governed by an \textit{unknown} dependence structure, for it satisfies $\tvar{\p}{\randomvarx} \geq \quantilefuncofrv{\randomvarx}{\p}$ and is sub-additive~\cite{Dhaene2006}.
Contrary to quantiles, weak convergence does not guarantee \gls*{tvar} convergence.
Lemma~\ref{lm:tvar-convergence} requires the random variables to be \gls*{aui}~\cite[p. 17]{vanderVaart1998}.

\begin{lemma}
	\label{lm:tvar-convergence}
	Let $\{\randomvarx_{\genindexn}\}_{\genindexn = 1}^{\infty}$ be an \gls*{aui} sequence of random variables satisfying $\weaklyconverges{\randomvarx_{\genindexn}}{\randomvarx}$ for some \gls*{rv} $\randomvarx$ with a strictly increasing \gls*{cdf}.
	Then, $\limninf \tvar{\p}{\randomvarx_{\genindexn}} = \tvar{\p}{\randomvarx}$ for all $\p \in (0, 1)$, being the limit finite.
\end{lemma}

Then, Lemma~\ref{lm:tvar-convergence} allows proving the main result.

\begin{theorem}
	\label{th:inference-eigenvalue-bumps}
	Let us fix $\bandwidth \in (0, 1)$.
	Let $\supreumumnormerror{\funcarg}$ and $\bootstrapsupreumumnormerror{\funcarg}$ be as defined in~\eqref{eq:inference-uniform-error} and~\eqref{eq:inference-uniform-error-bootstrap} with \gls*{kde} based on the same kernel $\kernel \in \classcontinous{2}{\rd}$ with bounded second derivatives.
	Taking $\difforder = 2$, suppose that the class $\kernelsvcclass$ in~\eqref{eq:vc-class-kernels} is \gls*{vc}-type.
	For any $\oneminusconfidence \in (0, 1)$, define the margin
	\begin{equation*}
		\approxmargin
		=
		\sumitod
		\sum_{\genindexj = 1}^{\dimension}
		\tvar{\confidence}{\bootstrapsupnormderivativesij}
		\,,
	\end{equation*}
	where $\secondorderdifffunctional{\genindexi}{\genindexj}$ denotes second-order partial differentiation in the $\genindexi$ and $\genindexj$ variables.
	Moreover, let us assume the following:
	\begin{enumerate}
		\item \label{itm:continuous-cdf} Letting $\normalizedgaussianprocessrvh[\secondorderdifffunctional{\genindexi}{\genindexj}]$ be the \gls*{rv} such that $\convergesinprob{\rootsamplesize \ \supnormsecondderivativesij}{\sumsupgaussians}$, the sum \gls*{rv} $\limitingmarginrv = \sumijtod \normalizedgaussianprocessrvh[\secondorderdifffunctional{\genindexi}{\genindexj}]$ has a continuous and strictly increasing \gls*{cdf}.
		      \vspace{1ex}
		\item \label{itm:aui-errors} For each pair $(\genindexi, \genindexj)$, we have:
		      \vspace{1ex}
		      \begin{itemize}[leftmargin=2em]
			      \renewcommand{\labelitemi}{\tiny$\blacklozenge$}
			      \item $\{ \rootsamplesize \ \supnormsecondderivativesij \}_{\samplesize = 1}^{\infty}$ is \gls*{aui}
			      \item $\{ \rootsamplesize \ \bootstrapsupnormderivativesij \}_{\samplesize = 1}^{\infty}$ is \gls*{aui} \gls*{as}
		      \end{itemize}
	\end{enumerate}
	\vspace{1ex}
	Then, for all $\oneminusconfidence \in (0, 1)$, the asymptotic validity of the inference scheme~\eqref{eq:asymptotic-validity-inference} holds \gls*{as} for the smoothed version of the concave bump~\eqref{eq:concave-bump} and the convex dip~\eqref{eq:convex-bump}.
\end{theorem}

The assumptions (\ref{itm:continuous-cdf}) and (\ref{itm:aui-errors}) seem natural.
Hypothesis (\ref{itm:continuous-cdf}) asks a sum of nonnegative random variables with continuous and strictly increasing \glspl{cdf} to have a continuous and strictly increasing \gls{cdf} too, which should be valid except in pathological cases.
Similarly, knowing both sequences in hypothesis (\ref{itm:aui-errors}) converge weakly, being \gls{aui} amounts to the convergence of their expectations~\cite[Theorem 2.20]{vanderVaart1998}.

\vspace{1em}
\paragraph{\textit{Gaussian bumps}}

A similar result to Theorem~\ref{th:inference-eigenvalue-bumps} holds for Gaussian bumps~\eqref{eq:hessian-determinant-bump}.

\begin{theorem}
	\label{th:inference-gaussian-bumps}
	Consider the same hypotheses in Theorem~\ref{th:inference-eigenvalue-bumps} in the case $\dimension = 2$.
	Assume a Gaussian kernel $\kernel$.
	Further, assume that the true \gls*{pdf} $\pdf$ is bounded.
	Let $\constant$ be a constant such that $\constant > (\pi\bandwidth^4)^{-1}$.
	For any $\oneminusconfidence \in (0, 1)$, define the margin
	\begin{equation*}
		\approxmargin
		=
		\constant \
		\sumijtotwo
		\tvar{\confidence}{\bootstrapsupnormderivativesij}
		\,.
	\end{equation*}
	Then, for all $\oneminusconfidence \in (0, 1)$, the asymptotic validity of the inference scheme~\eqref{eq:asymptotic-validity-inference} holds \gls*{as} for the smoothed version of the Gaussian bump~\eqref{eq:hessian-determinant-bump}.
\end{theorem}

		\section{Application}

\label{sec:application}

We will explore a \textit{sports analytics} application for $\dimension = 2$ in the \gls*{nba}.
See the \supplement{} for additional applications with $\dimension \in \{1, 3\}$ in two American leagues: the \gls*{nfl} and the \gls*{mlb}.
Each player and team has its own style, a form of DNA.
Following the biological analogy, if a single gene activates a trait in natural DNA, even minor bumps in data may reveal essential features.

All three sports applications are representative of the use of kernel methods for \gls*{eda}.
Moreover, our proposal has a marked visual intent, thus excelling in low dimensions.
In this context, the \textit{curse of dimensionality} that harms kernel methods, demanding larger sample sizes to retain precision, becomes less relevant~\cite[Section 2.8]{Chacon2018}.

\vspace{1em}
\paragraph{\textbf{Bivariate made shots in the NBA}}

Most people are familiar with basketball's \gls*{3pl}, behind which a made shot earns not two but three points.
Sports analytics have demonstrated that attempting more of these shots is well worth the risk, given the increased efficiency of three-point shooters.
This trend has recently changed the basketball landscape, especially in the \gls*{nba}.

\citeauthor{Chacon2020} exemplified univariate multimodality with shooting distances to the basket in the \gls*{nba}~\cite{Chacon2020}.
We could see that the highest mode in a \gls*{pdf} model of all shots for the 2014-2015 season peaked beyond the \gls*{3pl}.
Looking at shots from a bivariate perspective will reveal the \gls*{3pl} not as two separate modes but as a \textit{ridge}~\cite{Chacon2018}.

We will examine bumps from shot data by the three best scorers in the 2015-2016 \gls*{nba} season: Stephen Curry, James Harden and Kevin Durant.
\figurename~\ref{fig:nba-surface} and \figurename~\ref{fig:nba-scatter} present different perspectives on concave and Laplacian bumps.
Setting the near-the-rim shots aside, the three players have different shooting DNAs.
Stephen Curry (\figurename~\ref{fig:scatter-curry}) operates beyond the \gls*{3pl}, covering the entire ridge.
He also demonstrates good range with even some half-court shots.
However, he barely uses the mid-range area.
His shooting patterns are mostly symmetrical.
James Harden (\figurename~\ref{fig:scatter-harden}) has similar trends to Curry's.
He almost covers the \gls*{3pl} while leaning towards some mid-range areas without half-court shots.
Some notable asymmetries are present.
Kevin Durant (\figurename~\ref{fig:scatter-durant}) has a more balanced game between mid and long shots.
He shoots facing the basket mainly, with lower usage of lateral shots.

\figurename~\ref{fig:nba-confidence-sets} complements the previous figures with confidence sets.
As refers to concave bumps, a wholly or partially ring-shaped area around the basket can be excluded with confidence for the three players.
Apart from the shots near the rim, we cannot find other spots likely contained in the concave bumps.
Regarding Laplacian bumps, the lower bound confidence sets become more relevant, even far apart from the rim.
For Curry, up to four high-confidence spots appear beyond the \gls*{3pl}, including the left-field corner; for Harden, the number of outside high-confidence spots decreases to two, while for Durant, there is only one.

\renewcommand{\subfigurewidthtwoandone}{0.4\textwidth}
\renewcommand{\graphicswidthtwoandone}{\textwidth}

\begin{figure}
	\centering
	\begin{subfigure}[b]{\columnwidth}
		\centering
		\includegraphics[width=\subfigurewidthtwoandone]{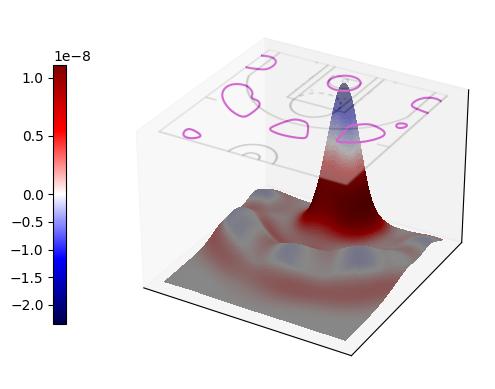}%
		\hspace{1 ex}
		\includegraphics[width=\subfigurewidthtwoandone]{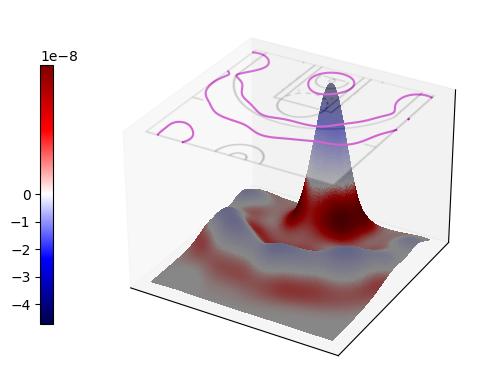}
		\caption{Stephen Curry}
	\end{subfigure}
	\begin{subfigure}[b]{\columnwidth}
		\centering
		\includegraphics[width=\subfigurewidthtwoandone]{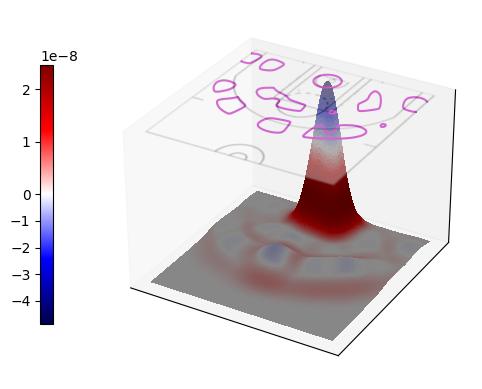}%
		\hspace{1 ex}
		\includegraphics[width=\subfigurewidthtwoandone]{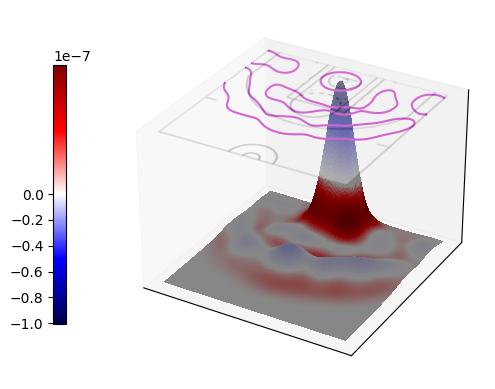}
		\caption{James Harden}
	\end{subfigure}
	\begin{subfigure}[b]{\columnwidth}
		\centering
		\includegraphics[width=\subfigurewidthtwoandone]{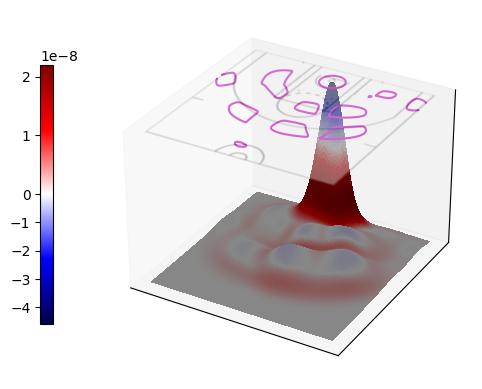}%
		\hspace{1 ex}
		\includegraphics[width=\subfigurewidthtwoandone]{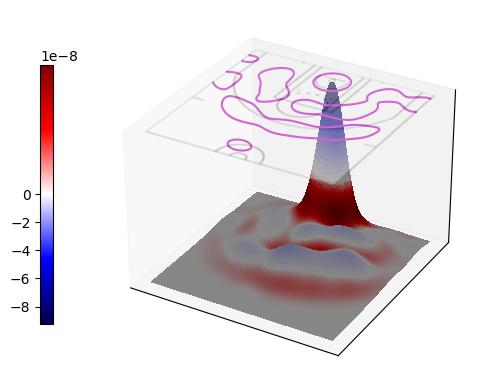}
		\caption{Kevin Durant}
	\end{subfigure}
	\caption{
		Concave and Laplacian bumps for Stephen Curry, James Harden and Kevin Durant.
		The three sub-figures have the same structure.
		On the left are concave bumps~\eqref{eq:concave-bump}; on the right are Laplacian bumps~\eqref{eq:laplacian-bump}.
		On either side, the two-dimensional surface is the fitted \gls*{kde} \gls*{pdf}.
		The area colour refers to the curvature functional value at each point.
		The bump boundaries appear as lines on a flat basketball court at the top.
	}
	\label{fig:nba-surface}
\end{figure}

\renewcommand{\subfigurewidthtwoandone}{0.4\textwidth}
\renewcommand{\graphicswidthtwoandone}{\textwidth}

\begin{figure}
	\centering
	\begin{subfigure}[b]{\subfigurewidthtwoandone}
		\centering
		\includegraphics[width=\graphicswidthtwoandone]{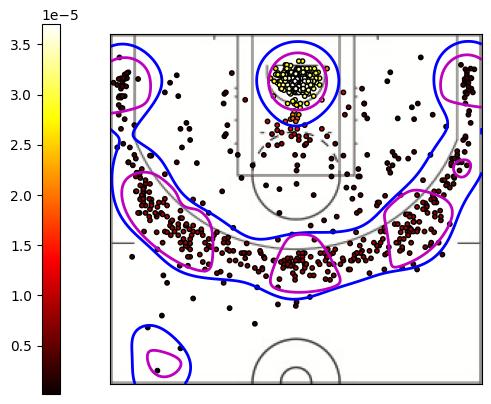}
		\caption{Stephen Curry}
		\label{fig:scatter-curry}
	\end{subfigure}
	\hspace{1 ex}
	\begin{subfigure}[b]{\subfigurewidthtwoandone}
		\centering
		\includegraphics[width=\graphicswidthtwoandone]{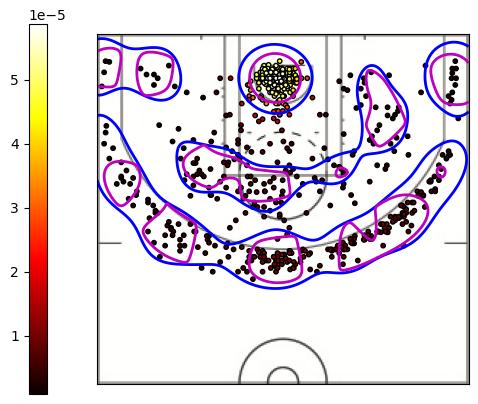}
		\caption{James Harden}
		\label{fig:scatter-harden}
	\end{subfigure}
	\begin{subfigure}[b]{\subfigurewidthtwoandone}
		\centering
		\vspace{3 ex}
		\includegraphics[width=\graphicswidthtwoandone]{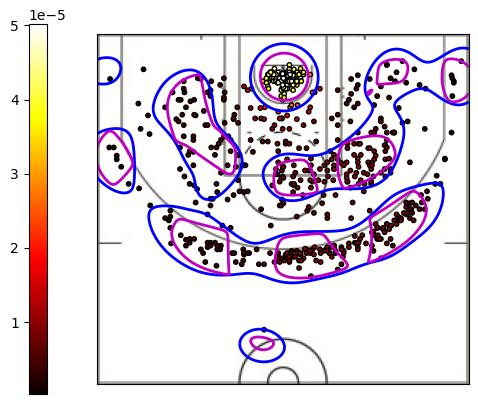}
		\caption{Kevin Durant}
		\label{fig:scatter-durant}
	\end{subfigure}
	\caption{
		Shot scatter data with concave and Laplacian bumps for Stephen Curry, James Harden and Kevin Durant.
		The three sub-figures have the same structure.
		Each point corresponds to a made shot location.
		The number of observations is 804, 710 and 698 for Stephen Curry, James Harden and Kevin Durant.
		The lines represent bump boundaries: magenta for concave bumps~\eqref{eq:concave-bump}; blue, Laplacian bumps~\eqref{eq:laplacian-bump}.
		The colour of the dots in the scatter plot conveys the value of the \gls*{kde} \gls*{pdf} at each point.
	}
	\label{fig:nba-scatter}
\end{figure}

\renewcommand{\subfigurewidthtwoandone}{0.4\textwidth}
\renewcommand{\graphicswidthtwoandone}{0.7\textwidth}

\begin{figure}
	\centering
	\begin{subfigure}[b]{\graphicswidthtwoandone}
		\centering
		\includegraphics[width=\subfigurewidthtwoandone]{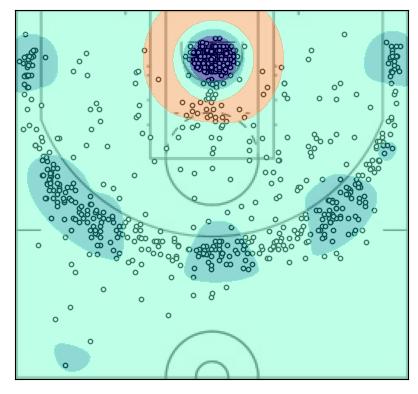}%
		\hspace{1 cm}
		\includegraphics[width=\subfigurewidthtwoandone]{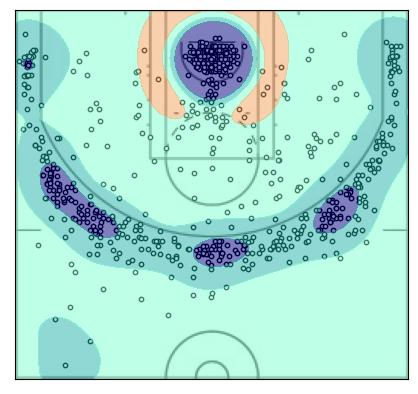}
		\caption{Stephen Curry}
	\end{subfigure}
	\begin{subfigure}[b]{\graphicswidthtwoandone}
		\centering
		\includegraphics[width=\subfigurewidthtwoandone]{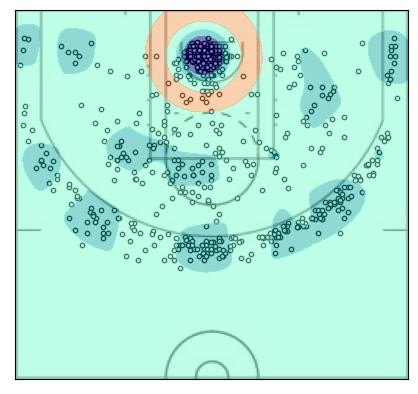}%
		\hspace{1 cm}
		\includegraphics[width=\subfigurewidthtwoandone]{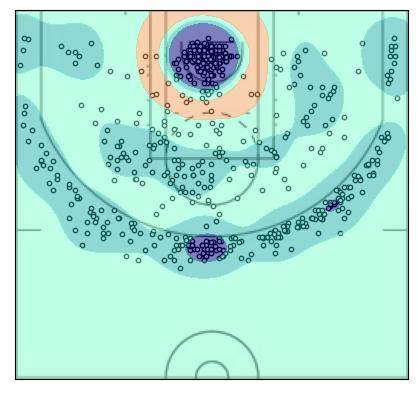}
		\caption{James Harden}
	\end{subfigure}
	\begin{subfigure}[b]{\graphicswidthtwoandone}
		\centering
		\includegraphics[width=\subfigurewidthtwoandone]{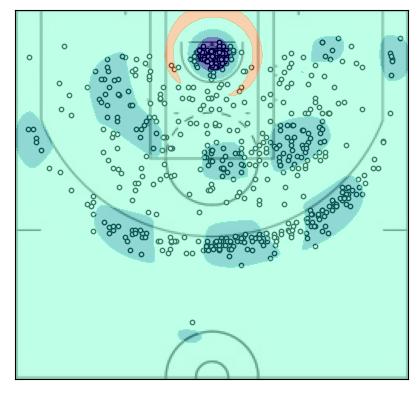}%
		\hspace{1 cm}
		\includegraphics[width=\subfigurewidthtwoandone]{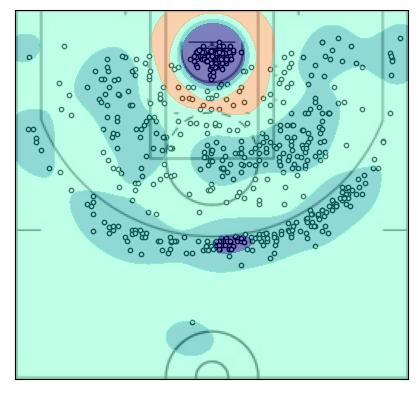}
		\caption{Kevin Durant}
	\end{subfigure}
	\caption{
		Confidence sets for Stephen Curry, James Harden and Kevin Durant's bumps.
		The three sub-figures have the same structure.
		On the left, 90\%-confidence sets for concave bumps~\eqref{eq:concave-bump}; on the right, 90\%-confidence sets for Laplacian bumps~\eqref{eq:laplacian-bump}.
		The confidence margins are based on 200 bootstrap samples, each with the same resample size as the original one.
		On either side, the area colours convey the same meaning.
		The non-blue \textit{sandy} areas fall \textit{outside} the confidence set bounds; the blue-coloured areas lie \textit{inside} the confidence region.
		The darkest blue corresponds to the lower bound confidence set: a set that is likely contained in the bump.
		The remaining blue areas cover the upper bound confidence set: a set that likely contains the bump.
		Finally, the mid-light blue colour points out the estimated bump.
	}
	\label{fig:nba-confidence-sets}
\end{figure}

		\section{Discussion}

\label{sec:discussion}

Our curvature \gls*{bh} methodology represents the next step in density \gls*{bh} techniques, a path opened by~\citeauthor{Good1980}~\cite{Good1980} and consolidated with~\citeauthor{Hyndman1996}~\cite{Hyndman1996} and~\citeauthor{Chaudhuri1999}~\cite{Chaudhuri1999}.
Rather than sticking to a purely probabilistic view on \glspl*{pdf}, our proposal thrives on sound geometry principles that have produced good results in applied areas like image processing~\cite{Haralick1992}.

Our work strongly relies on \gls*{kdde}, continuing the exploration of applications for higher-order partial derivatives of the \gls*{pdf}~\cite{Chacon2013}.
On the other hand, we bring to curvature \gls*{bh} some cutting-edge techniques for level set estimation and inference that extend the pointwise-oriented initial works by~\citeauthor{Godtliebsen2002}~\cite{Godtliebsen2002} and~\citeauthor{Duong2008}~\cite{Duong2008}.

The presented curvature framework shows great applicability from a theoretical standpoint.
Under mild assumptions, the mean curvature, Laplacian and Gaussian bumps are consistent with affordable convergence rates.
The confidence regions for Laplacian bumps are also asymptotically valid and consistent.
The cases for Gaussian bumps (inference), concave bumps and convex dips (consistency and inference) are slightly more technical.
Notwithstanding, pathological cases should not often appear in practice.

The \gls*{nba} application shows promise for \gls*{eda} and \textit{clustering}.
\figurename~\ref{fig:scatter-curry} presents a most pleasing result, identifying the \gls*{3pl} area and the most relevant shooting spots.
Both bumps are valuable and combine to produce insightful visualizations.
Comparing the pictures in \figurename~\ref{fig:nba-scatter}, we see that curvature bumps capture the players' rich shooting DNAs.
Despite the ultimately unavoidable threat of the \textit{curse of dimensionality} in \gls*{kde} settings~\cite{Chacon2018}, the relatively small sample sizes did not detract from the accuracy of the results.

Our methodology's apparent least impressive achievement is confidence regions despite asymptotic guarantees.
In \figurename~\ref{fig:nba-confidence-sets}, the upper-bound confidence sets tend to be conservative.
This was not wholly unexpected, as~\citeauthor{Chen2017} warned~\cite{Chen2017}.
The margin is especially coarse for the concave bumps.
In practice, we can mitigate this effect by splitting the bump and calculating the margin over smaller domains, employing a pilot estimation for guidance.
Nonetheless, further research following~\citeauthor{Mammen2013}'s universal approach~\cite{Mammen2013} should yield even better results.

		\ifnum\blind=0
			\section*{Acknowledgements}

The first author's research has been supported by the MICINN grant PID2019-109387GB-I00 and the Junta de Extremadura grant GR21044.
The second author would like to thank Amparo Ba\'{i}llo Moreno for her advice as a doctoral counsellor at the Autonomous University of Madrid.
Finally, we thank two anonymous reviewers for their helpful comments.

		\fi
	\end{refsection}

	\renewcommand*{\mkbibcompletename}[1]{\textsc{#1}}
	\printbibliography[section=1, filter=references, title={References}, sorting=nyt]
	\renewcommand*{\mkbibcompletename}[1]{#1}

	\ifnum\material=1
		\begin{refsection}
			\cleardoublepage

\pagenumbering{roman}
\setcounter{page}{1}
\setcounter{equation}{0}
\renewcommand{\theequation}{S\arabic{equation}}

\part*{Supplementary material}

The following \textit{appendices} are provided as \textit{supplementary material} to the manuscript \textit{\MyTitle}.
Theorem and equation numbers refer to the manuscript.
The numbers for new equations are prefixed.
References are included at the end.
Notation and acronyms are reused from the manuscript.

Appendix~\ref{sec:proofs} contains proof of the results in the manuscript.
Appendix~\ref{sec:curvature} introduces and motivates the geometrical concepts behind the curvature bumps in Section~\ref{sec:methods}.
Appendix~\ref{sec:extra-applications} illustrates the \gls*{bh} methodology for $\dimension \in \{1, 3\}$ with data from the \gls*{nfl} and the \gls*{mlb}.
Finally, Appendix~\ref{sec:computing} adds some comments on the underlying algorithms and data and acknowledges computational resources.

			\appendix

			\section{Proofs}

\label{sec:proofs}

\subsection{Consistency}

\begin{proof}[Proof of Lemma~\ref{lm:bias-order}]
	Since the partial derivatives of $\pdf$ and $\kernel$ are bounded and vanish at infinity, respectively, integration by parts and a change of variables yields
	\begin{align}
		\smallexpectedvalueof{\partialkdefnhx}
		 & =
		\integralonrd \indexeddiffoperator[\kernelh(\funcarg - \y)](\x) \ \pdf(\y) \ d\y
		\nonumber                     \\
		 & =
		\integralonrd (-1)^{\derivativevectororder} \ \indexeddiffoperator[\kernelh(\x - \funcarg)](\y) \ \pdf(\y) \ d\y
		\nonumber                     \\
		 & =
		(-1)^{\derivativevectororder} \integralonrd (-1)^{\derivativevectororder} \ \kernelh(\x - \y) \ \partialpdf(\y) \ d\y
		\label{eq:kde-expected-value} \\
		 & =
		\frac{1}{\bandwidth^{\dimension}} \integralonrd \kernel \left( \frac{\x - \y}{\bandwidth} \right) \ \partialpdf(\y) \ d\y
		\nonumber                     \\
		 & =
		\integralonrd \partialpdf(\x - \bandwidth\z) \ \kernel(\z) \ d\z
		\nonumber
		\,.
	\end{align}
	Therefore,
	\begin{equation}
		\label{eq:bias-integral-form}
		\bias
		=
		\integralonrd
		[\biasintegrand]
		\ \kernel(\z)
		\ d\z
		\\
		\,.
	\end{equation}

	Now, depending on the differentiability of $\pdf$, we have one of the following Taylor expansions, using the integral form of the remainder~\cite[Theorem 2.68]{Folland2002}:
	\begin{equation}
		\label{eq:bias-taylor-expansion}
		\biasintegrand
		=
		\begin{cases}
			- \grad{\partialpdfx}\z \ \bandwidth
			+ \taylorremainder{2} \ \bandwidth^2,
			 & \mathrm{if} \ \extrasmoothing \geq 2
			\\
			\taylorremainder{1} \ \bandwidth,
			 & \mathrm{if} \ \extrasmoothing = 1
		\end{cases}
		\,,
	\end{equation}
	where
	\begin{equation}
		\label{eq:integral-taylor-remainder}
		\taylorremainder{\minmaxexponent}
		=
		\begin{cases}
			\int_0^1
			(1 - \vart) \ \transpose{\z} \hessian{\partialpdf}(\x - \vart\bandwidth\z) \z
			\ d\vart,
			 & \mathrm{if} \ \minmaxexponent = 2
			\\
			-\int_0^1
			\grad{\partialpdf}(\x - \vart\bandwidth\z) \z
			\ d\vart,
			 & \mathrm{if} \ \minmaxexponent = 1
		\end{cases}
		\,.
	\end{equation}
	The $\z$ term vanishes after plugging~\eqref{eq:bias-taylor-expansion} into~\eqref{eq:bias-integral-form} because of the zero-mean constraint on $\kernel$, yielding the following explicit expression for the bias:
	\begin{equation}
		\label{eq:bias-integral-form-with-h}
		\bias
		=
		\left(
		\integralonrd
		\taylorremainder{\minmaxexponent}
		\ \kernel(\z)
		\ d\z
		\right)
		\bandwidth^{\minmaxexponent}
		\,,
	\end{equation}
	where $\minmaxexponent = \minbetweenextrasmoothingandtwo$.

	From~\eqref{eq:bias-integral-form-with-h}, there only remains to bound~\eqref{eq:integral-taylor-remainder} to finish the proof in the differentiable case.
	For $\extrasmoothing \geq 2$, let $\constant_2 = \supxinrd \norminf{\vectorize \ \hessian{\partialpdfx}}$, where $\norminf{\funcarg}$ denotes the maximum absolute value of a vector's components.
	Likewise, for $\extrasmoothing = 1$, let $\constant_1 = \supxinrd \norminf{\grad{\partialpdfx}}$.
	Then, it is straightforward to see
	\begin{equation*}
		\smallabsval{\bias}
		\leq
		\begin{cases}
			\frac{1}{2}
			\constant_2
			\left(
			\sumijtod
			\integralonrd
			\smallabsval{\z_{\genindexi}\z_{\genindexj}}
			\ \kernel(\z)
			\ d\z
			\right)
			\bandwidth^2,
			 & \mathrm{if} \ \extrasmoothing \geq 2
			\\
			\constant_1
			\left(
			\sumitod
			\integralonrd
			\smallabsval{\z_{\genindexi}}
			\ \kernel(\z)
			\ d\z
			\right)
			\bandwidth,
			 & \mathrm{if} \ \extrasmoothing = 1
		\end{cases}
		\,,
	\end{equation*}
	which yields the desired orders after considering the moment constraints on $\kernel$.

	In turn, when $\extrasmoothing = 0$, we can resort to Hölder continuity, if $\holderexponent > 0$.
	Considering the $\norm{\funcarg}_1$ norm without loss of generality and letting $\constant > 0$ be the corresponding Hölder constant, from~\eqref{eq:bias-integral-form} follows $\smallabsval{\bias} \leq \constant \left( \integralonrd \norm{\z}_1^{\holderexponent} \kernel(\z) \ d\z \right) \bandwidth^{\holderexponent}$.
	Then, the integral is finite because of the moment constraints on $\kernel$ after applying Jensen's inequality with $\lambdafunction{\norm{\z}_1^{\holderexponent}}{(\norm{\z}_1^{\holderexponent})^{1/\holderexponent} = \norm{\z}_1 = \sumitod \smallabsval{\z_{\genindexi}}}$, which is convex for $\holderexponent \in (0, 1]$.
	This leads to a bound similar to the case $\extrasmoothing = 1$, proving the order $\bigoh(\bandwidth^{\holderexponent})$.

	Finally, if $\holderexponent = 0$, we apply uniform continuity.
	Branching from~\eqref{eq:kde-expected-value} with a change of variables, we get $\smallexpectedvalueof{\partialkdefnhx} = \integralonrd \partialpdf(\x - \y) \ \kernelh(\y) \ d\y$.
	Hence, letting $\constant > 0$ be an upper bound for $\smallabsval{\partialpdf}$ on $\rd$ and letting $\integralradius > 0$,
	\begin{align*}
		\smallabsval{\bias}
		 & \leq
		\integralonrd
		\smallabsval{\partialpdf(\x - \y) - \partialpdfx}
		\ \kernelh(\y)
		\ d\y
		\\
		 & \leq
		\sup_{\norm{\y} \leq \integralradius} \smallabsval{\partialpdf(\x - \y) - \partialpdfx}
		+ 2\constant
		\int_{\norm{\y} > \integralradius}
		\kernelh(\y)
		\ d\y
		\,.
	\end{align*}
	By uniform continuity, the supremum term will be small as long as $\integralradius$ is small regardless of $\x$.
	Once fixed a sufficiently small $\integralradius$, the integral term approaches zero as $\bandwidth$ does.
	In conclusion, the bias is $\littleohone$ as $\bandwidth \goesto 0$.
\end{proof}

\begin{proof}[Proof of Theorem~\ref{th:bump-boundary-convergence}]
	The result follows after applying Theorem~\ref{th:stability-theorem} with $\criterionfunction = \curvaturefunctionalof{\pdf}$ and $\approxcriterionfunction = \curvaturefunctionalof{\kdefnh}$.
	We must check that $\norminfcurvaturediff$, $\derivativeorder \in \{0, 1\}$, can be made arbitrarily small.

	From Theorem~\ref{th:derivative-convergence}, uniform convergence for all the \gls*{kde}'s derivatives is ensured, except for a finite union of zero-probability events.
	Then, since partial derivatives of $\pdf$ are bounded, those of $\kdefnh$ are also eventually bounded \gls*{as}
	Therefore, we can restrict the domain of the function $\underlyingcurvature$ to a compact set $\openset$.
	This claim implies that both $\curvaturefunctionalof{\testpdf}$ and $\grad{\curvaturefunctionalof{\testpdf}}$ are uniformly continuous functions of partial derivatives of $\testpdf$, which proves that $\norminfcurvaturediff$, for $\derivativeorder \in \{0, 1\}$, can be made arbitrarily small.
	In particular, $\underlyingcurvature$ is Lipschitz continuous, which yields the desired convergence order for the Hausdorff distance.
\end{proof}

\begin{proof}[Proof of Proposition~\ref{prop:distinct-eigenvalues}]
	From~\cite[Wielandt-Hoffman theorem]{Wilkinson1988}, any ordered eigenvalue function is Lipschitz continous.
	In particular, plug-in estimators of the eigenvalues are consistent in the uniform norm \gls*{as}
	Therefore, using the triangle inequality, one arrives at
	\begin{equation*}
		\begin{aligned}
			\infdiffconsecutiveeig{\kdefnh}
			 & \geq
			\infdiffconsecutiveeig{\pdf}
			\\
			 & - \erroreigenvalue{\genindexi}
			\\
			 & - \erroreigenvalue{\genindexi + 1}
			>
			0
			\,,
		\end{aligned}
	\end{equation*}
	\gls*{as} for $\samplesize$ sufficiently large, which proves~\eqref{eq:eigenvalue-separability} for $\kdefnh$.
	Let $\functiondef{\eigenvalue_{\genindexj}}{\reals^{\dimension^2}}{\reals}$ be the $j$-th eigenvalue function defined over $\dimension$-dimensional square matrices.
	From~\cite[Theorem 5.16, p. 119]{Kato1995}, for every $\x \in \fatbumpboundary$, $\eigenvalue_{\genindexj}$ is infinitely differentiable in some neighbourhoods of $\hessian{\pdf}(\x)$ and $\hessian{\kdefnh}(\x)$, since all eigenvalues have multiplicity one in both cases.
\end{proof}

\subsection{Inference}

\begin{proof}[Proof of Theorem~\ref{th:abstract-inference-indirect}]
	It is an exercise to show that $\superrorcurvature \leq \approxmargin$ implies $\genericlowerconfidentbumpapprox \subset \genericsmoothedbump \subset \genericupperconfidentbumpapprox$.
	Therefore, for sufficiently large $\samplesize$,
	\begin{equation*}
		\begin{aligned}
			\probabilityof{
				\genericlowerconfidentbumpapprox \subset \genericsmoothedbump \subset \genericupperconfidentbumpapprox
			}
			 & \geq
			\probabilityof{
				\superrorcurvature \leq \approxmargin
			}
			\\
			 & \geq
			\probabilityof{\marginrv \leq \approxmargin}
			\\
			 & =
			\probabilityof{\marginrv \leq \margin + \littleohinvsqrtn}
			\\
			 & \geq
			\probabilityof{\marginrv \leq \quantilefuncofrv{\marginrv}{\confidence} + \littleohinvsqrtn}
			\\
			 & =
			\probabilityof{\rootsamplesize\marginrv \leq \quantilefuncofrv{\rootsamplesize \marginrv}{\confidence} + \littleohone}
			\\
			 & =
			\probabilityof{\limitingmarginrv \leq \quantilefuncofrv{\rootsamplesize \marginrv}{\confidence} + \littleohone} + \littleohone
			\\
			 & =
			\confidence + \littleohone
			\,,
		\end{aligned}
	\end{equation*}
	where we have used that the weak convergence to a continuous \gls*{rv} implies convergence in the Kolmogorov distance~\cite[Lemma 2.11]{vanderVaart1998},
	and $\quantilefuncofrv{\rootsamplesize \marginrv}{\confidence}$ converges to $\quantilefuncofrv{\limitingmarginrv}{\confidence}$ for all $\oneminusconfidence \in (0, 1)$~\cite[Lemma 21.2]{vanderVaart1998}, since having $\limitingmarginrv$ a strictly increasing \gls*{cdf} implies its quantile function is continuous~\cite[p. 305]{vanderVaart1998}.
	Finally, $\limninf \approxmargin = 0$ follows from $\limninf \margin = 0$ and $\approxmargin = \margin + \littleohinvsqrtn$.
\end{proof}

\begin{proof}[Proof of Theorem~\ref{th:gaussian-approximation-supnorm}]
	One can easily check that $\genericdiffop\smoothedpdf(\x) = \smallexpectedvalueof{\genericdiffop\kernelh(\x - \samplevar_{\genindexi})}$, for all $\genindexi$.
	This leads to $\scaling \supreumumnormerror{\genericdiffop} = \supovervcclasshof{\empiricalprocess}$, for $\vcclassh$ as in~\eqref{eq:vc-class-scaled-kernels}.
	On the other hand, the class
	\begin{equation*}
		\vcclasshprime
		=
		\left \{
		\lambdafunction{\y \in \rd}{%
			\genericdiffop\kernel \left( \frac{\x - \y}{\bandwidth} \right)
		}
		: \x \in \bumpdomain
		\right \}
		=
		\{
		\sqrt{\bandwidth^{\dimension + \difforder}} \vcclassmember
		: \vcclassmember \in \vcclassh
		\}
		\,,
	\end{equation*}
	for any $\bandwidth > 0$, is \gls*{vc}-type by our assumption on the kernel class $\kernelsvcclass$ in~\eqref{eq:vc-class-kernels}, since \gls*{vc}-type classes are closed under summation and product~\cite{Chernozhukov2014}.
	Also, $\vcclasshprime$ is clearly \gls*{pm}, as the continuity of $\genericdiffop\kernel$ allows for a countable subset indexed by a rational $\dimension$-tuple $\x$.
	Therefore, we can apply Lemma~\ref{lm:chernozhukov} to $\vcclasshprime$, yielding the existence of $\gaussianprocessrvhprime \equaldistribution \supovervcclasshprimeof{\gaussianprocess}$ near $\supovervcclasshprimeof{\empiricalprocess}$ with high probability.
	Now, by our boundedness assumption on the $\difforder$-th derivatives of $\kernel$, letting $\vckernelsbound > 0$ be an envelope for $\vcclassh$, and assumming $\bandwidth \in (0, 1)$, the same $\vckernelsbound$ will be an envelope for $\vcclasshprime$, yielding
	$
		\supovervcclasshprime \expectedvalueof{
			\vcclassmember(\samplevar_1)^2
		}
		\leq
		\bandwidth^{\dimension + \difforder} \vckernelsbound^2
		\leq
		\vckernelsbound^2
	$,
	whence we can take $\vcbound = \vckernelsbound$ and $\processstdv = \vckernelsbound \sqrt{\bandwidth^{\dimension + \difforder}}$ in Lemma~\ref{lm:chernozhukov}.
	Dividing both sides of the inequality inside $\probability$ in Lemma~\ref{lm:chernozhukov} by $\sqrt{\bandwidth^{\dimension + \difforder}}$, we arrive at
	\begin{equation*}
		\probabilityof{
			\absoluteval{\supovervcclasshof{\empiricalprocess} - \frac{1}{\sqrt{\bandwidth^{\dimension + \difforder}}} \gaussianprocessrvhprime}
			>
			\chernozhukovfirstconst
			\frac{\vckernelsbound \log^{2/3} \samplesize}
			{\chernozhukovfreeconst^{1/3} (\samplesize \bandwidth^{\dimension + \difforder})^{1/6}}
		}
		\leq
		\chernozhukovsecondtconst \chernozhukovfreeconst
		\,,
	\end{equation*}
	for $\samplesize$ sufficiently large.
	Letting $\gaussianprocessrvh = \gaussianprocessrvhprime / \sqrt{\bandwidth^{\dimension + \difforder}}$, $\modifiedchernofirstconst = \vckernelsbound \chernozhukovfirstconst$, and bypassing the empirical process, we get
	\begin{equation}
		\label{eq:applied-chernozhukov}
		\probabilityof{
			\absoluteval{\scaling \supreumumnormerror{\genericdiffop} - \gaussianprocessrvh}
			>
			\modifiedchernofirstconst
			\frac{\log^{2/3} \samplesize}
			{\chernozhukovfreeconst^{1/3} (\samplesize \bandwidth^{\dimension + \difforder})^{1/6}}
		}
		\leq
		\chernozhukovsecondtconst \chernozhukovfreeconst
		\,.
	\end{equation}
	Note that $\gaussianprocessrvh \equaldistribution \supovervcclasshof{\gaussianprocess}$~\cite[Supplementary material]{Chen2016}.

	To go from~\eqref{eq:applied-chernozhukov} to the desired result, we apply~\cite[Lemma 10 -- supplementary material]{Chen2017}, which in turn derives from~\cite[Lemma 2.3]{Chernozhukov2014}.
	As in~\cite{Chen2017}, the assumptions (A1)-(A3) hold under our hypotheses.
	(A1) is a \gls*{pm} requirement, which again stands valid for $\vcclassh$.
	Then, (A2) is satisfied with the bound $\vckernelsbound$, and we can use~\cite[Lemma 2.1]{Chernozhukov2014} to infer the pre-Gaussian requirement (A3) from (A2) and the fact that \gls*{vc}-type classes have finite \textit{uniform entropy integral}~\cite{Chernozhukov2014, vanderVaart1998}.
	Therefore,
	\begin{equation*}
		\begin{split}
			\sup_{\cdfarg \geq 0}
			\absoluteval{
				\probabilityof{\scaling \supreumumnormerror{\genericdiffop} < \cdfarg}
				-
				\probabilityof{\gaussianprocessrvh < \cdfarg}
			}
			&\leq
			\anticoncentrationconstant \expectedvalueof{\gaussianprocessrvh}
			\left(
			\modifiedchernofirstconst
			\frac{\log^{2/3} \samplesize}
			{\chernozhukovfreeconst^{1/3} (\samplesize \bandwidth^{\dimension + \difforder})^{1/6}}
			\right)
			+
			\chernozhukovsecondtconst \chernozhukovfreeconst
			\\
			&\leq
			\anticoncentrationconstantprime
			\modifiedchernofirstconst
			\frac{\log^{7/6} \samplesize}
			{\chernozhukovfreeconst^{1/3} (\samplesize \bandwidth^{\dimension + \difforder})^{1/6}}
			+
			\chernozhukovsecondtconst \chernozhukovfreeconst
			\,,
		\end{split}
	\end{equation*}
	for positive constants $\anticoncentrationconstant, \anticoncentrationconstantprime$, where it has been used that $\expectedvalueof{\gaussianprocessrvh} = \bigoh(\sqrt{\log \samplesize})$.
	The order of $\expectedvalueof{\gaussianprocessrvh}$ is obtained noting that, by the coupling hypothesis between $\bandwidth$ and $\samplesize$, $\expectedvalueof{\gaussianprocessrvh} = \bigoh(1 / \sqrt{\bandwidth^{\dimension + \difforder}}) = \bigoh(\sqrt{\log \samplesize})$, where we have used that $\expectedvalueof{\gaussianprocessrvhprime} = \bigoh(1)$ uniformly in $\bandwidth$ due to Dudley's inequality~\cite[Corollary 2.2.8]{vanderVaart1996} applied to the \gls*{vc}-type class $\bigcup_{\delta > 0} \vcclass_{\delta}'$, which is larger than $\vcclasshprime$.
	Differentiating the right-hand side of the last inequality with respect to $\chernozhukovfreeconst$ yields the optimal order $\bigoh[ \log^{7/8} \samplesize / (\samplesize \bandwidth^{\dimension + \difforder})^{1/8}]$, which also happens to be that of $\chernozhukovfreeconst$.

	Finally, plugging the optimal $\chernozhukovfreeconst$ in~\eqref{eq:applied-chernozhukov}, given $\epsilon > 0$, there exists $\samplesizestart \in \naturals$ such that, for $\samplesize \geq \samplesizestart$, we have
	$
		\probabilityof{
			\absoluteval{
				\rootsamplesize \ \supreumumnormerror{\genericdiffop}
				-
				\normalizedgaussianprocessrvh
			}
			>
			\epsilon
		}
		=
		\bigoh ( [ \samplesize^{-1} \log^7 \samplesize ]^{1/8} )
	$, meaning $\convergesinprob{\rootsamplesize \ \supreumumnormerror{\genericdiffop}}{\normalizedgaussianprocessrvh}$.
\end{proof}

\begin{proof}[Proof of Theorem~\ref{th:gaussian-approximation-supnorm-bootstrap}]
	Let us start by decoupling the sample sizes in~\eqref{eq:inference-uniform-error-bootstrap}.
	Let us fix some observed values $\originalsample = \{\samplevarrealization_1, \dots, \samplevarrealization_{\samplesize}\}$ and take $\bootstrapsamplevar_1, \dots, \bootstrapsamplevar_{\resamplesize}$ \gls*{iid} random variables from $\empiricalbootstrapprob$, for $\resamplesize$ independent of $\samplesize$.
	Keeping $\samplesize$ fixed, let us consider the decoupled version of~\eqref{eq:inference-uniform-error-bootstrap} $\bootstrapsupreumumnormerrorm{\genericdiffop} = \smallsupnormbumpdomain{\genericdiffop\bootstrapkdefmhof{\x} - \genericdiffop\conditionalkdefnhof{\x}}$.
	Note that, for $\genindexi \in \{1, \dots, \resamplesize\}$, we have $\expectedvalueof{\kernelh(\x - \bootstrapsamplevar_{\genindexi}) | \originalsample} = \conditionalkdefnhof{\x}$.
	This leads to
	\begin{equation*}
		\genericdiffop\bootstrapkdefmhof{\x} - \genericdiffop\conditionalkdefnhof{\x}
		=
		\frac{1}{\resamplesize\bandwidth^{\dimension + \difforder}}
		\sum_{\genindexi = 1}^{\resamplesize}
		\genericdiffop\kernel \left(
		\frac{\x - \bootstrapsamplevar_{\genindexi}}{\bandwidth}
		\right)
		-
		\expectedvalueof{
			\genericdiffop\kernel \left(
			\frac{\x - \bootstrapsamplevar_{\genindexi}}{\bandwidth}
			\right)
			\biggr\rvert
			\originalsample
		}
		\,,
	\end{equation*}
	and, subsequently, $\sqrt{\resamplesize \bandwidth^{\dimension + \difforder}} \bootstrapsupreumumnormerrorm{\genericdiffop} = \supovervcclasshof{\empiricalprocessbootstrap}$, for $\vcclassh$ as in~\eqref{eq:vc-class-scaled-kernels}.
	From this point on, one can apply exactly the same steps in the proof of Theorem~\ref{th:gaussian-approximation-supnorm} to obtain an approximation for $\bootstrapsupreumumnormerrorm{\genericdiffop}$ in the form of the supremum $\gaussianprocessrvnh$ of a \gls*{gp}.
	The covariance structure for this $\gaussianprocessrvnh$ immediately follows by considering $\bootstrapsamplevar_1$ and $\empiricalbootstrapprob$ in~\eqref{eq:covariance-gaussian-process}.

	A question remains to be solved to finish the proof: undo the decoupling of $\resamplesize$ from $\samplesize$ and the original sample $\originalsample$.
	We have proved that, given $\originalsample$, the Gaussian approximation works for sufficiently large $\resamplesize$, but we have no guarantee the same would work if we took $\resamplesize = \samplesize \goesto \infty$.
	A close look at Lemma~\ref{lm:chernozhukov} and Theorem~\ref{th:gaussian-approximation-supnorm} shows all the constants involved are universal or depend only on $\vcclassh$, which remains unchanged for every $\originalsample$.
	Therefore, we can take $\resamplesize = \samplesize$ and conclude that the result holds for a sufficiently large $\samplesize$.
\end{proof}

\begin{proof}[Proof of Corollary~\ref{cor:bootstrap-approximation}]
	It suffices to check that $\diffcdfgaussiansup$ converges to zero.
	First, note that the \gls*{gp} $\gaussianprocessbootstrap$ converges weakly to the \gls*{gp} $\gaussianprocess$ since the sample covariances converge to the population covariances~\eqref{eq:covariance-gaussian-process} \gls*{as} due to the law of large numbers.
	Therefore, applying the \textit{continuous mapping theorem}~\cite[Theorem 1.3.6]{vanderVaart1996}, we assure the suprema of the \glspl*{gp} also converge weakly.
	Let us assume for a moment that $\gaussianprocessrvh$ has a continuous \gls*{cdf}.
	Then, the uniform convergence in distribution over all $\cdfarg \in \reals$ would follow using~\cite[Lemma 2.11]{vanderVaart1998}, finishing the proof.
	Now, $\gaussianprocessrvh$ has a continuous and strictly increasing \gls*{cdf} following~\cite[Corollary 1.3 + Remark 4.1]{Gaenssler2007} since we automatically rule out the case $\kernelsvcclass$ in~\eqref{eq:vc-class-kernels} is degenerate by the regularity assumptions on the kernel $\kernel$.
	That is, we have $\kernelsvcclass \neq \{ \lambdafunction{\y}{0} \}$.
	Therefore, $\gaussianprocessrvh$ cannot be the Dirac distribution concentrated at zero.
	Meanwhile, the properties of $\normalizedgaussianprocessrvh$ derive from those of $\gaussianprocessrvh$.
\end{proof}

\begin{proof}[Proof of Lemma~\ref{lm:tvar-convergence}]
	From~\cite[Theorem 2.1]{Dhaene2006}, we have, for all $\p \in (0, 1)$,
	\begin{equation*}
		\tvar{\p}{\randomvarx_{\genindexn}}
		=
		\quantilefuncofrv{\randomvarx_{\genindexn}}{\p}
		+
		\frac{1}{1 - \p}
		\expectedvalueof{\positivepart{\randomvarx_{\genindexn} - \quantilefuncofrv{\randomvarx_{\genindexn}}{\p}}}
		\,,
	\end{equation*}
	where $\positivepart{\funcarg} \equiv \max\{\funcarg, 0\}$.
	Now, since $\randomvarx$ has a strictly increasing \gls*{cdf}, its quantile function is continuous at every $\p \in (0, 1)$.
	Following~\cite[Lemma 21.2]{vanderVaart1998}, we thus get $\limninf \quantilefuncofrv{\randomvarx_{\genindexn}}{\p} = \quantilefuncofrv{\randomvarx}{\p}$.
	On the other hand, using~\cite[Theorem 2.20]{vanderVaart1998} and Slutsky's~\cite[Lemma 2.8 (i)]{vanderVaart1998}, it is not difficult to see that $\{\randomvarx_{\genindexn} - \quantilefuncofrv{\randomvarx_{\genindexn}}{\p}\}_{\genindexn = 1}^{\infty}$ is \gls*{aui}
	Next, noting that $\positivepart{\randomvarx_{\genindexn} - \quantilefuncofrv{\randomvarx_{\genindexn}}{\p}} \leq \absoluteval{\randomvarx_{\genindexn} - \quantilefuncofrv{\randomvarx_{\genindexn}}{\p}}$ it easily follows that $\{\positivepart{\randomvarx_{\genindexn} - \quantilefuncofrv{\randomvarx_{\genindexn}}{\p}}\}_{\genindexn = 1}^{\infty}$ is \gls*{aui} too.
	Then, using again~\cite[Theorem 2.20]{vanderVaart1998}, the latter implies $\limninf \expectedvalueof{\positivepart{\randomvarx_{\genindexn} - \quantilefuncofrv{\randomvarx_{\genindexn}}{\p}}} = \expectedvalueof{\positivepart{\randomvarx - \quantilefuncofrv{\randomvarx}{\p}}}$ after applying the \textit{continuous mapping theorem}~\cite[Theorem 2.3]{vanderVaart1998} on $\positivepart{\funcarg}$.
	Finally, from the previous derivations and~\cite[Theorem 2.1]{Dhaene2006}, $\limninf \tvar{\p}{\randomvarx_{\genindexn}} = \tvar{\p}{\randomvarx}$.
\end{proof}

\begin{proof}[Proof of Theorem~\ref{th:inference-eigenvalue-bumps}]
	The eigenvalues are Lipschitz continuous as functions of matrix entries.
	As shown in~\cite[Equation 48.1, p. 104]{Wilkinson1988}, the absolute difference between eigenvalues at distinct matrices is bounded by the Euclidean 2-norm~\cite[Equation 54.1, p. 57]{Wilkinson1988} of the difference between the corresponding matrices.
	Let $\eigenvalueof{}{\funcarg}$ be either $\eigenvalueof{1}{\funcarg}$ or $\eigenvalueof{\dimension}{\funcarg}$ in~\eqref{eq:ordered-eigenvalues}.
	Taking suprema at both sides of the inequality, and noting that $\norm{\funcarg}_2 \leq \norm{\funcarg}_1$, we get
	\begin{equation*}
		\begin{aligned}
			\supreumumnormerrorfunctional{\eigenvalue}
			\leq
			\supxinrd
			\sumijtod
			\absvalsecondderivatives
			\leq
			\sumijtod
			\supnormsecondderivativesij
			\equiv
			\marginrv
		\end{aligned}
		\,.
	\end{equation*}
	Note that $\convergesinprob{\rootsamplesize \ \marginrv}{\limitingmarginrv = \sumijtod\sumsupgaussians}$ since the terms in the sum jointly converge~\cite[Theorem 2.7 (vi)]{vanderVaart1998}, and then we can apply the \textit{continuous mapping theorem}~\cite[Theorem 2.3]{vanderVaart1998}.
	This fulfils the pillar (\ref{itm:bounding-rv}) in Theorem~\ref{th:abstract-inference-indirect}.

	Now, $\marginrv$ is the sum of random variables whose \glspl*{cdf} can be approximated via bootstrap.
	The obstacle here lies in the dependencies between the terms in $\marginrv$; Corollary~\ref{cor:bootstrap-approximation} only applies to the margin \glspl*{cdf}, not the joint \gls*{cdf}.
	Notwithstanding, we only need some $\margin$ satisfying $\margin \geq \quantilefuncofrv{\marginrv}{\confidence}$.
	This $\margin$ can be obtained by resorting to the \gls*{tvar} concept, which bounds from above the corresponding quantile at the same confidence level while sub-additive~\cite[Equation 37]{Dhaene2006}.
	Interestingly, the \gls*{tvar} is the lowest \textit{concave distortion risk measure} satisfying both requirements~\cite[Theorem 5.2.2]{Dhaene2006}.
	Therefore,
	\begin{equation*}
		\quantilefuncofrv{\marginrv}{\confidence}
		\leq
		\tvar{\confidence}{\marginrv}
		\leq
		\sumijtod
		\tvar{\confidence}{\supnormsecondderivativesij}
		\equiv
		\margin
		\,.
	\end{equation*}
	Note that $\margin$ converges to zero because, by Lemma~\ref{lm:tvar-convergence} and assumption (\ref{itm:aui-errors}), for every pair $(\genindexi, \genindexj)$, $\limsamplesizeinf \rootsamplesize \ \tvar{\confidence}{\supnormsecondderivativesij} = \tvar{\confidence}{\sumsupgaussians} < \infty$, where we have used Corollary~\ref{cor:bootstrap-approximation} and $\rootsamplesize \ \tvar{\p}{\randomvarx} = \tvar{\p}{\rootsamplesize \randomvarx}$, for~\cite[Lemma 2.1]{Dhaene2006}.
	This completes hypothesis (\ref{itm:margin-bounding-quantile}) in Theorem~\ref{th:abstract-inference-indirect}.

	Next, consider the bootstrap margin $\approxmargin$.
	There only remains to verify (\ref{itm:convergent-approx-margin}) in Theorem~\ref{th:abstract-inference-indirect}, i.e., $\smallabsval{\approxmargin - \margin} = \littleohinvsqrtn$, to finish the proof.
	In turn, this follows using again~\cite[Lemma 2.1]{Dhaene2006} after checking, for every pair $(\genindexi, \genindexj)$,
	\begin{equation}
		\label{eq:tvar-approximation}
		\absoluteval{
			\tvar{\confidence}{\rootsamplesize \ \bootstrapsupnormderivativesij}
			-
			\tvar{\confidence}{\rootsamplesize \ \supnormsecondderivativesij}
		}
		\almostsureconverges
		0
		\,.
	\end{equation}
	Finally,~\eqref{eq:tvar-approximation} is a direct consequence of Lemma~\ref{lm:tvar-convergence} and the assumption (\ref{itm:aui-errors}), considering the two sequences have the same weak limit $\sumsupgaussians$, by Corollary~\ref{cor:bootstrap-approximation}.
\end{proof}

\begin{proof}[Proof of Theorem~\ref{th:inference-gaussian-bumps}]
	First, it is an exercise to check that, for the Gaussian kernel, $\norminfk{\kernelh}{2} = (2\pi\bandwidth^4)^{-1}$ on $\rtwo$ and, thus, $\norminfk{\kernelh}{2} < \constant / 2$ on $\rtwo$.
	Then, one can quickly check that $\norminfk{\smoothedpdf}{2} < \constant / 2$ on $\bumpdomain$, and, from Lemma~\ref{lm:supnorm-bigoh}, $\norminfk{\kdefnh}{2} < \constant / 2$ on $\bumpdomain$ for sufficiently large $\samplesize$ \gls*{as}
	Therefore, we can use~\cite[Theorem 2.12, p. 768]{Ipsen2008} along with~\cite[Equation 54.5, p. 57]{Wilkinson1988} and the fact that $\norm{\funcarg}_2 \leq \norm{\funcarg}_1$ to obtain
	\begin{equation*}
		\begin{split}
			\supreumumnormerrorfunctional{\determinant{\hessian{\funcarg}}}
			&\leq
			2
			\max\{
			\norminfk{\smoothedpdf}{2},
			\norminfk{\kdefnh}{2}
			\}
			\supxinbumpdomain
			\sumijtotwo
			\absvalsecondderivatives
			\\
			&\leq
			\constant \
			\sumijtotwo
			\supnormsecondderivativesij
			\equiv
			\marginrv
			\,,
		\end{split}
	\end{equation*}
	\gls*{as} for sufficiently large $\samplesize$.
	From this point on, the proof follows the same steps as that of Theorem~\ref{th:inference-eigenvalue-bumps}.
\end{proof}

			\section{Curvature}

\label{sec:curvature}

The rationale behind Section~\ref{sec:methods} lies in classical geometry.
We refer the reader to~\citeauthor{doCarmo1992}'s books for a comprehensive introduction to differential and Riemannian geometry~\cite{doCarmo1976, doCarmo1992} and~\citeauthor{Grinfeld2013}'s for an operational perspective on tensor calculus~\cite{Grinfeld2013}.
Then, \citeauthor{Petersen1998}'s book has a complete chapter devoted to the central topic of hypersurfaces~\cite[Chapter 4]{Petersen1998}.
The appendix of~\citeauthor{Ecker2004}'s book also includes some key concepts~\cite{Ecker2004}.
Finally, we recommend~\citeauthor{Gray2006}'s book to build up some intuition in $\reals^3$~\cite[Chapter 13]{Gray2006}.

\subsection{Outline}

After recalling some basic notions, we will address the concept of principal curvature.
Concavity and convexity are defined through specific sign configurations among the principal curvatures.
On the other hand, the \textit{mean curvature} and the \textit{Gaussian curvature} are alternative scalar summaries of principal curvatures.
All these measures involve first and second-order derivatives of the \gls*{pdf}.
Theoretical considerations below will discard the need for higher-order derivatives beyond those.
Meanwhile, combining derivatives of different orders is a major inconvenience in a \gls*{kde} setting.
Hence, we present similarly-intended features that rely only on the Hessian matrix.

\subsection{Fundamental concepts}

Let $\functiondef{\pdf}{\rd}{\closedzeroinfinterval}$ be a $\dimension$-dimensional \gls*{pdf}.
Let us assume that $\pdf$ is as smooth as needed.
The graph of $\pdf$, that is, $\graphf = \setgraphf$, is a hypersurface embedded in $\rdpone$ and a $\dimension$-dimensional smooth manifold, in particular.
Let $\functiondef{\chart}{\graphf}{\rd}$ be the global coordinates chart of $\graphf$ mapping every $\graphfpair \in \graphf$ to $\x$.
Let also $\functiondef{\embedding}{\rd}{\graphf}$ be the global parameterization of $\graphf$, defined as $\embedding(\x) = \chart^{-1}(\x) = \graphfpair$.

At every $\graphfpoint \in \graphf$, we can define a $\dimension$-dimensional tangent vector space $\tangentspace$ comprising different ways of deriving smooth functions from $\graphf$ to $\reals$ in a neighbourhood of $\graphfpoint$~\cite{doCarmo1992}.
In particular, the $\genindexi$-th element of the canonical basis determined by the chart is the $\genindexi$-th partial derivative $\ithtangentvectorp$ acting on functions $\functiondef{\testfunction}{\graphf}{\reals}$ as $\ithtangentvectorp(\testfunction) = [\chartcomp{\ithdifferential{\genindexi}(\testfunction \comp \embedding)}](\graphfpoint)$, where $\ithdifferential{\genindexi}$ takes the $\genindexi$-th partial derivative of $\functiondef{\testfunction \comp \embedding}{\rd}{\reals}$.

The union of all tangent spaces across $\graphfpoint \in \graphf$ is also a smooth manifold, known as the tangent bundle $\tangentbundle$~\cite{doCarmo1992}.
Elements in $\tangentbundle$ are called vector fields.
A vector field can be thought of as a function mapping each point $\graphfpoint \in \graphf$ to a tangent vector in $\tangentspace$.
Every vector field $\vectorfieldx$ can be expressed, using Einstein's summation convention~\cite{Grinfeld2013}, as $\vectorfieldx\vert_{\funcarg} = \vectorfieldx^{\genindexi}(\funcarg) \tangentvector{\genindexi}{\funcarg}$, where every $\vectorfieldx^{\genindexi}$ is a smooth function $\functionmap{\graphf}{\reals}$.

Let us now build on the previous concepts considering $\graphf$ is a hypersurface.

\subsubsection{Tangent vector fields}

Being $\graphf$ immersed in $\rdpone$, there exists a canonical identification of basis vector fields $\basisvectorfield{\genindexi}$ with real functions $\functionmap{\rd}{\rdpone}$~\cite{Ecker2004}.
To see this, consider the canonical identity immersion $\functiondef{\immersion}{\graphf}{\rdpone}$.
The differential $\tangentdifferential{\immersion}$ transforms vector fields in $\tangentbundle$ into vector fields in $\rdpone$.
Given a smooth function $\functiondef{\testfunction}{\rdpone}{\reals}$, we have $[\dimmersionithbasis](\testfunction) = \ithdifferential{\genindexi}(\testfunction \comp \immersedembedding)$, where $\functiondef{\immersedembedding \equiv \immersion \comp \embedding}{\rd}{\rdpone}$ is the immersed parameterization of $\graphf$.
Now, using the chain rule, for any $\x \in \rd$, $\ithdifferential{\genindexi}(\testfunction \comp \immersedembedding)(\x) = \innerprod{\grad{\testfunction}\graphfpair}{\embeddingfield(\x)}$, where $\embeddingfield$ is the $\genindexi$-th column vector field in the Jacobian matrix of $\immersedembedding$.
Combining the last two equations, we see that $\dimmersionithbasis$ \textit{derives} the function $\testfunction$ in the \textit{direction} of the vector field $\embeddingfield$.
Therefore, there is a canonical isomorphism
\begin{equation}
	\label{eq:canonical-isomorphism}
	\dimmersionithbasis
	\cong
	\embeddingfield
	=
	\euclideanbasis{\genindexi} + \differential_{\genindexi}\pdf \ \euclideanbasis{\dimension + 1}
	\,,
\end{equation}
where $\euclideanbasis{\genindexi}$ denotes the $\genindexi$-th canonical basis vector in the Euclidean space $\rdpone$.

The identification~\eqref{eq:canonical-isomorphism} unlocks both intrinsic and extrinsic properties of $\graphf$.
On the one hand, we have an inner product in $\rdpone$ that can be brought into $\tangentbundle$ via~\eqref{eq:canonical-isomorphism}.
On the other hand, we have $\dimension$ tangent vectors in $\tangentbundle$ and $\dimension + 1$ linearly independent vector fields in $\rdpone$, leaving space for an orthogonal complement \textit{normal} to the hypersurface.

\subsubsection{Normal vector field}

At any given $\x \in \rd$, the vectors $\embeddingfield(\x)$ form a basis for the $\dimension$-dimensional tangent hyperplane $\tangenthyperplane{\x} \subset \rdpone$ at $\embedding(\x)$.
Its orthogonal complement $\tangenthyperplane{\x}^{\perp}$ is a straight line passing through $\embedding(\x)$.
One can check that
\begin{equation}
	\label{eq:normal-vector-field}
	\downwardnormalvecfield(\x)
	=
	\frac{(\gradf(\x), -1)}{\sqrt{1 + \norm{\gradf(\x)}^2}}
\end{equation}
corresponds to the downward unit normal vector to $\tangenthyperplane{\x}$, perpendicular to vectors~\eqref{eq:canonical-isomorphism}, thus generating $\tangenthyperplane{\x}^{\perp}$~\cite[Equation A.2]{Ecker2004}.

Studying how the normal vector field~\eqref{eq:normal-vector-field} changes along all directions over $\graphf$ provides a means to characterize curvature.

\subsubsection{First fundamental form}

Utilizing~\eqref{eq:canonical-isomorphism}, the Euclidean inner product in $\rdpone$ induces a metric field $\metric$ on $\graphf$ with components~\cite{Ecker2004}
\begin{equation}
	\label{eq:metric}
	\metric_{\genindexi \genindexj}
	\equiv
	\metric(\basisvectorfield{\genindexi}, \basisvectorfield{\genindexj})
	\eqdef
	\innerprod
		{\chartcomp{\embeddingfieldindex{\genindexi}}}
		{\chartcomp{\embeddingfieldindex{\genindexj}}}
	=
	\kroneckerdd{\genindexi}{\genindexj}
	+
	(\chartcomp{\differential_{\genindexi}\pdf})
	(\chartcomp{\differential_{\genindexj}\pdf})
	\,,
\end{equation}
where $\kroneckerdd{\genindexi}{\genindexj}$ denotes the Kronecker delta twice covariant tensor.
The metric tensor $\metric$ is known as the \textit{first fundamental form}.
It allows measuring angles and lengths on $\tangentbundle$ \textit{intrinsically}.
In particular, it defines at each point a norm $\norm{\vectorfieldx} = \sqrt{\metric(\vectorfieldx, \vectorfieldx)}$, for $\vectorfieldx \in \tangentbundle$.

The metric tensor has an inverse $\metric^{\genindexi \genindexj}$ satisfying $\kroneckerup{\genindexi}{\genindexk} = \metric^{\genindexi \genindexj} \metric_{\genindexj \genindexk}$, where $\kroneckerup{\genindexi}{\genindexk}$ is the Kronecker delta $(1, 1)$-tensor.
From the previous equation and~\eqref{eq:metric}, using the Sherman–Morrison formula~\cite[Corollary 18.2.11]{Harville1997}, one can check that
\begin{equation*}
	\metric^{\genindexi \genindexj}
	=
	\kroneckerdd{\genindexi}{\genindexj}
	-
	\frac
	{(\chartcomp{\differential_{\genindexi}\pdf})
	 (\chartcomp{\differential_{\genindexj}\pdf})}
	{1 + \norm{\chartcomp{\gradf}}^2}
	\,.
\end{equation*}
A simple expression for the metric determinant follows from~\cite[Corollary 18.1.3]{Harville1997} when applied to~\eqref{eq:metric}, yielding $\determinant{\metric} = 1 + \norm{\chartcomp{\gradf}}^2$.

In our context, compatibility with the metric tensor leads to a \textit{natural} definition of derivatives for tangent vector fields over $\graphf$.

\subsubsection{Field divergence}

The metric tensor allows defining \textit{covariant derivatives} over vector fields through the \textit{Levi-Civita connection} $\connection$~\cite{doCarmo1992}, with components given by the Christoffel symbols $\christoffelwithindices{\genindexi}{\genindexj}{\genindexk}$~\cite[Equation 5.66]{Grinfeld2013}.
Letting $\vectorfieldx, \vectorfieldy \in \tangentbundle$, the covariant derivative of $\connection_{\vectorfieldy} \vectorfieldx$ is a vector field in $\tangentbundle$.
See~\cite[Equation 8.9]{Grinfeld2013} for an explicit expression of $\connection_{\genindexi} \vectorfieldx \equiv \connection_{\basisvectorfield{\genindexi}} \vectorfieldx$ involving the Christoffel symbols, the components of $\vectorfieldx$ and the basis vector fields.

The covariant derivative then can be used to define the \textit{divergence} operator~\cite[Equation 8.2]{Grinfeld2013} on $\vectorfieldx \in \tangentbundle$ as $\intrinsicdivergence{\vectorfieldx} = (\connection_{\genindexi} \vectorfieldx)^{\genindexi}$, which is a function $\functionmap{\graphf}{\reals}$.
As it turns out, the divergence can be extended to \textit{extrinsic} vector fields $\functiondef{\vectorfieldv}{\rd}{\rdpone}$ living in the ambient space through~\cite{Ecker2004}
\begin{equation}
	\label{eq:extrinsic-divergence}
	\extrinsicdivergence{\vectorfieldv}
	=
	\metric^{\genindexi \genindexj}
	\innerprod
		{\chartcomp{\differential_{\genindexi} \vectorfieldv}}
		{\chartcomp{\embeddingfieldindex{\genindexj}}}
	\,,
\end{equation}
meaning $\extrinsicdivergence{\funcarg}$ coincides with $\intrinsicdivergence{\funcarg}$ over tangent vector fields, i.e., for all $\genindexj$, $\extrinsicdivergence{\embeddingfieldindex{\genindexj}} = \intrinsicdivergence{\basisvectorfield{\genindexj}}$.
To see this, note that $\differential_{\genindexi} \embeddingfieldindex{\genindexj} = (\christoffelwithindices{\genindexi}{\genindexj}{\genindexk} \comp \embedding) \ \embeddingfieldindex{\genindexk} + \normalvecfield$, where $\normalvecfield$ is a normal vector field~\cite{Ecker2004}.

The divergence operator~\eqref{eq:extrinsic-divergence} appears in an alternative expression for the mean curvature defined below.

\subsubsection{Second fundamental form}

Similarly to how the first fundamental form~\eqref{eq:metric} measures change intrinsically, we can define another twice covariant tensor measuring changes in the normal vector field along different tangent vector fields.
This new tensor $\secondfundform$ is the \textit{second fundamental form}~\cite{Ecker2004}.
It has components
\begin{equation}
	\label{eq:second-fundamental-form}
	\secondfundform_{\genindexi \genindexj}
	\equiv
	\secondfundform
	(\basisvectorfield{\genindexi}, \basisvectorfield{\genindexj})
	\eqdef
	\innerprod
		{\chartcomp{\ithdiffnormal}}
		{\chartcomp{\embeddingfieldindex{\genindexj}}}
	=
	\frac
		{\chartcomp{\secondorderdifffunctional{\genindexi}{\genindexj}\pdf}}
		{\sqrt{1 + \norm{\chartcomp{\gradf}}^2}}
	\,,
\end{equation}
where $\secondorderdifffunctional{\genindexi}{\genindexj}$ denotes second-order partial differentiation in the $\genindexi$ and $\genindexj$ variables~\cite{Ecker2004}.
The second fundamental form is \textit{extrinsically} defined since it contains the normal vector, which \textit{lives} outside the tangent space.

\subsubsection{Shape operator}

Having introduced the first and second fundamental forms, the shape operator $\functiondef{\shapeop}{\tangentbundle}{\tangentbundle}$ connects the two via $\metric(\shapeop(\basisvectorfield{\genindexi}), \basisvectorfield{\genindexj}) = \secondfundform_{\genindexi \genindexj}$~\cite{Petersen1998}.
One can easily verify that the shape operator has components $\shapeop^{\genindexi}_{\genindexk} = \metric^{\genindexi \genindexj} \secondfundform_{\genindexj \genindexk}$~\cite{Ecker2004}.

Since $\secondfundform$ is symmetric, i.e., $\secondfundform_{\genindexi \genindexj} = \secondfundform_{\genindexj \genindexi}$, one can check that $\shapeop$ is a self-adjoint operator, meaning $\metric(\shapeop(\vectorfieldx), \vectorfieldy) = \metric(\vectorfieldx, \shapeop(\vectorfieldy))$, for $\vectorfieldx, \vectorfieldy \in \tangentbundle$.
This implies that, at any point $\graphfpoint \in \graphf$, all the eigenvalues of $\shapeop_{\graphfpoint}$ are real and there exists an associated orthonormal basis consisting of eigenvectors of $\shapeop_{\graphfpoint}$~\cite[Section 2.4.6]{Petersen1998}.
These eigenvalues correspond to some curvature measure along the direction of the corresponding eigenvectors.

The shape operator is easily confused with the second fundamental form.
\citeauthor{Petersen1998} refers to this issue as a matter of perspective~\cite[Section 2.3.1]{Petersen1998}.
\citeauthor{Ecker2004} even employs the same letter $\secondfundform$ instead of $\shapeop$ to refer to the shape operator, relying on the position of the indices to distinguish between them ($\shapeop$ is a $(1, 1)$-tensor, while $\secondfundform$ is $(2, 0)$).
Meanwhile,~\citeauthor{Gray2006} equivalently define $\secondfundform$ from $\shapeop$ and $\metric$~\cite[Definition 13.28]{Gray2006}.

The \textit{curvature tensor}, pervasive in Riemannian geometry, takes a simple form for hypersurfaces, exclusively involving $\metric$ and $\shapeop$~\cite[Section 4.2]{Petersen1998}.
This virtually ensures that curvature can be apprehended by inspecting only the first and second derivatives of the \gls*{pdf} $\pdf$.

\subsubsection{Normal and principal curvatures}

The definition of the \textit{normal curvature}~\cite{doCarmo1976, Gray2006} of $\graphf$ in the non-null direction $\vectorfieldx = \vectorfieldx^{\genindexi}\basisvectorfield{\genindexi} \in \tangentbundle$ also applies to hypersurfaces:
\begin{equation}
	\label{eq:normal-curvature}
	\normalcurvature(\vectorfieldx)
	=
	\frac{
		\metric(
			\shapeop(\vectorfieldx),
			\vectorfieldx)
	}{
		\metric(
			\vectorfieldx,
			\vectorfieldx)
	}
	=
	\frac{
		\secondfundform_{\genindexi \genindexj}
		\vectorfieldx^{\genindexi}
		\vectorfieldx^{\genindexj}
	}{
		\metric_{\genindexi \genindexj}
		\vectorfieldx^{\genindexi}
		\vectorfieldx^{\genindexj}
	}
	\,,
\end{equation}
where the denominator ensures $\normalcurvature(\lambda \vectorfieldx) = \normalcurvature(\vectorfieldx)$, for all $\lambda \neq 0$.

Using the Cauchy–Schwarz inequality, for all non-null $\vectorfieldx \in \tangentbundle$, we have at any point $|\normalcurvature(\vectorfieldx)| \leq \norm{\shapeop(\unitvec{\vectorfieldx})}$, where $\unitvec{\vectorfieldx} = \vectorfieldx / \norm{\vectorfieldx}$.
Moreover, the equality holds \gls*{iff} $\vectorfieldx$ is an eigenvector of $\shapeop$.
In that case, $\normalcurvature(\vectorfieldx)$ is equal to the respective eigenvalue of $\vectorfieldx$.
Given an orthonormal basis of eigenvectors $\vectorfieldx_{\genindexi}$ of $\shapeop$, we say that $\vectorfieldx_{\genindexi}$ is the $\genindexi$-th \textit{principal direction} of $\graphf$ and its corresponding eigenvalue $\principalcurvature_{\genindexi}$ is known as the $\genindexi$-th \textit{principal curvature}~\cite{doCarmo1992}.
The \textit{minimax principle} establishes a variational connection between principal curvatures and the quotient~\eqref{eq:normal-curvature}~\cite[Section 6.10]{Kato1995}.
In particular, for any non-null $\vectorfieldx \in \tangentbundle$, $ \min_{\genindexi} \principalcurvature_{\genindexi} \leq \normalcurvature(\vectorfieldx) \leq \max_{\genindexi} \principalcurvature_{\genindexi}$.

The principal curvatures can be summarized into a single value according to several criteria, giving rise to the mean and Gaussian curvatures.

\subsubsection{Mean curvature}

At any point in $\graphf$, the mean curvature $\meancurvature$ is defined as the sum of all principal curvatures, that is, the trace of the diagonalized matrix of $\shapeop$~\cite{Ecker2004}.
Since the trace is a matrix invariant, we have $\meancurvature = \sumitod \principalcurvature_{\genindexi} = \trace{\shapeop} = \shapeop_{\genindexi}^{\genindexi}$.
Equivalently, using~\eqref{eq:extrinsic-divergence} we get~\eqref{eq:mean-curvature-normal}, while extra calculations yield~\eqref{eq:mean-curvature-gradient}:
\begin{multicols}{2}
	\noindent
	\begin{equation}
		\label{eq:mean-curvature-normal}
		\meancurvature
		=
		\extrinsicdivergence{\downwardnormalvecfield}
		\,,
	\end{equation}
	\columnbreak
	\begin{equation}
		\label{eq:mean-curvature-gradient}
		\meancurvature
		=
		\chartcomp{\rddivergence{\normalizedgradof{\pdf}}}
		\,,
	\end{equation}
\end{multicols}
\vspace*{-\multicolsep}
where $\normalizedgraddefinition$~\cite[Equation A.3]{Ecker2004}.
Both expressions convey the same idea: $\meancurvature$ is the divergence of a \textit{normalized} vector field.
The mean curvature will take negative values when $\downwardnormalvecfield$ and $\normalizedgradof{\pdf}$ locally converge towards a point.

\subsubsection{Gaussian curvature}

The shape operator contains another way of summarizing principal curvatures, this time by its determinant:
\begin{equation}
	\label{eq:gaussian-curvature}
	\gaussiancurvature
	\eqdef
	\prod_{\genindexi = 1}^{\dimension}
	\principalcurvature_{\genindexi}
	=
	\determinant{\shapeop}
	=
	\frac{\determinant{\secondfundform}}{\determinant{\metric}}
	=
	\chartcomp{\frac{\determinant{\hessian{\pdf}}}{(1 + \norm{\gradf}^2)^{1 + \frac{\dimension}{2}}}}
	\,.
\end{equation}
The Gaussian curvature $\gaussiancurvature$ owns outstanding geometrical and topological properties.
Gauss' \textit{Theorema Egregium} states that $\gaussiancurvature$ and $\absoluteval{\gaussiancurvature}$ are intrinsic invariants of a hypersurface, respectively, if $\dimension$ is even or odd~\cite[Lemma 3.1, p. 96]{Petersen1998}.
Interestingly, when $\dimension = 2$, $\gaussiancurvature$ is precisely half the \textit{scalar curvature}, the scalar contraction of the curvature tensor (see~\cite[Section 2.2.5]{Petersen1998} and~\cite[Proposition 2.1, p. 92]{Petersen1998}).

\subsubsection{Hessian matrix}

The Hessian matrix is present in the previous curvature concepts through the second fundamental form~\eqref{eq:second-fundamental-form}.
The Hessian $\hessian{\pdf}$ corresponds to the non-normalized version of $\secondfundform$, replacing $\downwardnormalvecfield$ by $\normalvecfield = (\grad{\pdf}, -1)$ in~\eqref{eq:second-fundamental-form}.
The eigenvalues of $\hessianmatrix$ are scalar multiples of those of $\secondfundformmatrix$.
Moreover, expressing the shape operator in matrix form as $\shapeopmatrix = \inversemetricmatrix \secondfundformmatrix$, and using Sylvester's law of inertia~\cite[Exercise 12 (c), p. 275]{Harville1997}, we can infer that $\shapeopmatrix$ coincides in \textit{signature} with $\secondfundformmatrix$ and, consequently, with $\hessianmatrix$.
On the other hand, it is well-known that the eigenvalues of the Hessian $\hessian{\pdf}$ relate to the concavity and convexity of $\pdf$.
Namely, if all the eigenvalues at a point are negative (alternatively, positive), i.e., $\hessian{\pdf}$ is negative (positive) definite, then $\pdf$ is locally strictly concave (convex) around that point.
Hence, the shape operator also determines local concavity and convexity because of the previous argument~\cite[p. 91]{Petersen1998}.

The shape operator and the Hessian are equal \gls*{iff} $\metricmatrix$ is the identity matrix, or, equivalently, \gls*{iff} $\norm{\grad{\pdf}} = 0$.
The connection between Gaussian curvature and the Hessian determinant is evident from~\eqref{eq:gaussian-curvature}.
Furthermore, the trace of the Hessian is equal to the divergence of the gradient, meaning $\trace{\hessian{\pdf}} = \rddivergence{\grad{\pdf}}$.
Therefore, attending to~\eqref{eq:mean-curvature-gradient}, the mean curvature differs from the Hessian trace by a normalizing term acting on the gradient argument (see~\cite[Equation 5.28]{Folland2002} for an explicit link between the two).
In conclusion, the essence of curvature lies in the Hessian $\hessian{\pdf}$.

			\section{Extra applications}

\label{sec:extra-applications}

\subsection{Univariate receiving yards in the NFL}

American football is inherently one-dimensional: only distances projected on the sidelines matter, despite the game taking place on a rectangular field.
Hence, we propose a case study about the \gls*{nfl} to illustrate our methods on univariate data.

The New England Patriots from the \gls*{nfl} experienced noticeable changes in their passing DNA while winning the \textit{Super Bowl} in 2014, 2016 and 2018.
We will analyze the bumps in their receiving profiles in those three championship seasons.
\figurename~\ref{fig:nfl} shows the original data and the estimated curvature bumps.
All three sub-figures have a primary bump between 0 and 40 yards.
In 2014, this bump is just slightly narrower than in 2016 and 2018.
What differentiates all three seasons is the existence or absence of a secondary bump.
In 2014, the second bump ranges between 80 and 100 yards and is perfectly visible.
Two years later, this bump still exists but has become smoother.
Finally, in 2018, the bump has almost been wholly \textit{ironed}.
This evolution reflects the transition from a team with a few go-to players to a more numerous receiver squad sharing responsibilities.

\renewcommand{\subfigurewidthtwoandone}{0.4\textwidth}
\renewcommand{\graphicswidthtwoandone}{0.8\textwidth}

\begin{figure}
	\centering
	\begin{subfigure}[b]{\graphicswidthtwoandone}
		\centering
		\includegraphics[width=\subfigurewidthtwoandone]{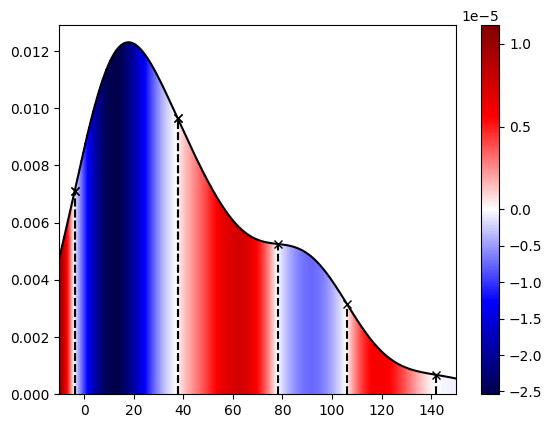}%
		\hspace{1 cm}
		\includegraphics[width=\subfigurewidthtwoandone]{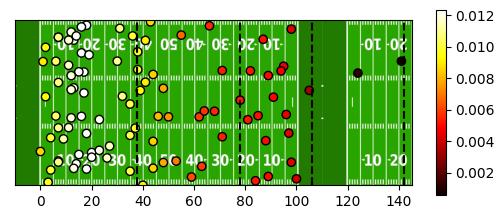}
		\caption{2014}
	\end{subfigure}
	\begin{subfigure}[b]{\graphicswidthtwoandone}
		\centering
		\includegraphics[width=\subfigurewidthtwoandone]{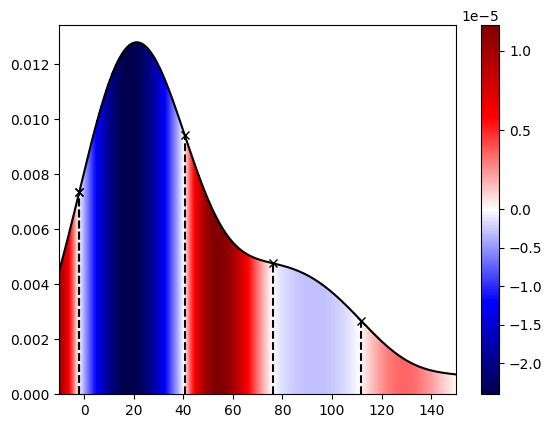}%
		\hspace{1 cm}
		\includegraphics[width=\subfigurewidthtwoandone]{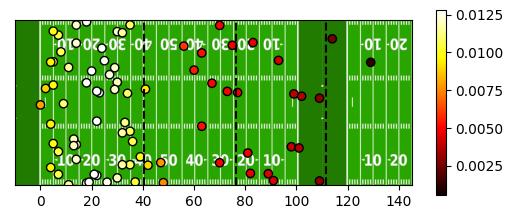}
		\caption{2016}
	\end{subfigure}
	\begin{subfigure}[b]{\graphicswidthtwoandone}
		\centering
		\includegraphics[width=\subfigurewidthtwoandone]{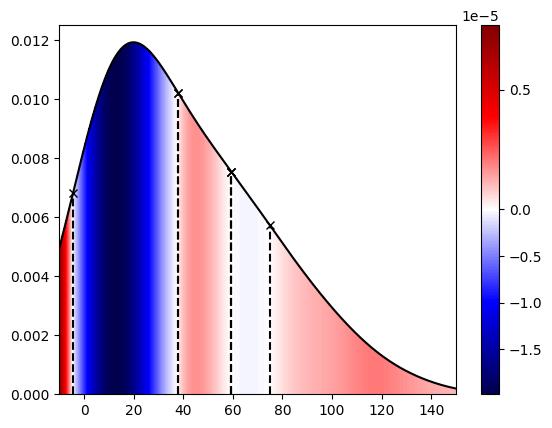}%
		\hspace{1 cm}
		\includegraphics[width=\subfigurewidthtwoandone]{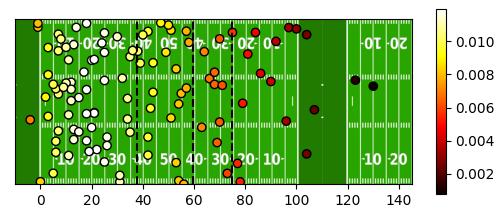}
		\caption{2018}
	\end{subfigure}
	\caption{
		New England Patriots' receiving yards in the 2014, 2016 and 2018 seasons.
		The three sub-figures have the same structure.
		On the left is the underlying \gls*{kde} \gls*{pdf} graph; on the right is a one-dimensional scatter plot of the original observations with random jitter for better resolution.
		Each point in the scatter plot corresponds to the receiving yards from Tom Brady by some player at some game.
		The number of observations is 98, 82 and 111 for the 2014, 2016 and 2018 seasons.
		The coloured areas under the \gls*{pdf}'s curve represent the sign of the second derivative: red means convex; blue, concave.
		The colour of the dots in the scatter plot conveys the value of the \gls*{kde} \gls*{pdf} at each point.
		Inflection points show up as vertical dotted lines on both sides.
	}
	\label{fig:nfl}
\end{figure}

\subsection{Trivariate pitches in MLB}

The pitching event lives at least in two dimensions.
The batter needs to anticipate the arrival coordinates of the ball.
\gls*{kde} methods have shortly become standard for these bivariate samples.
However, the pitched ball has other attributes that will make the batter's job difficult.
A first approximation suggests considering pitches as a pair of location coordinates with a \textit{speed} value measured from the release point.

We propose studying trivariate pitching samples from three top \gls*{mlb} pitchers in the 2019 season: Clayton Kershaw, Justin Verlander and Max Scherzer.
\figurename~\ref{fig:mlb-scatter} shows pitching scatter data and bumps.
In general, the three pitching patterns resemble Gaussian mixtures.
It seems that each player chooses his pitches from a finite arsenal.
Then, each pitch type has its location and speed variability.
Most pitchers employ three mechanisms: the \textit{fastball}, which relies on sheer speed; the \textit{slider}, which moves sideways at relatively high speeds; and the \textit{curveball}, which sinks at low speeds.

In \figurename~\ref{fig:scatter-kershaw}, Kershaw's fastballs and sliders have speeds ranging between 85 and 90 \gls*{mph}.
Fastballs usually range beyond 95 \gls*{mph}, as for Verlander (\figurename~\ref{fig:scatter-verlander}) and Scherzer (\figurename~\ref{fig:scatter-scherzer}), who have a defined separation between both pitch types.
As for curveballs, Kershaw's have low speed and high vertical variability.
Verlander and Scherzer's curveball clouds have similar compact shapes, but the former usually aims for the right side of the strike zone and the latter for the left.
Scherzer's usage of curveballs is the lowest of all three.

\renewcommand{\subfigurewidthtwoandone}{0.4\textwidth}
\renewcommand{\graphicswidthtwoandone}{\textwidth}

\begin{figure}
	\centering
	\begin{subfigure}[b]{\subfigurewidthtwoandone}
		\centering
		\includegraphics[width=\graphicswidthtwoandone]{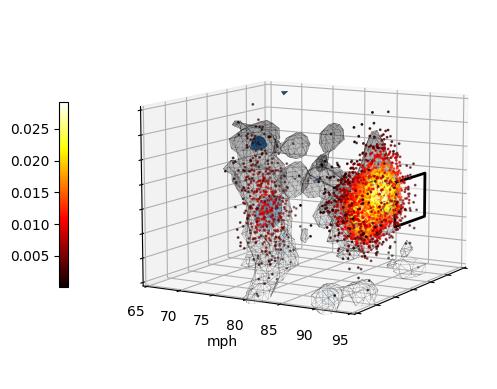}
		\caption{Clayton Kershaw}
		\label{fig:scatter-kershaw}
	\end{subfigure}
	\begin{subfigure}[b]{\subfigurewidthtwoandone}
		\centering
		\includegraphics[width=\graphicswidthtwoandone]{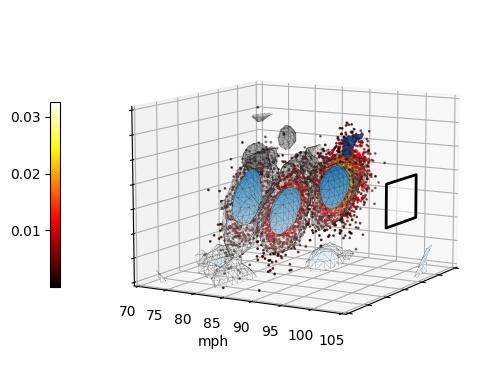}
		\caption{Justin Verlander}
		\label{fig:scatter-verlander}
	\end{subfigure}
	\begin{subfigure}[b]{\subfigurewidthtwoandone}
		\centering
		\includegraphics[width=\graphicswidthtwoandone]{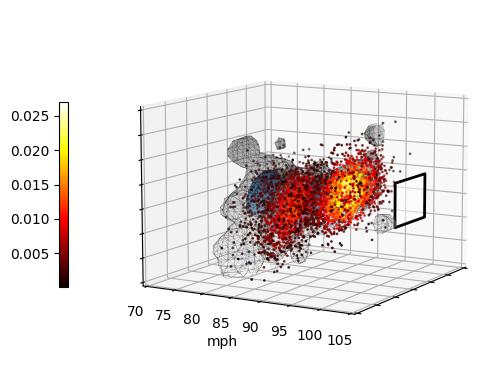}
		\caption{Max Scherzer}
		\label{fig:scatter-scherzer}
	\end{subfigure}
	\caption{
		Pitching scatter data for Clayton Kershaw, Justin Verlander and Max Scherzer with 2680, 3839 and 3284 pitches, respectively.
		Concave bumps~\eqref{eq:concave-bump} and Laplacian bumps~\eqref{eq:laplacian-bump} show in blue and grey, respectively.
		Speeds are reported in \gls*{mph}.
		The arrival coordinates at the home plate for each speed value should be interpreted relative to an \textit{averaged} strike zone, which shows at the front.
		The colour of each point corresponds to the value of the underlying \gls*{kde} \gls*{pdf}.
	}
	\label{fig:mlb-scatter}
\end{figure}

			\section{Computational details}

\label{sec:computing}

\subsection{Source code}

We provide this project's source code for reproducibility purposes.
This paper's figures can be safely generated from scratch in any environment by leveraging Docker's containerization~\cite{Merkel2014}.
See the \texttt{README.md} and \texttt{LICENSE.txt} files for further details\sourcecode{}.

\subsection{Scientific computing}

We employed Python's \textit{NumPy}~\cite{Harris2020}, \textit{SciPy}~\cite{Virtanen2020} and \textit{pandas}~\cite{McKinney2010} packages for array and data manipulation, optimization, and linear algebra operations, among other uses.

\subsection{Data}

The three main American professional sports leagues maintain vast data collections and publicly expose them through web APIs.
\gls*{nba} shot records are accessible via the \texttt{nba\_api} Python package~\cite{Patel2021}.
\gls*{nfl} passing data can be retrieved through the \texttt{nflgame-redux} Python package~\cite{Gallant2020}.
\gls*{mlb} pitching samples are available from the \texttt{pybaseball} Python package~\cite{LeDoux2021}.
To ease reproducibility, we include the downloaded data as CSV files and the parameterized scripts to call the APIs.
See the \texttt{README.md} file for more information\sourcecode{}.

\subsection{Bandwidth selection}

The bandwidth selection process targeted the second derivatives of the \gls*{pdf} in all the examples.
We employed the smoothed cross-validation criterion for the \gls*{nfl} and \gls*{nba} applications, beginning with a Gaussian pilot bandwidth~\cite{Chacon2018}; three and two stages were used, respectively.
In the \gls*{nba} case, we limited the optimization to 50 iterations.
For the \gls*{mlb} application, we selected the bandwidth assuming the true \gls*{pdf} was a Gaussian mixture~\cite{Chacon2018} with three components, following the three main pitching mechanisms discussed.
The Gaussian mixture was fit using Python's package \textit{scikit-learn}~\cite{Pedregosa2011}.

\subsection{Graphing}

Visualizing results is crucial in our proposal, as demonstrated in many examples.
Up to three dimensions, computing and picturing bump boundaries as level sets is workable.
In the univariate case, we have used conventional root-finding algorithms for $\pdfsecondderivative$.
In dimension two, we employed the graphical routine \texttt{contour} from Python's package \textit{Matplotlib}~\cite{Hunter2007}, which uses the algorithm \textit{marching squares}.
Finally, for $\dimension = 3$, we resorted to the routine \texttt{marching\_cubes} from Python's package \textit{scikit-image}~\cite{vanderWalt2014}, implementing the homonym algorithm for dimension three.

		\end{refsection}

		\renewcommand*{\mkbibcompletename}[1]{\textsc{#1}}
		\printbibliography[section=2, filter=references, title={References}, sorting=nyt]
		\printbibliography[section=2, filter=software, title={Software}, sorting=none]
	\fi
\end{document}